\documentclass[11pt]{article}

\usepackage[T1]{fontenc} 
\usepackage{alphabeta}

\usepackage{amssymb,amsmath,amsthm,hyperref}
\usepackage{graphicx}
\usepackage{todonotes}
\usepackage{fullpage}
\usepackage{latexsym,url,wrapfig}
\usepackage{enumerate}
\usepackage{cite}
\usepackage{ifthen}
\usepackage{algorithm, algorithmic,aliascnt}

\newcommand{\mynewtheorem}[2]{
	\newaliascnt{#1}{dummy}
	\newtheorem{#1}[#1]{#2}
	\aliascntresetthe{#1}
}

\theoremstyle{plain}
\mynewtheorem{theorem}{Theorem}
\mynewtheorem{proposition}{Proposition}
\mynewtheorem{corollary}{Corollary}
\mynewtheorem{lemma}{Lemma}

\theoremstyle{definition}

\newcommand{\yes}{{\sf yes}}
\newcommand{\no}{{\sf no}}
\newcommand{\bd}{{\sf bd}}
\newcommand{\ann}{{\sf ann}}
\newcommand{\inter}{{\sf int}}
\newcommand{\remove}[1]{}

\newcommand{\bigmid}{\;\big|\;}
\newcommand{\w}{\operatorname{{\bf w}}}

\newcommand{\cen}{{\sf center}}

\newcommand{\comp}{{\sf Comp}}
\newcommand{\perim}{{\sf Perim}}

\renewcommand{\Bbb}[1]{{\mathbb{#1}}}

\usepackage{color}

\definecolor{MidnightBlue}{rgb}{0.1,0.1,0.44}
\definecolor{Black}{rgb}{0,0, 0}
\definecolor{Blue}{rgb}{0, 0 ,1}
\definecolor{Red}{rgb}{1, 0 ,0}
\definecolor{White}{rgb}{1, 1, 1}
\definecolor{Grey}{rgb}{.6, .6, .6}
\definecolor{Mygreen}{rgb}{.0, .7, .0}
\definecolor{Yellow}{rgb}{.55,.55,0}
\definecolor{mustard}{rgb}{1.0, 0.86, 0.35}
\definecolor{applegreen}{rgb}{0.55, 0.71, 0.0}
\definecolor{darkturquoise}{rgb}{0.0, 0.81, 0.82}
\definecolor{celestialblue}{rgb}{0.29, 0.59, 0.82}
\definecolor{green-yellow}{rgb}{0.68, 1.0, 0.18}
\definecolor{crimsonglory}{rgb}{0.75, 0.0, 0.2}
\definecolor{darkmagenta}{rgb}{0.30, 0.0, 0.30}

\newcommand{\mnb}[1]{{\color{MidnightBlue}#1}}

\newcommand{\red}[1]{{\color{Red}#1}}

\newcounter{func}
\newcommand{\newfun}[1]{f_{\refstepcounter{func}\label{#1}\thefunc}}
\newcommand{\funref}[1]{\hyperref[#1]{f_{\ref*{#1}}}} 

\newcounter{con}
\newcommand{\newcon}[1]{c_{\refstepcounter{con}\label{#1}\thecon}}
\newcommand{\conref}[1]{\hyperref[#1]{c_{\ref*{#1}}}} 

\newtheorem{observation}{Observation}

\newcommand{\Oh}{\mathcal{O}}

\newcommand{\tw}{{\mathbf{tw}}}

\newcommand{\hh}{
\end{document}
}

\usepackage{float}
\usepackage{tikz,pgf}
\usetikzlibrary{backgrounds}
\usetikzlibrary{calc,3d}
\usetikzlibrary{intersections,decorations.pathmorphing,shapes,decorations.pathreplacing,fit,shadows,fadings}

\tikzset{black node/.style={draw, circle, fill = black, minimum size = 4pt, inner sep = 0pt}}

\usetikzlibrary{arrows.meta}

\makeatletter

\pgfdeclarearrow{
  name = ipe _linear,
  defaults = {
    length = +1bp,
    width  = +.666bp,
    line width = +0pt 1,
  },
  setup code = {
    \pgfarrowssetbackend{0pt}
    \pgfarrowssetvisualbackend{
      \pgfarrowlength\advance\pgf@x by-.5\pgfarrowlinewidth}
    \pgfarrowssetlineend{\pgfarrowlength}
    \ifpgfarrowreversed
      \pgfarrowssetlineend{\pgfarrowlength\advance\pgf@x by-.5\pgfarrowlinewidth}
    \fi
    \pgfarrowssettipend{\pgfarrowlength}
    \pgfarrowshullpoint{\pgfarrowlength}{0pt}
    \pgfarrowsupperhullpoint{0pt}{.5\pgfarrowwidth}
    \pgfarrowssavethe\pgfarrowlinewidth
    \pgfarrowssavethe\pgfarrowlength
    \pgfarrowssavethe\pgfarrowwidth
  },
  drawing code = {
    \pgfsetdash{}{+0pt}
    \ifdim\pgfarrowlinewidth=\pgflinewidth\else\pgfsetlinewidth{+\pgfarrowlinewidth}\fi
    \pgfpathmoveto{\pgfqpoint{0pt}{.5\pgfarrowwidth}}
    \pgfpathlineto{\pgfqpoint{\pgfarrowlength}{0pt}}
    \pgfpathlineto{\pgfqpoint{0pt}{-.5\pgfarrowwidth}}
    \pgfusepathqstroke
  },
  parameters = {
    \the\pgfarrowlinewidth,%
    \the\pgfarrowlength,%
    \the\pgfarrowwidth,%
  },
}

\pgfdeclarearrow{
  name = ipe _pointed,
  defaults = {
    length = +1bp,
    width  = +.666bp,
    inset  = +.2bp,
    line width = +0pt 1,
  },
  setup code = {
    \pgfarrowssetbackend{0pt}
    \pgfarrowssetvisualbackend{\pgfarrowinset}
    \pgfarrowssetlineend{\pgfarrowinset}
    \ifpgfarrowreversed
      \pgfarrowssetlineend{\pgfarrowlength}
    \fi
    \pgfarrowssettipend{\pgfarrowlength}
    \pgfarrowshullpoint{\pgfarrowlength}{0pt}
    \pgfarrowsupperhullpoint{0pt}{.5\pgfarrowwidth}
    \pgfarrowshullpoint{\pgfarrowinset}{0pt}
    \pgfarrowssavethe\pgfarrowinset
    \pgfarrowssavethe\pgfarrowlinewidth
    \pgfarrowssavethe\pgfarrowlength
    \pgfarrowssavethe\pgfarrowwidth
  },
  drawing code = {
    \pgfsetdash{}{+0pt}
    \ifdim\pgfarrowlinewidth=\pgflinewidth\else\pgfsetlinewidth{+\pgfarrowlinewidth}\fi
    \pgfpathmoveto{\pgfqpoint{\pgfarrowlength}{0pt}}
    \pgfpathlineto{\pgfqpoint{0pt}{.5\pgfarrowwidth}}
    \pgfpathlineto{\pgfqpoint{\pgfarrowinset}{0pt}}
    \pgfpathlineto{\pgfqpoint{0pt}{-.5\pgfarrowwidth}}
    \pgfpathclose
    \ifpgfarrowopen
      \pgfusepathqstroke
    \else
      \ifdim\pgfarrowlinewidth>0pt\pgfusepathqfillstroke\else\pgfusepathqfill\fi
    \fi
  },
  parameters = {
    \the\pgfarrowlinewidth,%
    \the\pgfarrowlength,%
    \the\pgfarrowwidth,%
    \the\pgfarrowinset,%
    \ifpgfarrowopen o\fi%
  },
}

\newdimen\ipeminipagewidth

\tikzstyle{ipe import} = [
  x=1bp, y=1bp,
  ipe node stretch/.store in=\ipenodestretch,
  ipe stretch normal/.style={ipe node stretch=1},
  ipe stretch normal,
  ipe node/.style={
    anchor=base west, inner sep=0, outer sep=0, scale=\ipenodestretch
  },
  ipe mark scale/.store in=\ipemarkscale,
  ipe mark tiny/.style={ipe mark scale=1.5},
  ipe mark small/.style={ipe mark scale=2},
  ipe mark normal/.style={ipe mark scale=3},
  ipe mark large/.style={ipe mark scale=5},
  ipe mark normal, 
  ipe circle/.pic={
    \draw[line width=0.2*\ipemarkscale]
      (0,0) circle[radius=0.5*\ipemarkscale];
    \coordinate () at (0,0);
  },
  ipe disk/.pic={
    \fill (0,0) circle[radius=0.6*\ipemarkscale];
    \coordinate () at (0,0);
  },
  ipe fdisk/.pic={
    \filldraw[line width=0.2*\ipemarkscale]
      (0,0) circle[radius=0.5*\ipemarkscale];
    \coordinate () at (0,0);
  },
  ipe box/.pic={
    \draw[line width=0.2*\ipemarkscale, line join=miter]
      (-.5*\ipemarkscale,-.5*\ipemarkscale) rectangle
      ( .5*\ipemarkscale, .5*\ipemarkscale);
    \coordinate () at (0,0);
  },
  ipe square/.pic={
    \fill
      (-.6*\ipemarkscale,-.6*\ipemarkscale) rectangle
      ( .6*\ipemarkscale, .6*\ipemarkscale);
    \coordinate () at (0,0);
  },
  ipe fsquare/.pic={
    \filldraw[line width=0.2*\ipemarkscale, line join=miter]
      (-.5*\ipemarkscale,-.5*\ipemarkscale) rectangle
      ( .5*\ipemarkscale, .5*\ipemarkscale);
    \coordinate () at (0,0);
  },
  ipe cross/.pic={
    \draw[line width=0.2*\ipemarkscale, line cap=butt]
      (-.5*\ipemarkscale,-.5*\ipemarkscale) --
      ( .5*\ipemarkscale, .5*\ipemarkscale)
      (-.5*\ipemarkscale, .5*\ipemarkscale) --
      ( .5*\ipemarkscale,-.5*\ipemarkscale);
    \coordinate () at (0,0);
  },
  /pgf/arrow keys/.cd,
  ipe arrow normal/.style={scale=1},
  ipe arrow tiny/.style={scale=.4},
  ipe arrow small/.style={scale=.7},
  ipe arrow large/.style={scale=1.4},
  ipe arrow normal,
  /tikz/.cd,
  ipe arrows/.style={
    ipe normal/.tip={
      ipe _pointed[length=1bp, width=.666bp, inset=0bp,
                   quick, ipe arrow normal]},
    ipe pointed/.tip={
      ipe _pointed[length=1bp, width=.666bp, inset=0.2bp,
                   quick, ipe arrow normal]},
    ipe linear/.tip={
      ipe _linear[length = 1bp, width=.666bp,
                  ipe arrow normal, quick]},
    ipe fnormal/.tip={ipe normal[fill=white]},
    ipe fpointed/.tip={ipe pointed[fill=white]},
    ipe double/.tip={ipe normal[] ipe normal},
    ipe fdouble/.tip={ipe fnormal[] ipe fnormal},
    ipe arc/.tip={ipe normal},
    ipe farc/.tip={ipe fnormal},
    ipe ptarc/.tip={ipe pointed},
    ipe fptarc/.tip={ipe fpointed},
  },
  ipe arrows, 
]

\tikzset{
  rgb color/.code args={#1=#2}{%
    \definecolor{tempcolor-#1}{rgb}{#2}%
    \tikzset{#1=tempcolor-#1}%
  },
}

\makeatother

\tikzstyle{ipe stylesheet} = [
  ipe import,
  even odd rule,
  line join=round,
  line cap=butt,
  ipe pen normal/.style={line width=0.4},
  ipe pen heavier/.style={line width=0.8},
  ipe pen fat/.style={line width=1.2},
  ipe pen ultrafat/.style={line width=2},
  ipe pen normal,
  ipe mark normal/.style={ipe mark scale=3},
  ipe mark large/.style={ipe mark scale=5},
  ipe mark small/.style={ipe mark scale=2},
  ipe mark tiny/.style={ipe mark scale=1.1},
  ipe mark normal,
  /pgf/arrow keys/.cd,
  ipe arrow normal/.style={scale=7},
  ipe arrow large/.style={scale=10},
  ipe arrow small/.style={scale=5},
  ipe arrow tiny/.style={scale=3},
  ipe arrow normal,
  /tikz/.cd,
  ipe arrows, 
  <->/.tip = ipe normal,
  ipe dash normal/.style={dash pattern=},
  ipe dash dotted/.style={dash pattern=on 1bp off 3bp},
  ipe dash dashed/.style={dash pattern=on 4bp off 4bp},
  ipe dash dash dotted/.style={dash pattern=on 4bp off 2bp on 1bp off 2bp},
  ipe dash dash dot dotted/.style={dash pattern=on 4bp off 2bp on 1bp off 2bp on 1bp off 2bp},
  ipe dash normal,
  ipe node/.append style={font=\normalsize},
  ipe stretch normal/.style={ipe node stretch=1},
  ipe stretch normal,
  ipe opacity 10/.style={opacity=0.1},
  ipe opacity 30/.style={opacity=0.3},
  ipe opacity 50/.style={opacity=0.5},
  ipe opacity 75/.style={opacity=0.75},
  ipe opacity opaque/.style={opacity=1},
  ipe opacity opaque,
]
\definecolor{red}{rgb}{1,0,0}
\definecolor{blue}{rgb}{0,0,1}
\definecolor{green}{rgb}{0,1,0}
\definecolor{yellow}{rgb}{1,1,0}
\definecolor{orange}{rgb}{1,0.647,0}
\definecolor{gold}{rgb}{1,0.843,0}
\definecolor{purple}{rgb}{0.627,0.125,0.941}
\definecolor{gray}{rgb}{0.745,0.745,0.745}
\definecolor{brown}{rgb}{0.647,0.165,0.165}
\definecolor{navy}{rgb}{0,0,0.502}
\definecolor{pink}{rgb}{1,0.753,0.796}
\definecolor{seagreen}{rgb}{0.18,0.545,0.341}
\definecolor{turquoise}{rgb}{0.251,0.878,0.816}
\definecolor{violet}{rgb}{0.933,0.51,0.933}
\definecolor{darkblue}{rgb}{0,0,0.545}
\definecolor{darkcyan}{rgb}{0,0.545,0.545}
\definecolor{darkgray}{rgb}{0.663,0.663,0.663}
\definecolor{darkgreen}{rgb}{0,0.392,0}
\definecolor{darkmagenta}{rgb}{0.545,0,0.545}
\definecolor{darkorange}{rgb}{1,0.549,0}
\definecolor{darkred}{rgb}{0.545,0,0}
\definecolor{lightblue}{rgb}{0.678,0.847,0.902}
\definecolor{lightcyan}{rgb}{0.878,1,1}
\definecolor{lightgray}{rgb}{0.827,0.827,0.827}
\definecolor{lightgreen}{rgb}{0.565,0.933,0.565}
\definecolor{lightyellow}{rgb}{1,1,0.878}
\definecolor{black}{rgb}{0,0,0}
\definecolor{white}{rgb}{1,1,1}

\hypersetup{
hyperfootnotes=false,
colorlinks = true,
linkcolor = blue!70!black,
citecolor = black!30!green
}

\usepackage{titling}
\newcommand*\samethanks[1][\value{footnote}]{\footnotemark[#1]}

\date{\empty}
\begin{document}

\title{An Algorithmic Meta-Theorem for Graph Modification to  Planarity and FOL\thanks{A conference version of this paper appeared in the \emph{Proceedings of the 28th Annual European Symposium on Algorithms (\textbf{ESA}), volume 173 of LIPICs, pages 7:1--7:23, \textbf{2020}}. The two first authors have been supported by the Research Council of Norway  via the project BWCA (314528).
The two last authors have been supported  by  the ANR projects DEMOGRAPH (ANR-16-CE40-0028) and ESIGMA (ANR-17-CE23-0010) and the French-German Collaboration ANR/DFG Project UTMA (ANR-20-CE92-0027).}}

\author{Fedor V. Fomin\thanks{Department of Informatics, University of Bergen, Norway.
Emails:  \texttt{\{fedor.fomin, petr.golovach\}@uib.no}.} \and
Petr A. Golovach\samethanks[2] \and
	Giannos Stamoulis\thanks{LIRMM, Univ Montpellier, CNRS, Montpellier, France. Emails: \texttt{giannos.stamoulis@lirmm.fr}, \texttt{sedthilk@thilikos.info}.}	\and
	Dimitrios  M. Thilikos\samethanks[3]
}
%
%
%
%

\maketitle

\begin{abstract}
\noindent In general, a {\sl graph modification problem} is defined by a graph modification operation $\boxtimes$ and a target graph property  ${\cal P}$. Typically, the modification operation $\boxtimes$ may be {\sf vertex deletion}, {\sf edge deletion}, {\sf edge contraction}, or {\sf edge addition} and the question is, given a graph $G$ and an integer $k$, whether it is possible to transform $G$ to a graph in ${\cal P}$ after applying  the operation $\boxtimes$ $k$ times on $G$.  This problem has been extensively 
studied for particular instantiations of $\boxtimes$ and ${\cal P}$. 
In this paper we consider the general property ${\cal P}_{φ}$ of being  planar and, additionally, being a model of some 
First-Order Logic  sentence $φ$ (an FOL-sentence). We call the corresponding meta-problem 
 \mnb{\sc  Graph $\boxtimes$-Modification to Planarity and $φ$} and prove the following algorithmic meta-theorem:  there exists a function $f:\Bbb{N}^{2}\to\Bbb{N}$ such that, 
for every $\boxtimes$ and every FOL sentence $φ$, the \mnb{\sc  Graph $\boxtimes$-Modification to Planarity and $φ$}
is  solvable in $f(k,|φ|)\cdot n^2$ time.
The proof constitutes a hybrid of two different classic techniques in graph algorithms. The first is the {\em  irrelevant vertex technique} that is typically used in the context of Graph Minors and deals with properties such as planarity or surface-embeddability (that are {\sl not} FOL-expressible)
 and the second is the use of {\em Gaifman's Locality Theorem}  that is the theoretical 
base for the meta-algorithmic study of FOL-expressible problems.
\end{abstract}

\noindent{\bf Keywords:} Graph Modification Problems,  Algorithmic Meta-theorems, First-Order Logic, Irrelevant Vertex Technique, Planar Graphs.


\newpage

\section{Introduction}\label{sec:intoduction}

The term \emph{algorithmic meta-theorems} was coined by Grohe in his seminal exposition in \cite{Grohe08logi} in order to describe results providing 
general conditions, typically of logical and/or combinatorial nature,
that  automatically guarantee the existence of certain types of algorithms for wide families of problems.~Algorithmic meta-theorems
 reveal  deep relations
between logic and combinatorial structures, which is a fundamental issue
of computational complexity. Such theorems not only  
yield a better understanding of the scope of general algorithmic
techniques and the limits of tractability but often provide (or induce) a variety of new algorithmic results.
The archetype of algorithmic meta-theorems is  Courcelle's theorem~\cite{Courcelle90,Courcelle92} stating that all graph
properties expressible in Monadic Second-Order Logic (in short, {\em MSOL-expressible properties}) 
are fixed-parameter tractable when parameterized by the size of the  sentence and the treewidth of the graph. 

Our meta-theorem belongs to the intersection of two algorithmic research directions: 
Deciding First-Order Logic properties on sparse graphs and graph planarization algorithms.   

\medskip\noindent\textbf{FOL-expressible properties on sparse graphs.}
For graph properties expressible in First-Order Logic (in short {\em FOL-expressible properties}),  a rich family of algorithmic meta-theorems was developed within the last decades. Each of these meta-theorems can be stated in the following form: for a graph class $\mathcal{C}$,  
 deciding FOL-expressible properties is fixed-parameter tractable on   $\mathcal{C}$, i.e. there is an algorithm running in time $f(|φ|,h_{\mathcal{C}})\cdot n^{{\cal O}(1)}$
  , where $|φ|$ is the size of the input  FOL-sentence $φ$, $h_{\mathcal{C}}$  is a constant depending on the class $\mathcal{C}$, and $n$ is the number of vertices of the input graph.    
  The starting point in the chain of such meta-theorems is the  work of Seese  \cite{Seese96line} for  $\mathcal{C}$ being the class of 
  graphs of  bounded degree  \cite{Seese96line}. The first significant extension of Seese's theorem was obtained by Frick and Grohe \cite{FrickG01deci} for the class 
  $\mathcal{C}$ of graphs of bounded local treewidth \cite{FrickG01deci}. The class of graphs of bounded local treewidth contains graphs of bounded degree, planar graphs, graphs of bounded genus, and apex-minor-free graphs.
  The next step was done by 
  Flum and Grohe \cite{FlumG01fixe}, who panelled these results up to graph classes  excluding some minor.  Dawar, Grohe, and Kreutzer  \cite{DawarGK07loca} pushed the tractability border up to  graphs locally excluding a minor. 
   Further extension was due to Dvo\v{r}\'{a}k,   Kr\'{a}l,   and Thomas, who proved tractability for the class  $\mathcal{C}$
of being locally bounded expansion \cite{DvorakKT13}. 
Finally, Grohe,  Kreutzer, and Siebertz  \cite{GroheKS17} established fixed-parameter tractability for classes that are
effectively nowhere dense. In some sense, the result of Grohe et al. is the
  culmination of  this long line of meta-theorems, because for somewhere dense  graph classes closed under taking subgraphs  
 deciding first-order properties   is unlikely to be  fixed-parameter tractable \cite{DvorakKT13,Kreutzer11algo}. 
 
 Notice that the above line of results also shed some light
on  graph modification problems.  In particular, since many modification operations   are FOL-expressible, in some situations when the target property ${\cal P}$ is FOL-expressible, the above meta-algorithmic results can be panelled to graph modification problems. As a concrete example, consider the problem of  deleting at most $k$ vertices to obtain a graph of degree at most $3$. All vertices of the input graph of degree at least $4+k$ should be deleted, so we delete them and adapt the parameter $k$ accordingly. In the remaining graph all vertices are of degree at most $3+k$ and the property of deleting at most $k$ vertices from such a graph to obtain a graph of degree at most $3$ is FOL-expressible. Hence the Seese's theorem  implies that there is an algorithm of running time $f(k)\cdot n^{\Oh(1)}$ solving this problem. 
 However these theories are not applicable  with instantiations of ${\cal P}$, like planarity,  that are not FOL-expressible. 
 
 Another island of tractability for graph modification problems is provided by  Courcelle's theorem and similar theorems on graphs of bounded widths. For example,  
 graph modification problems are  fixed-parameter tractable 
 in cases where the target property ${\cal P}$ is MSOL-expressible under the additional assumption that  the graphs in
 ${\cal P}$ have fixed treewidth (or bounded rankwidth, for MSOL$_1$-properties, see e.g.,  \cite{CourcelleO07vert}). 
 
 To conclude, according to the current state of the art, all known algorithmic meta-theorems concerning fixed-parameter tractability of  graph modification problems 
are attainable either when the target property ${\cal P}$  is FOL-expressible and the structure is sparse or when  ${\cal P}$ is MSOL/MSOL$_{1}$-expressible and the structure has bounded tree/rank-width.
Interestingly,  {\em planarity} is the typical  property that escapes the above pattern: it is {\sl not} FOL-expressible and it has  {\sl unbounded} treewidth.

\medskip\noindent\textbf{Graph planarization.} 
The  {\mnb{\sc Planar Vertex Deletion} problem
is a generalization of planarity testing. For a given graph $G$ the goal is to find a  vertex set  of   size  at most $k$ whose deletion makes the resulting graph planar. 
Planarity is a nontrivial and hereditary graph property, hence by the result of  Lewis and Yannakakis \cite{LewisY80then}, the    
decision version of  \mnb{\sc Planar Vertex Deletion}  is {\sf NP}-complete. 
The parameterized complexity of this  problem has been extensively studied.

The non-uniform fixed-parameter tractability of   \mnb{\sc Planar Vertex Deletion}  (parameterized by $k$) follows from 
 the deep result of Robertson and Seymour  in Graph Minors theory \cite{RobertsonS04GMXX}, that every minor-closed graph class can be recognized in   polynomial time.  Since the class of graphs that can be made planar by deleting at most $k$ vertices is minor-closed,  the result of Robertson and Seymour implies that  for \mnb{\sc Planar Vertex Deletion}, for each $k$, there exists a (non-uniform) algorithm  that in time $\Oh( n^3)$  solves  \mnb{\sc Planar Vertex Deletion}. 
Significant amount of  work was involved to improve the enormous constants hidden in the big-O 
and the polynomial dependence on $n$.
Marx and  Schlotter~\cite{MarxS07obta} gave an algorithm that solves the problem in time $f(k)\cdot n^2$, where $f$ is some function of $k$ only. 
Kawarabayashi   \cite{Kawarabayashi09} obtained the first linear time algorithm of running time $f(k)\cdot n$ and 
Jansen, Lokshtanov, and Saurabh~\cite{JansenLS14}  obtained an algorithm of running time $\Oh(2^{\Oh(k\log{k})}\cdot n)$.
For the related problem of contracting at most $k$ edges to obtain a planar graph, 
\mnb{\sc Planar Edge Contraction}}, an $f(k)\cdot n^{\Oh(1)}$ time algorithm was obtained by 
Golovach, van ’t Hof and  Paulusma
\cite{GolovachHP13obta}. 
Approximation algorithms for \mnb{\sc Planar Vertex Deletion} and for    \mnb{\sc Planar Edge Deletion} were studied in  
\cite{ChekuriS18,Chuzhoy11,ChuzhoyMS11}.

 \medskip\noindent\textbf{Our results.}  Let $\boxtimes$ be one of the following operations on graphs: {\sf Vertex deletion},   {\sf edge deletion}, {\sf edge contraction},  or 
 {\sf edge addition}.  We are interested whether, for a given graph $G$ and an FOL-sentence  $φ$, 
  it is possible to transform $G$ by applying at most $k$  $\boxtimes$-operations, into a planar graph with the property defined by $φ$. We refer to this problem 
 as the 
 \mnb{\sc  Graph $\boxtimes$-Modification to Planarity and $φ$} problem. 
 For example, when  $\boxtimes$ is the vertex deletion operation, then the problem is \mnb{\sc Planar Vertex Deletion}. Similarly,  \mnb{\sc  Graph $\boxtimes$-Modification to Planarity and $φ$} generalizes  \mnb{\sc Planar Edge Deletion} and {\mnb{\sc Planar Edge Contraction}}. On the other hand, for the special case of $k=0$ this is the problem of 
 deciding FOL-expressible properties on planar graphs.
 
 Examples of first-order expressible properties are deciding whether there the input graph $G$ contains a fixed graph $H$ as a subgraph
 ($H$-{\mnb{\sc Subgraph Isomorphism}}), deciding whether there is a  homomorphism from  a fixed graph $H$ to $G$ to ($H$-{\mnb{\sc Homomorphism}}), satisfying  
 degree constraints (the degree of every vertex of the graph should be between $a$ and $b$ for some constants $a$ and $b$), excluding a subgraph of constant size or having a dominating set of constant size. 
Thus \mnb{\sc  Graph $\boxtimes$-Modification to Planarity and $φ$} encompasses the variety of graph modification problems to planar graphs with specific properties. 
For example,  can we delete $k$ vertices (or edges) such that the obtained graph is planar and 
each vertex belongs to a triangle? Reversely,
can we  delete at most $k$ vertices (or edges) from a graph  such that the resulting graph is a triangle-free planar graph?  
Can we add (or contract) at most $k$ edges such that the resulting graph is  $4$-regular and  planar?  Or can we delete at most $k$ edges resulting in a square-free or claw-free planar graph?

 \smallskip
Informally, our main result can be stated as follows.

 \smallskip

\noindent\textbf{Theorem} (Informal) {\em \mnb{\sc  Graph $\boxtimes$-Modification to Planarity and $φ$} is solvable in time $f(k, φ)\cdot n^2$, for some function $f$ depending on $k$ and $φ$ only. Thus the problem is fixed-parameter tractable, when  parameterized by $k+|φ|$. }

 \smallskip

Our theorem not only implies that \mnb{\sc Planar Vertex Deletion} is fixed-parameter tractable parameterized by $k$ (proved in~\cite{JansenLS14,MarxS07obta}) and that deciding whether a  planar graph has a first-order logic property $φ$ is fixed-parameter tractable parameterized by $|φ|$ (that follows from~\cite{FrickG01deci,DawarGK07loca,DvorakKT13,GroheKS17}).  It also implies a variety of  new algorithmic results about graph modification problems to planar graphs with some specific properties 
that cannot be obtained by applying the known results directly.
Of course, for some formulas $φ$, \mnb{\sc  Graph $\boxtimes$-Modification to Planarity and $φ$} can be solved by more simple techniques. For example, if $φ$ defines a hereditary property characterized by a finite family of forbidden induced subgraphs $\mathcal{F}$, then deciding,
whether it is possible to delete at most $k$ vertices to obtain a planar $\mathcal{F}$-free graph, can be done by combining the straightforward branching algorithm and, say, the algorithm of  Jansen, Lokshtanov, and  Saurabh~\cite{JansenLS14} for \mnb{\sc Planar Vertex Deletion}.
For this, we  iteratively find a copy of each $F\in \mathcal{F}$ and if such a copy exists we branch on all the possibilities to destroy this copy of $F$ by deleting a vertex.
By this procedure, we obtain a search tree of depth at most $k$, whose leaves are all $\mathcal{F}$-free induced subgraphs of the input graph that could be obtained by at most $k$ vertex deletions. Then for each leaf, we use the planarization algorithm limited by the remaining budget. 
However, this does not work for edge modifications,  because deleting an edge in order to ensure planarity may result in creating a copy of a forbidden induced subgraph. For problems with similar features, even for very ``simple'' ones, like deleting $k$ edges to obtain a claw-free planar graph,  or planar graph without induced cycles of length $4$, 
our theorem establishes the first fixed-parameter algorithms. Also  our theorem is applicable to the situation when  $φ$ defines a hereditary property that requires an infinite family of forbidden subgraphs for its characterization and for non-hereditary properties expressible in FOL.

To our knowledge this is the first time that an algorithmic meta-theorem is able to express modification problems such as  \mnb{\sc Planar Vertex Deletion} and its variants.

The price we pay for such generality is the running time. While the polynomial factor in the running time of our algorithm is comparable with the running time of the algorithm of Marx and  Schlotter~\cite{MarxS07obta} for \mnb{\sc Planar Vertex Deletion},
it is worse than the    more advanced algorithms of  Kawarabayashi   \cite{Kawarabayashi09} and 
Jansen et al.~\cite{JansenLS14}. Similarly, the algorithms  for deciding  first-order logic properties on graph classes \cite{DvorakKT13,FrickG01deci,GroheKS17} are faster than our algorithm.  

The proof of the  main theorem is based on a non-trivial combination of the \emph{irrelevant vertex} technique 
 of Robertson and Seymour~\cite{RobertsonS86GMII,RobertsonS95GMXIII}
with the {\sl Gaifman's  Locality Theorem}~\cite{Gaifman82onlo}.  While both techniques were widely used, see  \cite{AdlerKKLST11,CyganMPP13,JansenLS14,GolovachHP13obta,GroheKMW11find,Marx10cany} and 
\cite{DawarGK07loca,FlumG01fixe,FrickG01deci}, the combination of the two techniques requires  novel ideas. 
Following the popular trend in Theoretical Computer Science,  an alternative title for our paper could be {\em ``Robertson and Seymour meet Gaifman''}.

\section{Problem definition and preliminaries}

In this section we formally define  the general \mnb{\sc  Graph $\boxtimes$-Modification to Planarity and $φ$} problem (\autoref{sadffsdgsdgfsg}),  present the theoretical background around Gaifman's Locality Theorem  (\autoref{sdagdfsgdfgdfgdfgdfgasgf}), and provide the main algorithm supporting the proof (\autoref{asdfdfsdfgdhfdsdhgfgdfg}) whose more precise description is postponed until \autoref{dsfskjsdlfkjdks}.

\subsection{Modifications on graphs.}
\label{sadffsdgsdgfsg}

We define ${\sf OP}:=\{{\sf vd}, {\sf ed}, {\sf ec}, {\sf ea}\}$, that is the set of graph operations of {\sf vertex deletion}, {\sf edge deletion}, {\sf edge contraction}, and {\sf edge addition}, respectively.
Given an operation  $\boxtimes\in{\sf OP}$, a graph $G$, and a vertex set $R\subseteq V(G)$,
we define the {\em application domain} of the operation $\boxtimes$ as

$$\boxtimes\langle G,R\rangle=
\begin{cases}
	R, & \text{if }\boxtimes={\sf vd},\\
	E(G)\cap \binom{R}{2}, & \text{if }\boxtimes={\sf ed}, {\sf ec}, \text{and}\\
	\binom{R}{2}\setminus E(G), &\text{if } \boxtimes={\sf ea}.
\end{cases}
$$
Notice that $\boxtimes\langle G,R\rangle$ is either a vertex set or a set of subsets of vertices each of size two.

Given a set $S\subseteq \boxtimes\langle G,R\rangle,$ we define $G\boxtimes S$
as the graph obtained after applying the operation $\boxtimes$ on the elements of $S$.
The vertices of $G$ that are {\em affected} by the modification of $G$ to $G\boxtimes S$, denoted by $A(S)$, 
are the vertices in $S$, in case $\boxtimes={\sf vd}$ or the endpoints of the edges of $S$, in case $\boxtimes\in \{{\sf ed}, {\sf ec}, {\sf ea}\}$. 

Given an FOL-sentence $φ$ and some $\boxtimes\in {\sf OP}$ , we define the following meta-problem:
\begin{center}
	\fbox{
		\begin{minipage}{12cm}
			\noindent\mnb{\sc  Graph $\boxtimes$-Modification to Planarity and $φ$} (In short: \mnb{\sc  G$\boxtimes$MP$φ$})\\ 
			\noindent {\bf Input:}~~A graph $G$ and  a non-negative integer $k$.\\
			{\bf Question:}~~Is there a set $S\subseteq \boxtimes\langle G,V(G)\rangle$ of size $k$ such that $G\boxtimes S$ is a planar graph
			and $G \boxtimes S\models φ$?
		\end{minipage}
	}
\end{center}

{Let $(x_{1}, \ldots, x_{\ell})\in \Bbb{N}^{\ell}$ and $f,g: \Bbb{N} \to \Bbb{N}$. We use notation $f(n)={\cal O}_{x_{1}, \ldots, x_{\ell}}(g(n))$ to denote that there exists a computable function $h:\Bbb{N}^{\ell}\to \Bbb{N}$ such that $f(n)=h(x_{1}, \ldots, x_{\ell})\cdot g(n)$.}
We are ready to give the  formal statement of the main theorem of this paper.
\begin{theorem}\label{jnvcdjkvse}
For every FOL-sentence $φ$ and for every $\boxtimes\in{\sf OP}$,   \mnb{\sc  \mnb{\sc  G$\boxtimes$MP$φ$}} is solvable in time ${\cal O}_{k,|φ|}(n^2)$.
\end{theorem}

\subsection{Gaifman's  theorem}
\label{sdagdfsgdfgdfgdfgdfgasgf}

 For vertices $u,v$ of graph $G$, we use $d_G(u,v)$ to denote the distance between $u$ and $v$ in $G$. We also use  $N_{G}^{(r)}(v)$ to denote the set of vertices of $G$ at distance at most $r$ from $v$. 

Gaifman's locality theorem is an important ingredient of our proof. We use the shortcut {\em FOL-formula/sentence} for logical formulas/sentences in First-Order Logic.
Given an FOL-formula $\psi(x)$ with one free variable $x$, we say that $\psi(x)$ is
{\em $r$-local} if the validity of $\psi(x)$  depends  only on the $r$-neighborhood of $x$, that is  for every graph $G$ and $v\in V(G)$ we have  
\[G\models \psi(v) \iff G[N_{G}^{(r)}(v)]\models \psi(v).\]
Observe that there exists an FOL-formula $\delta_{r}(x,y)$ such that for every graph $G$ and $v,u\in V(G)$, we have $d_{G}(u,v)\leq r\iff G\models \delta_{r}(v,u)$ (see~\cite[Lemma~12.26]{FlumG06para}).

We say that an FOL-sentence $φ$ is a {\em Gaifman sentence} when it is 
a Boolean combination of sentences $φ_{1}, \ldots, φ_{m}$ such that, for every $h\in[m]$, 		
\begin{eqnarray}
φ_{h}=\exists x_{1}\ldots\exists x_{\ell_{h}}\big( \bigwedge_{1\leq i<j\leq \ell_{h}} d(x_{i}, x_{j})> 2r_{h}\wedge \bigwedge_{i\in [\ell_{h}]}\psi_{h}(x_{i})\big),\label{gaifd}
\end{eqnarray}
 where $\ell_{h},r_{h}\geq 1$ and $\psi_{h}$ is an $r_{h}$-local formula with one free variable.
We refer to the variables $x_{1},\ldots,x_{\ell_{h}}$ for each $h\in[m]$ as the {\em basic variables} of $φ$.
Moreover, for every $ h\in [m]$, 
we call $φ_{h}$ a {\em basic sentence} of $φ$ and the formula $\psi_{h}$ a {\em basic local formula} of $φ$.

\begin{proposition}[Gaifman's Theorem~\cite{Gaifman82onlo}]\label{dnklsdnvkla}
	Every first-order sentence $φ$ is equivalent to a Gaifman sentence $φ'$.
	 Furthermore, 
$φ'$ can be computed effectively.
\end{proposition}

\subsection{Equivalent formulations}
\label{asdfdfsdfgdhfdsdhgfgdfg}

Given a Gaifman sentence $φ$ combined from sentences $φ_{1}, \ldots, φ_{m}$ and a unary relation symbol $R$, we define 
$\tilde{φ}$ as the sentence that is the same Boolean combination 
of sentences $\tilde{φ}_{1}, \ldots, \tilde{φ}_{m}$ such that, for every $h\in[m]$, 	
\begin{eqnarray}
\tilde{φ}_{h}=\exists x_{1}\ldots\exists x_{\ell_{h}}\big(\bigwedge_{i\in [\ell_{h}]}x_{i}\in R\wedge \bigwedge_{1\leq i<j\leq \ell_{h}} d(x_{i}, x_{j})> 2r_{h}\wedge \bigwedge_{i\in [\ell_{h}]}\psi_{h}(x_{i})\big),\label{gaidsfd}
\end{eqnarray}
where $\ell_{h},r_{h}\geq 1$ and $\psi_{h}$ is an $r_{h}$-local formula with one free variable.
Notice that $\tilde{φ}$ is evaluated on annotated graphs of the form $(G,R)$.

Let $(G,k)$ be an instance of the \mnb{\sc  \mnb{\sc  G$\boxtimes$MP$φ$}} problem.
 We may assume, because of~\autoref{dnklsdnvkla}, that $φ$ is a Gaifman  sentence. We  consider an enhanced version of the \mnb{\sc  \mnb{\sc  G$\boxtimes$MP$φ$}} problem as follows.
Let $(G,R,k)$ be a triple, where $G$ is a graph, $R\subseteq V(G)$, and $k\in \Bbb{N}$. 
We say that $(G,R,k)$ is a {\em $(φ,\boxtimes)$-triple}  if there exists set $S\subseteq\boxtimes\langle G,R\rangle$
such that $|S|\leq k$, $G\boxtimes S$ is a planar graph, and $(G\boxtimes S,R)\models \tilde{φ}$.
It is easy to observe that the property that $(G,R,k)$
is a $(φ,\boxtimes)$-triple can be expressed in MSOL.
This is easy in case $\boxtimes\in\{{\sf vd}, {\sf ed}, {\sf ec}\}$. In the case  where $\boxtimes={\sf ea}$,
we use some syntactic interpretation argument, given in~\autoref{fenflknflkas} (\autoref{kdsafksdfkalsdfkla}).

Also, we say that a set $S\subseteq\boxtimes\langle G,V(G)\rangle$ is a {\em $\boxtimes$-planarizer} of $G$ if $G\boxtimes S$ is planar.
\autoref{jnvcdjkvse}  is a consequence of  the following lemma.

\begin{lemma}\label{sdmgflsdskgfnaksdffdglsgklaamgfrlka}
	Given a 
	Gaifman sentence $φ$ and a $\boxtimes\in {\sf OP}$, there exists a function $\newfun{lksgjreklgjrnjighpotrr}:\Bbb{N}^{2}\to\Bbb{N}$, 
	and an algorithm with the following specifications:\medskip
	
	\noindent{\bf Reduce\_Instance}$(k,G,S,R)$\\
	\noindent {\sl Input:} an integer  $k\in\Bbb{N}$, a graph  $G$, a set $R\subseteq V(G)$, and a set $S\subseteq R$ that is a  ${\sf vd}$-planarizer of $G$ of size at most $k$.
	
	\noindent{\sl Output:} One of the following:
	
	\begin{enumerate}
	\item 
	\begin{itemize}
	\item if $\boxtimes\in \{{\sf ed}, {\sf ec},{\sf ea}\}$: a report that $(G,k)$ is a \no-instance of \mnb{\mnb{\sc \sc  G$\boxtimes$MP$φ$}}.
		\item if $\boxtimes={\sf vd}$: a  vertex $u\in S$ such that $S\setminus \{u\}$ is a  ${\sf vd}$-planarizer of $G\setminus u$ of size at most $k-1$ and   $(G,k)$ and $(G\setminus u, k-1)$ are equivalent instances  of 
		\mnb{\mnb{\sc \sc  G$\boxtimes$MP$φ$}}.
	\end{itemize}
		
		\item  a vertex set  $X\subseteq V(G)$
		and a vertex $v\in X$  such that $S\subseteq R\setminus X$ and $(G,R,k)$ is a $(φ,\boxtimes)$-triple if and only if $(G\setminus v,R\setminus X, k)$ is a $(φ,\boxtimes)$-triple.
				
		\item a tree decomposition of $G$ of width at most $\funref{lksgjreklgjrnjighpotrr}(k,|φ|)$.
	\end{enumerate}
	Moreover, this algorithm runs in ${\cal O}_{k,|φ|}(n)$ steps.
\end{lemma}
We postpone  the formal definitions of a tree decomposition and treewidth till  \autoref{fenflknflkas}.
Given \autoref{sdmgflsdskgfnaksdffdglsgklaamgfrlka}, we proceed to provide the proof of \autoref{jnvcdjkvse}.
Before this, we present two results that will also be used in the proof of \autoref{jnvcdjkvse}.

First, we use the algorithm of Jansen, Lokshtanov, and Saurabh~\cite{JansenLS14} for \mnb{\sc Planar Vertex Deletion}.
\begin{proposition}\label{planJLS}
There is an algorithm that, given a graph $G$ and an integer $k$, outputs, in time $2^{{\cal O}(k\log k)}\cdot n$, either a minimum-size ${\sf vd}$-planarizer $S$ of $G$ of size at most $k$, or a report that there is no ${\sf vd}$-planarizer $S$ of $G$ of size at most $k$.
\end{proposition}

Also, the following result of Golovach, van ’t Hof, and Paulusma~\cite[Lemma 1]{GolovachHP13obta} will allow us to argue about the existence of a ${\sf vr}$-planarizer of a graph $G$ of size at most $k$,
if  an {\sf ec}- or an {\sf ed}-planarizer of $G$ of size at most $k$ exists.

\begin{proposition}\label{planGHP}
If there is an {\sf ec}- or an {\sf ed}-planarizer of $G$ of size at most $k$, then there is a ${\sf vr}$-planarizer of $G$ of size at most $k$.
\end{proposition}

\begin{proof}[Proof of \autoref{jnvcdjkvse}]
Let $φ$ be an FOL-formula.
By \autoref{dnklsdnvkla}, $φ$ is equivalent to a Gaifman sentence $φ'$.
Using the planarization algorithm of~\autoref{planJLS},  
we compute,  
in $2^{{\cal O}(k\log k)}\cdot n$ steps, a ${\sf vd}$-planarizer $S$ of $G$ of size at most $k$.
If $\boxtimes={\sf ea}$, then $S:=\emptyset$,
while if $\boxtimes\in\{{\sf vd}, {\sf ed}, {\sf ec}\}$, then if such a set does not exist, we safely return a negative answer
(for the case of $\boxtimes={\sf ed}, {\sf ec}$, this is due to 
the  fact that, due to~\autoref{planGHP},
if there exists an {\sf ec}- or an {\sf ed}-planarizer of $G$  of size at most $k$ then also a  ${\sf vd}$-planarizer of $G$ of size at most $k$  exists).
We are now in position to apply recursively the algorithm {\bf Reduce\_Instance}$(k,G,S,R)$ of \autoref{sdmgflsdskgfnaksdffdglsgklaamgfrlka} until either an answer or the third case appears.
In the first case, we either return a negative answer, if $\boxtimes\in\{{\sf ed}, {\sf ec}, {\sf ea}\}$, or set $(k,G,S,R):=(k-1,G\setminus v, S\setminus \{v\}, R)$ if $\boxtimes={\sf vd}$, while in the second case we set $(k,G,S,R):=(k,G\setminus v,S,R\setminus X)$.
In the third case we have that $\tw(G)\leq \funref{lksgjreklgjrnjighpotrr}(k,|φ'|)$. Recall that the property that $(G,R,k)$ is a $(φ,\boxtimes)$-triple can be expressed in MSOL, thus the status of the final equivalent instance  $(G,R,k)$  can be evaluated  in ${\cal O}_{k,|φ|}(n)$ steps by applying Courcelle's theorem.
As the recursion takes at most $n$ steps, we obtain the claimed running time.
\end{proof}

\section{The algorithm}\label{dsfskjsdlfkjdks}
In this section, we aim to present the proof of \autoref{sdmgflsdskgfnaksdffdglsgklaamgfrlka}.
In \autoref{sadgfsfghdfjgsfdhgfgshsdgf}, we present the two main lemmata (\autoref{sdmgflsadmgfrlka} and \autoref{wedndlsakfnlkdsafnkalen}) that support the proof of \autoref{sdmgflsdskgfnaksdffdglsgklaamgfrlka} and in \autoref{asfdfgjhfgr435ertyhgs} we sketch the proof of \autoref{wedndlsakfnlkdsafnkalen}, which contains the core of the arguments of this paper.

\subsection{Two main lemmata}
\label{sadgfsfghdfjgsfdhgfgshsdgf}
We now give two lemmata, whose combination gives the proof of \autoref{sdmgflsdskgfnaksdffdglsgklaamgfrlka}.
Before we state them, we give a series of definitions. Some of them will be given on an intuitive level, while their formal versions are postponed to \autoref{fenflknflkas}.
The proofs of the two lemmata are postponed to \autoref{dnksldngklfdnglkad} and \autoref{dksdlgdsgjkloew}, respectively.

Let $\boxtimes\in{\sf OP}$, $G$ be a graph, $k\in\Bbb{N}$,
and let $S$ be a $\boxtimes$-planarizer of $G$.
We say that $S$ is an 
{\em inclusion-minimal $\boxtimes$-planarizer}  of $G$
if none of its proper subsets is a $\boxtimes$-planarizer of $G$.
Notice that, in the special case where $\boxtimes={\sf ea}$, the unique inclusion-minimal $\boxtimes$-planarizer of $G$ is the empty set of edges.
We say that  a set $Q\subseteq V(G)$ is {\em $\boxtimes$-planarization irrelevant} if for every inclusion-minimal $\boxtimes$-planarizer $S$ of $G$ that has size at most $k$, it holds that $A(S) \cap Q = \emptyset$.
We say that a graph $G$ is {\em partially disk-embedded in some closed disk $\Delta$}, 
if there is some subgraph $K$ of $G$ that is embedded in $\Delta$ whose boundary, denoted by $\bd(\Delta)$,  is a cycle of $K$  and
no vertex in the interior of $\Delta$ is adjacent to a vertex not in $\Delta$.
We use the term {\em partially $\Delta$-embedded graph $G$}
to denote that a graph $G$ is  partially disk-embedded in some closed disk $\Delta$. We also call the graph $K$
{\em compass}
of the partially $\Delta$-embedded graph $G$ and we always assume that we accompany
a partially $\Delta$-embedded graph $G$ together with an embedding of its compass in $\Delta$ that is the set $G\cap \Delta$.

The concept of  $q$-wall, where $q$ is odd, is visualized in \autoref{asfdsfdsfasdfsdfdsf}. In the same figure 
are depicted the {\em layers} (in red and blue) and
 the {\em perimeter} (the outermost layer)  of a $q$-wall (the formal definitions are postponed to \autoref{fenflknflkas}).
 Also the {\em branch vertices} are depicted in yellow.
 \begin{figure}[ht]
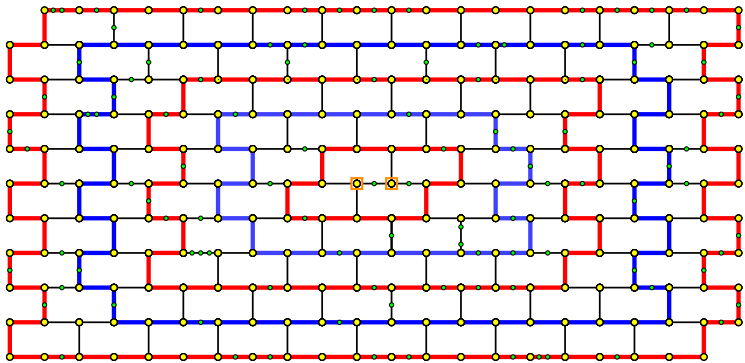

	\centering
\resizebox{10cm}{!}{

}
	\caption{An $11$-wall and its 5 layers.}
	\label{asfdsfdsfasdfsdfdsf}
\end{figure}  
Let $W$ be a wall of a graph $G$. We use $\perim(W)$ to denote the perimeter of $W$.
The two branch vertices of $W$ that do not belong to any layer and are connected by a path that does not intersect any layer are called the {\em central vertices} of $W$  (depicted by two orange squared vertices in \autoref{asfdsfdsfasdfsdfdsf}).
We denote the central vertices of $W$ by $\cen(W)$.
Let  $K'$ be the connected component of $G\setminus \perim(W)$ that contains $W\setminus \perim(W)$.
The {\em compass} of $W$, denoted by $\comp(W)$, is the graph $G[V(K')\cup V(\perim(W))]$. Observe that $W$ is a subgraph of $\comp(W)$ and $\comp(W)$ is connected.
In what follows we will always consider walls that are drawn 
inside the disk of a partially $\Delta$-embedded graph. Therefore, we can see the compass of $W$ as the part 
of the graph that is drawn inside the closed disk boundary the perimeter of $W$.
We are now in position to state the following two lemmata.

\begin{lemma}\label{sdmgflsadmgfrlka}
Given a Gaifman sentence $φ$ and a $\boxtimes\in {\sf OP}$, there exist two
functions $\funref{lksgjreklgjrnjighpotrr},\newfun{jsdfjbnvjfdak}:\Bbb{N}^{2}\to\Bbb{N}$, 
and an algorithm with the following specifications:\smallskip
	
	\noindent{\bf Find\_Area}$(k,q,G,S)$\\
	\noindent {\sl Input:} a $k\in\Bbb{N}$, an odd $q\in\Bbb{N}_{\geq 1}$, a graph  $G$, and a set $S\subseteq V(G)$ that is a ${\sf vd}$-planarizer of $G$ of size at most $k$.
	
	\noindent{\sl Output:} One of the following:
		
	\begin{enumerate}
		\item  
	\begin{itemize}
	\item if $\boxtimes\in \{{\sf ed}, {\sf ec},{\sf ea}\}$: a report that $(G,k)$ is a \no-instance of \mnb{\mnb{\sc \sc  G$\boxtimes$MP$φ$}}.
		\item if $\boxtimes={\sf vd}$: a  vertex $u\in S$ such that $S\setminus \{u\}$ is a  ${\sf vd}$-planarizer of $G\setminus u$ of size at most $k-1$ and   $(G,k)$ and $(G\setminus u, k-1)$ are equivalent instances  of 
		\mnb{\mnb{\sc \sc  G$\boxtimes$MP$φ$}}.
	\end{itemize}	
		
		\item  a $q$-wall $W$ of $G$ and a closed disk $\Delta$ such that
		\begin{itemize}
			\item the compass of $W$ has treewidth at most $\funref{jsdfjbnvjfdak}(k,q)$,
			\item $G$ is partially $\Delta$-embedded, where $G\cap \Delta=\comp(W)$, $\bd(\Delta)=\perim(W)$,
			\item $V(\comp(W))$ is $\boxtimes$-planarization irrelevant, and
			\item $N_{G}(S)\cap V(\comp(W))=\emptyset$, or
		\end{itemize}

		\item a tree decomposition of $G$ of width at most $\funref{lksgjreklgjrnjighpotrr}(k,q)$.
	\end{enumerate}
	
\noindent Moreover, this algorithm runs in ${\cal O}_{k,q}(n)$ steps.
\end{lemma}
By $N_G(S)$ we denote the vertices in $G\setminus S$ that are adjacent, in $G$, with vertices in $S$.
In the first possible output of the algorithm of \autoref{sdmgflsadmgfrlka} we have either a negative answer to the \mnb{\mnb{\sc \sc  G$\boxtimes$MP$φ$}} problem
or an equivalent instance of   \mnb{\mnb{\sc \sc  G$\boxtimes$MP$φ$}}  with reduced value of $k$.

The proof of \autoref{sdmgflsadmgfrlka}
is in~\autoref{dnksldngklfdnglkad} and its main steps are the following.
In case, $\boxtimes={\sf ea}$ we first check whether $G$ 
is planar. If not, we report a negative answer, otherwise
we find a wall $W$ in $G$ whose size is a ``big-enough'' function of $k$ and whose compass has ``small-enough'' treewidth using~\cite[Lemma 4.2]{GolovachKMT17thep}. 
This wall contains
an (also ``big-enough'') subwall of $W$ whose compass is not affected by $S$. In case $\boxtimes=\{{\sf vd},{\sf ed},{\sf ec}\}$, we consider the neighbors of $S$ in the planar graph $G'$, this is the set $N_{G}(S)$. Moreover, we consider a ``big-enough'' triangulated grid $Γ$ as a contraction of 
$G'$ (using~\cite[Theorem 3]{FominGT11cont}) and the set $N_{Γ}$ 
of the ``contraction-heirs'' of the vertices of $N_{G}(S)$  in $Γ$.
If $|N_{Γ}|$ is ``big-enough'', then we prove, using the main technical result of~\cite{DemaineFHT04bidi},  that 
some of the vertices of $S$ should be affected by every possible solution, in case $\boxtimes={\sf vd}$, or that 
we have a \no-instance, in case $\boxtimes\in \{{\sf ed},{\sf ec}\}$. If  $|N_{Γ}|$ is ``small-enough'', then we can find a ``big-enough'' wall $W$ in $G$ whose compass is not affected by $S$ (again using the previously mentioned  result of \cite{GolovachKMT17thep}). The proof is completed by proving  that 
this wall contains some ``big-enough'' subwall that is not affected by any inclusion-minimal $\boxtimes$-planarizer.

The next lemma deals with the second possible output of the algorithm of \autoref{sdmgflsadmgfrlka} and contains the ``core arguments'' of this paper.

\begin{lemma}\label{wedndlsakfnlkdsafnkalen}
Given a Gaifman sentence $φ$ and a $\boxtimes\in{\sf OP}$, there exist a function
$\newfun{sngklargnklrangl}:\Bbb{N}^{2}\to \Bbb{N}$, whose images are odd integers, and an algorithm with the following specifications:\smallskip
	
	\noindent{\bf Find\_Vertex}$(k,\Delta,G,R,\tilde{W})$\\
	\noindent {\sl Input:}  a $k\in \Bbb{N}$, a partially $\Delta$-embedded graph $G$, a set of (annotated) vertices $R\subseteq V(G)$, and a $q$-wall $\tilde{W}$ of $G$ such that 
	\begin{itemize}
	\item $q = \funref{sngklargnklrangl}(k,|φ|)$,

		\item the compass of $\tilde{W}$ has treewidth at most $\funref{jsdfjbnvjfdak}(k,q)$ (where $\funref{jsdfjbnvjfdak}$ is the function of~\autoref{sdmgflsadmgfrlka}),
		\item $G\cap \Delta=\comp(\tilde{W})$, $\bd(\Delta)=\perim(\tilde{W})$,
		\item $V(\comp(\tilde{W}))$ is $\boxtimes$-planarization irrelevant, and 
			\end{itemize}
	\noindent{\sl Output:} a vertex set  $X\subsetneq V(\comp(\tilde{W}))$
	and a vertex $v\in X$  such that $(G,R,k)$ is a $(φ,\boxtimes)$-triple if and only if $(G\setminus v,R\setminus X, k)$ is a $(φ,\boxtimes)$-triple.
	\smallskip
	
	\noindent{Moreover, this algorithm runs in ${\cal O}_{k,|φ|}(n)$ steps.}
\end{lemma}

Notice that the above algorithm produces 
a $(φ,\boxtimes)$-triple where both $R$ and $G$ are reduced.  
Given~\autoref{sdmgflsadmgfrlka} and~\autoref{wedndlsakfnlkdsafnkalen}, we proceed to prove~ \autoref{sdmgflsdskgfnaksdffdglsgklaamgfrlka}.

\begin{proof}[Proof of~\autoref{sdmgflsdskgfnaksdffdglsgklaamgfrlka}.]
We describe the algorithm {\bf Reduce\_Instance} for input $(k,G,S,R)$.
First, we call the algorithm {\bf Find\_Area} of~\autoref{sdmgflsadmgfrlka} for input $(k,q,G,S)$ which returns one of the following:
\begin{enumerate}
		\item  
	\begin{itemize}
	\item if $\boxtimes\in \{{\sf ed}, {\sf ec},{\sf ea}\}$: a report that $(G,k)$ is a \no-instance of \mnb{\mnb{\sc \sc  G$\boxtimes$MP$φ$}}.
		\item if $\boxtimes={\sf vd}$: a  vertex $u\in S$ such that $S\setminus \{u\}$ is a  ${\sf vd}$-planarizer of $G\setminus u$ of size at most $k-1$ and   $(G,k)$ and $(G\setminus u, k-1)$ are equivalent instances  of 
		\mnb{\mnb{\sc \sc  G$\boxtimes$MP$φ$}}.
	\end{itemize}	
		
		\item  a $q$-wall $W$ of $G$ and a closed disk $\Delta$ such that
		\begin{itemize}
			\item the compass of $W$ has treewidth at most $\funref{jsdfjbnvjfdak}(k,q)$,
			\item $G$ is partially $\Delta$-embedded, where $G\cap \Delta=\comp(W)$, $\bd(\Delta)=\perim(W)$,
			\item $V(\comp(W))$ is $\boxtimes$-planarization irrelevant, and
			\item $N_{G}(S)\cap V(\comp(W))=\emptyset$, or
		\end{itemize}

		\item a tree decomposition of $G$ of width at most $\funref{lksgjreklgjrnjighpotrr}(k,q)$.
	\end{enumerate}
If {\bf Find\_Area}$(k,q,G,S)$ returns either the first or the third possible output, then our algorithm terminates by returning the corresponding output.
In the second possible output,
we call the algorithm {\bf Find\_Vertex} of \autoref{wedndlsakfnlkdsafnkalen} for input $(k,\Delta,G,R,W)$, which outputs a  vertex set  $X\subsetneq V(\comp({W}))$
and a vertex $v\in X$  such that  $(G,R,k)$ is a $(φ,\boxtimes)$-triple if and only if $(G\setminus v,R\setminus X, k)$ is a $(φ,\boxtimes)$-triple.
Observe that since $N_{G}(S)\cap V(\comp({W}))=\emptyset$, then $S\subseteq R\setminus X$.
We insist that while in the output of {\bf Find\_Area} we demand that $N_{G}(S)\cap V(\comp({W}))=\emptyset$, this is used  only to guarantee that $S\subseteq R\setminus X$.
For the overall running time of our algorithm,
recall that the two algorithms of \autoref{sdmgflsadmgfrlka} and \autoref{wedndlsakfnlkdsafnkalen} run in ${\cal O}_{k,|φ|}(n)$ steps.
\end{proof}
	
\subsection{Sketch of the proof of \autoref{wedndlsakfnlkdsafnkalen}}
\label{asfdfgjhfgr435ertyhgs}

In order to prove \autoref{wedndlsakfnlkdsafnkalen},
we first find a ``large-enough'' collection ${\cal W}$ of subwalls of $\tilde{W}$ each with $\rho$
layers (where $\rho$ is ``big-enough''),
whose compasses are pairwise vertex-disjoint.
We keep in mind that every wall in ${\cal W}$ has height $2\rho+1$ and $\rho$ layers.

The key idea is to define a ``characteristic''
of each wall $W\in{\cal W}$ that encodes 
all possible ways that a $\boxtimes$-planarizer $S$ of $G$ affects $\comp(W)$ along with the different ways a vertex assignment to the basic variables of the Gaifman formula $φ$ in $\comp(W)$ can certify $G\boxtimes S\modelsφ$.
Recall that 
$\tilde{φ}$  is  a  Boolean combination 
of sentences $\tilde{φ}_{1}, \ldots, \tilde{φ}_{m}$ so that for every $h\in[m]$, 	
\begin{eqnarray*}
\tilde{φ}_{h}=\exists x_{1}\ldots\exists x_{\ell_{h}}\big(\bigwedge_{i\in [\ell_{h}]}x_{i}\in R\wedge \bigwedge_{1\leq i<j\leq \ell_{h}} d(x_{i}, x_{j})> 2r_{h}\wedge \bigwedge_{i\in [\ell_{h}]}\psi_{h}(x_{i})\big),\label{gaidsfd2}
\end{eqnarray*}
where $\ell_{h},r_{h}\geq 1$ and $\psi_{h}$ is an $r_{h}$-local formula with one free variable and that $\tilde{φ}$ is evaluated on annotated graphs of the form $(G,R)$.
Clearly, $\tilde{φ}$ is a sentence in Monadic Second Order Logic, in short, an MSOL-sentence. 
We set $r:=\max_{h\in[m]}\{r_{h}\}, \ell:=\sum_{h\in[m]}\ell_{h}$, and $d:=2(r+(\ell+1)r+r)$.

As a first step, let 
${\sf SIG}= 2^{[\ell_1]}\times\cdots\times 2^{[\ell_{m}]}\times[\rho]$.
Also, for every wall $W\in {\cal W}$, let $K:=\comp(W)$, for every $t\in [\rho]$, let $K^{(t)}:= \comp(W^{(2t+1)})$ and $P^{(t)}:=V(\perim(W^{(2t+1)}))$. Here, by $W^{(t)}$ we denote the subwall of $W$
that has height $t$, whose layers are the innermost $\frac{t-1}{2}$ layers of $W$, and which has the same center as $W$.
We set  ${\bf K} = (V(K^{(1)}),\ldots, V(K^{(\rho)}))$. We call the tuple $\mathfrak{K}_W = (K,{\bf K})$ the {\em panelled compass} of the wall $W$ in $G$.
Given the panelled compass $\mathfrak{K}_W$ of a wall $W\in {\cal W}$ in $G$, a set $R\subseteq V(\comp(W))$, an integer $z\in [d,\rho]$, and a set $S\subseteq \boxtimes\langle K,R\rangle$ such that $A(S)\subseteq V(K^{(z-d+1)})\cap R$, we define

\begin{eqnarray*}
{{\sf sig}}_{φ,\boxtimes}(\mathfrak{K}_W, R, z,S) & =&  \{(Y_{1},\ldots,Y_{m},t)\in {\sf SIG}  \mid  t\leq z \mbox{~and~} \exists \ (\tilde{X}_{1},\ldots,\tilde{X}_m)\mbox{~such that~} \forall h\in[m]\  \\
& &~~~~~~~~~~~~~~~~~~~~~~~~~~~~~~~~~~~~~~~~ \tilde{X}_{h}=\{x_{i}^{h}\mid i\in Y_{h}\},\\
& &~~~~~~~~~~~~~~~~~~~~~~~~~~~~~~~~~~~~~~~~\tilde{X}_{h}\subseteq V((K^{(t-r+1)}\boxtimes S)\setminus P^{(t-r+1)})\cap R,\\
 & &~~~~~~~~~~~~~~~~~~~~~~~~~~~~~~~~~~~~~~~~\tilde{X}_h \text{~is~} (|Y_{h}|, r_{h})\text{-scattered in }K^{(t)}\boxtimes S,\text{ and }\\
& &~~~~~~~~~~~~~~~~~~~~~~~~~~~~~~~~~~~~~~~~K^{(t)}\boxtimes S\models \bigwedge_{x\in \tilde{X}_h}\psi_{h}(x)\}.\\
\end{eqnarray*}

In the above definition, a set $X$ of vertices is {\em $(\alpha,\beta)$-scattered},
if  $|X|=\alpha$ and there are no two vertices in $X$ 
within distance $\leq 2β$.
Intuitively, $(Y_{1},\ldots,Y_{m},t)\in {{\sf sig}}_{φ,\boxtimes}(\mathfrak{K}_W, R, z,S)$ if the application of 
the operation $\boxtimes$ on $G$ as defined by $S$
gives rise to the existence of a collection of scattered sets $(\tilde{X}_1, \ldots, \tilde{X}_m)$
in $(K^{(t-r+1)}\boxtimes S)\setminus P^{(t-r+1)}$ (one scattered set for each basic sentence $φ_h$) so that when the vertices of $\tilde{X}_h$  are assigned to the basic variables of $φ_{h}$  corresponding to $Y_{h}$,  the local basic formula 
$\psi_{h}$ is satisfied for each $x\in \tilde{X}_h$ in the modified graph.
Let us elaborate more on the properties that
the sets $(\tilde{X}_1, \ldots, \tilde{X}_m)$ are asked to satisfy.
First, we ask that, for every $h\in[m]$, the set $\tilde{X}_{h}$ is 
$(|Y_h|,r_h)$-scattered in $K^{(t)}\boxtimes S$ and is a subset of 
$V((K^{(t-r+1)}\boxtimes S)\setminus P^{(t-r+1)})$.
Therefore,
for each $h\in[m]$ and each vertex $x\in \tilde{X}_h$, every
vertex of $G\boxtimes S$ of distance at most $r$ from $x$ is in $V(K^{(t)}\boxtimes S)$.
This implies that the satisfaction of the local basic formula 
$\psi_{h}$ for each $x\in \tilde{X}_h$ can be checked in the graph
$K^{(t)}\boxtimes S$.
Also, notice that $(Y_{1},\ldots,Y_{m},t)\in {{\sf sig}}_{φ,\boxtimes}(\mathfrak{K}_W, R, z,S)$ only if $r\leq t$.
Given that $t\leq z$, we have that $V(K^{(t-r+1)}\boxtimes S)\subseteq V(K^{(t)}\boxtimes S)\subseteq V(K^{(z)}\boxtimes S)$ and therefore 
for every $h\in [m]$, $\tilde{X}_{h}\subseteq V(K^{(z)}\boxtimes S)$.

It is now time to define the characteristic of a wall $W\in{\cal W}$.
Given the panelled compass $\mathfrak{K}_W$ of a wall $W\in {\cal W}$ in $G$ and a set $R\subseteq V(\comp(W))$, we define the {\em $(φ,\boxtimes)$-characteristic} of $(\mathfrak{K}_W, R)$
as follows
\begin{eqnarray*}
\text{\sf  $(φ,\boxtimes)$-char}(\mathfrak{K}_W,R) & =& \{(z,{\sf sig},s)\in  [d,\rho]\times 2^{{\sf SIG}}\times [0,k]\mid \exists S\subseteq\boxtimes\langle K,R\rangle\mbox{~such that},\\
& &\hspace{7cm} A(S)\subseteq V(K^{(z-d+1)})\cap R,\\ 
& &\hspace{7cm}  |S|=s, K\boxtimes S \text{~is planar, and}   \\
& &\hspace{7cm}  {{\sf sig}}_{φ,\boxtimes}(\mathfrak{K},R,z,S)={\sf sig}\}.
\end{eqnarray*}

Notice that all queries in the definition of 
$(φ,\boxtimes)\text{\sf -char}(\mathfrak{K}_W, R)$ can be expressed in MSOL.
Indeed, this is easy to see when 
$\boxtimes\in \{{\sf vd}, {\sf ed}, {\sf ec}\}$, as in this case
the query ``$\comp(W)\boxtimes S$ is planar'' is trivially true, since 
$V(\comp(\tilde{W}))$ is $\boxtimes$-planarization irrelevant.
In the case where $\boxtimes={\sf ea}$, the MSOL expressibility
is proved in \autoref{fenflknflkas} (\autoref{kdsafksdfkalsdfkla}).
As each $W\in {\cal W}$ has treewidth bounded by a function of $k$ and $|φ|$, it follows by the 
theorem of Courcelle that $(φ,\boxtimes)\text{-char}(\mathfrak{K}_W,R)$ can be computed 
in ${\cal O}_{k,|φ|}(n)$ time.

For every wall $W_i\in {\cal W}$, we set $K_i : = \comp(W_i)$, for every $j\in [\rho]$, $K_i^{(j)}: = \comp(W_i^{(2j+1)})$ and $P_i^{(j)}: = V(\perim(W_i^{(2j+1)}))$, $\mathfrak{K}_i: = \mathfrak{K}_{W_i}$ and $R_i := R\cap V(\comp(W_i))$.
We say that two walls $W_1, W_2$ are {\em$(φ,\boxtimes)$-equivalent} if $(\mathfrak{K}_{1}, R_1)$ and $(\mathfrak{K}_{2},R_2)$ have the same $(φ,\boxtimes)$-characteristic.
Since the collection ${\cal W}$ contains ``many-enough'' walls, we can find a, still ``large-enough'', collection ${\cal W}'\subseteq {\cal W}$ of walls that are pairwise equivalent.
We fix a wall $W_{1}\in {\cal W}'$ and we set $X:=V(\comp(W_{1}^{(r)}))$, where $r=\max_{h\in[m]}\{r_{h}\}$, and $v\in\cen(W_{1})$.

In what follows, we highlight the ideas of the proof of the fact that if $(G,R,k)$ is a $(φ,\boxtimes)$-triple, then $(G\setminus v,R\setminus X, k)$ is a $(φ,\boxtimes)$-triple.
We first consider a set $S\subseteq \boxtimes\langle G,R\rangle$ of size at most $k$ that certifies that $(G,R,k)$ is a $(φ,\boxtimes)$-triple.
Then, we pick a wall $W_{2}\in{\cal W}'\setminus \{W_{1}\}$ whose compass is not affected by $S$.
We are allowed to pick this wall since there are ``many-enough'' walls equivalent to $W_{1}$ in ${\cal W}'$.
Our strategy is to use the fact that $W_{1}$ and $W_{2}$ are $(φ,\boxtimes)$-equivalent in order to state a ``replacement argument'':
we can find a $z\in[\rho]$, such that the subset $S_{\sf in}$ of $S$ that affects $K_{1}^{(z)}$ and the set $X$ of vertices of $K_{1}^{(z)}$ that are assigned to the basic variables of $φ$ in order to certify that $G\boxtimes S\models φ$, can be replaced by their ``equivalent'' sets $\tilde{S}$ and $\tilde{X}$ in $K_{2}^{(z)}$.
As a consequence of this, for every possible solution $S$ and vertex assignment to the basic variables of $φ$, we can find both a new solution and a new vertex assignment that ``avoid'' the ``inner part'' of $W_{1}$.
This implies that the validity of any basic local formula of $φ$ does not depend on the central vertices of $W_{1}$.
Thus, we can declare one of them ``irrelevant'' and safely remove it from $G$, while storing  (by reducing $R$ to $R\setminus X$) the fact that  every possible solution $S$ and vertex assignment to the basic variables of $φ$ can ``avoid'' the ``inner part'' of $W_{1}$.

To further inspect how this ``replacement'' is achieved, we need to dive deeper into  
the technicalities of the proof (through an intuitive perspective). Given a wall $W$,
 we refer to
a {\em wall-annulus} of $W$ as the subgraph of $W$ that is obtained from $W$ after removing from $W$ all its layers, except a fixed number of consecutive layers. We think of every wall $W\in{\cal W}$ as divided in 
consecutive wall-annuli of fixed size.
Since $\rho$ is ``big-enough'', then we can find also ``many enough'' such wall-annuli. We denote each one of them by $A_{i}(W)$.
Given a $W\in{\cal W}$, every wall-annulus $A_{i}(W)$ is divided in some regions as
depicted in \autoref{ndajkgndskngdasklgnasd}.
\begin{figure}[ht]
	\centering
	\includegraphics[width=13cm]{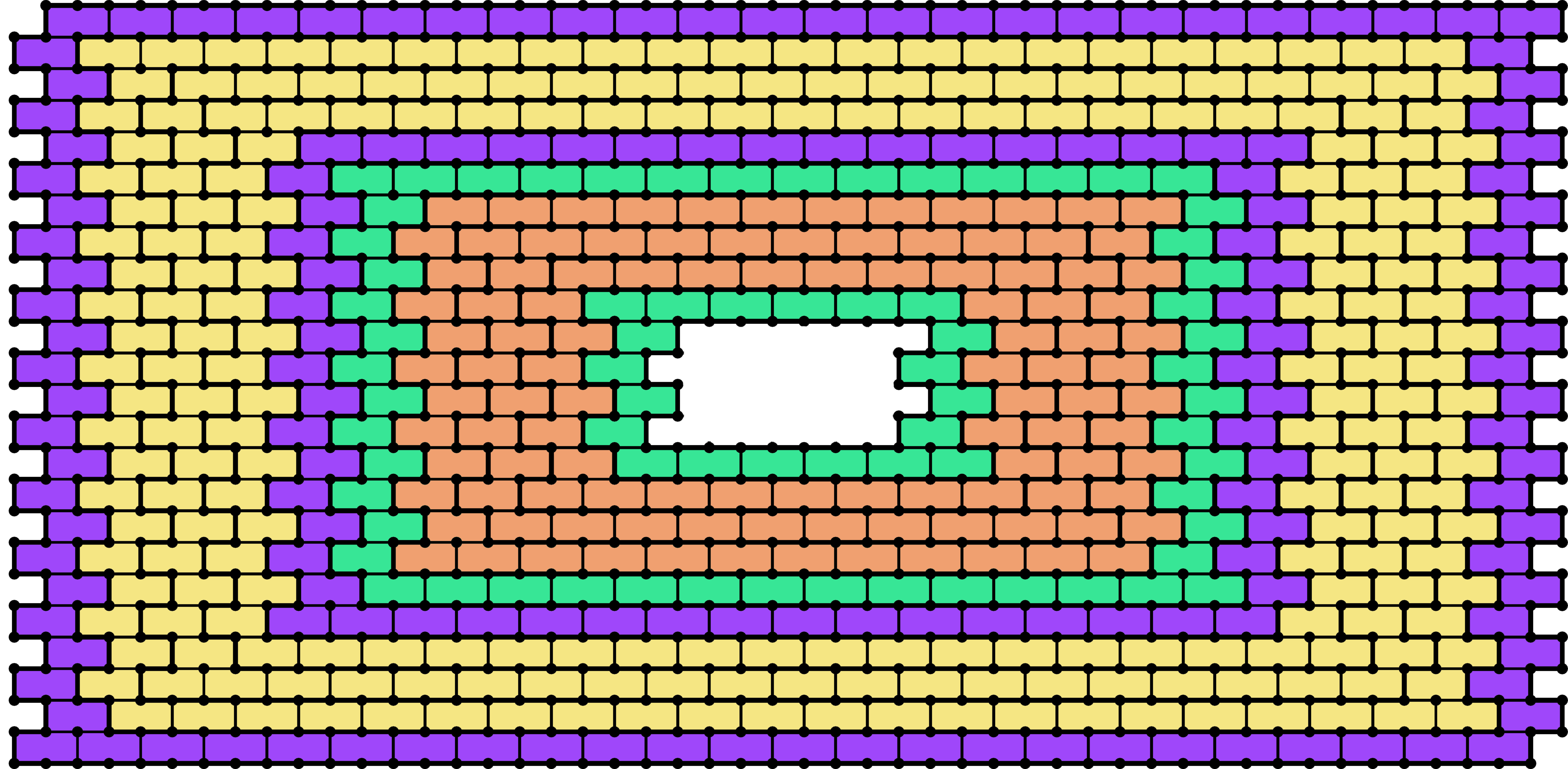}
	\caption{An example of a wall-annulus $A_{i}(W)$ of a wall $W\in\mathcal{W}$, together with its regions referred in the proof of \autoref{wedndlsakfnlkdsafnkalen}.}
	\label{ndajkgndskngdasklgnasd}
\end{figure}
The regions depicted in purple and green are consisting of $r$ layers of the wall $W$ (recall that $r=\max_{h\in[m]}\{r_{h}\}$).
The regions depicted in yellow and orange are both ``big-enough'' so as to be able to find, in each one of them, an also ``big-enough'' wall-annulus that ``avoids'' a given vertex assignment to the basic variables of $φ$.

Since $\rho$ is ``big-enough'', then we can find a wall-annulus $A_{i}(W_{1})$ that is not affected by $S$.
This allows us to partition $S$ in two sets, $S_{\sf in}$ and $S_{\sf out}$ in the obvious way.
The fact that $W_{1}$ and $W_{2}$ are $(φ,\boxtimes)$-equivalent implies the existence of a set $\tilde{S}$ in $W_{2}$ certifying that these two walls have the same characteristic.
Thus, by setting $S':=\tilde{S}\cup S_{\sf out}$, we have that $S'\subseteq\boxtimes\langle G,R'\rangle$, $|S'|=|S|$, and $G\boxtimes S'$ is planar.
The latter is guaranteed by the fact that $V(\comp(\tilde{W}))$ is $\boxtimes$-planarization irrelevant, in the case $\boxtimes\in\{{\sf vd},{\sf ed}, {\sf ec}\}$, while in the case that $\boxtimes={\sf ea}$, the existence of the outer purple buffer of  $A_{i}(W_{1})$ (resp. $A_{i}(W_{2})$) allows us to treat $S_{\sf in}$ (resp. $\tilde{S}$) and $S_{\sf out}$ separately, while not spoiling planarity.
The last part of the proof requires to prove that $(G\boxtimes S,R)\models \tilde{φ}\iff (G\boxtimes S',R')\models \tilde{φ}$.

For simplicity, here we only argue why $(G\boxtimes S,R)\models \tilde{φ}_{h}\implies (G\boxtimes S',R')\models \tilde{φ}_{h}$ holds, as the arguments in the proof of the inverse direction are completely symmetrical.
Therefore, given an $(\ell_{h}, r_{h})$-scattered set $X$ such that $φ_{h}$ is satisfied if the vertices of $X$ are assigned to the basic variables of $φ_{h}$, we aim to find a $t\in[\rho]$ in order to ``replace'' the vertices in $X\cap V(K_{1}^{(t)})$ with a set $\tilde{X}$ of vertices in $K_{2}^{(t)}$ such that the resulting vertex set $X^\star$ is $(\ell_{h}, r_{h})$-scattered and  $φ_{h}$ is satisfied if the vertices of $X^\star$ are assigned to the basic variables of $φ_{h}$.
Notice that for every $h\in[m]$ such that $(G\boxtimes S,R)\models \tilde{φ}_{h}$, these ``replacement arguments'' are pairwise independent.

We first deal with the possibility that the given scattered set $X$ intersects some ``inner part'' of $\comp(W_{2})$.
Thus, in order to ``clean'' the ``inner part'' of $\comp(W_{2})$, we find a wall $W_{3}\in{\cal W}'\setminus\{W_{1}, W_{2}\}$ that ``avoids'' both $S$ and $X$ (for different $h\in[m]$, the choice of $W_{3}$ may coincide).
\begin{figure}[ht]
	\centering
	{\resizebox{7cm}{!}{
	\begin{tikzpicture}[ipe stylesheet]
  \filldraw[shift={(32, 816)}, xscale=0.7896, yscale=0.9688, rgb color={fill=0.588 0.176 1}]
    (0, 0) rectangle (224, -128);
  \filldraw[shift={(35.892, 811.456)}, xscale=0.8077, yscale=1.0309, rgb color={fill=1 0.918 0.443}]
    (0, 0) rectangle (208, -112);
  \filldraw[shift={(68, 788)}, xscale=0.7222, yscale=0.85, rgb color={fill=0.588 0.176 1}]
    (0, 0) rectangle (144, -80);
  \filldraw[line cap=round, rgb color={fill=0.255 0.796 0.322}]
    (116, 760)
     -- (116, 760);
  \draw[shift={(188, 711.998)}, yscale=1.1053, gold, ipe pen ultrafat]
    (0, 0) rectangle (-136, 76);
  \draw[shift={(48, 800.002)}, yscale=1.0952, gold, ipe pen ultrafat]
    (0, 0) rectangle (144, -84);
  \draw[gold, ipe pen ultrafat]
    (44, 804) rectangle (196, 704);
  \filldraw[fill=white]
    (72, 784) rectangle (168, 724);
  \filldraw[fill=white]
    (72, 784) rectangle (168, 724);
  \filldraw[fill=white]
    (72, 784) rectangle (168, 724);
  \pic[ipe mark large, rgb color={draw=0.502 0 0}]
     at (52, 824) {ipe cross};
  \pic[ipe mark large, rgb color={draw=0.502 0 0}]
     at (16, 764) {ipe cross};
  \pic[ipe mark large, rgb color={draw=0.502 0 0}]
     at (192, 824) {ipe cross};
  \pic[ipe mark large, rgb color={draw=0.502 0 0}]
     at (216, 756) {ipe cross};
  \pic[ipe mark large, rgb color={draw=0.502 0 0}]
     at (200, 700) {ipe cross};
  \pic[ipe mark large, rgb color={draw=0 0 0.502}]
     at (24, 804) {ipe disk};
  \pic[ipe mark large, red]
     at (60, 784) {ipe disk};
  \pic[ipe mark large, red]
     at (176, 784) {ipe disk};
  \pic[ipe mark large, red]
     at (120, 756) {ipe disk};
  \pic[ipe mark large, red]
     at (152, 736) {ipe disk};
  \pic[ipe mark large, rgb color={draw=0 0 0.502}]
     at (160, 824) {ipe disk};
  \pic[ipe mark large, rgb color={draw=0 0 0.502}]
     at (76, 684) {ipe disk};
  \filldraw[shift={(232, 816)}, xscale=0.7896, yscale=0.9688, rgb color={fill=0.588 0.176 1}]
    (0, 0) rectangle (224, -128);
  \filldraw[shift={(235.892, 811.456)}, xscale=0.8077, yscale=1.0309, rgb color={fill=1 0.918 0.443}]
    (0, 0) rectangle (208, -112);
  \filldraw[shift={(268, 788)}, xscale=0.7222, yscale=0.85, rgb color={fill=0.588 0.176 1}]
    (0, 0) rectangle (144, -80);
  \filldraw[line cap=round, rgb color={fill=0.255 0.796 0.322}]
    (316, 760)
     -- (316, 760);
  \draw[shift={(388, 711.998)}, yscale=1.1053, gold, ipe pen ultrafat]
    (0, 0) rectangle (-136, 76);
  \draw[shift={(248, 800.002)}, yscale=1.0952, gold, ipe pen ultrafat]
    (0, 0) rectangle (144, -84);
  \draw[gold, ipe pen ultrafat]
    (244, 804) rectangle (396, 704);
  \filldraw[fill=white]
    (272, 784) rectangle (368, 724);
  \filldraw[fill=white]
    (272, 784) rectangle (368, 724);
  \filldraw[fill=white]
    (272, 784) rectangle (368, 724);
  \pic[ipe mark large, darkred]
     at (228, 828) {ipe cross};
  \pic[ipe mark large, darkred]
     at (380, 824) {ipe cross};
  \pic[ipe mark large, darkred]
     at (252, 684) {ipe cross};
  \pic[ipe mark large, darkred]
     at (420, 764) {ipe cross};
  \pic[ipe mark large, blue]
     at (264, 824) {ipe disk};
  \pic[ipe mark large, blue]
     at (320, 820) {ipe disk};
  \pic[ipe mark large, blue]
     at (304, 684) {ipe disk};
  \pic[ipe mark large, blue]
     at (416, 724) {ipe disk};
  \node[ipe node, font=\huge]
     at (168, 672) {$W_2$};
  \node[ipe node, font=\huge]
     at (368, 672) {$W_3
$};
  \draw[ipe pen ultrafat, ->]
    (160, 760)
     .. controls (200, 786.6667) and (248, 788) .. (304, 764);
\end{tikzpicture}}
	}~~~\raisebox{1cm}{$\rightarrow$}~~~
	{\resizebox{7cm}{!}{\begin{tikzpicture}[ipe stylesheet]
  \filldraw[shift={(32, 816)}, xscale=0.7896, yscale=0.9688, rgb color={fill=0.588 0.176 1}]
    (0, 0) rectangle (224, -128);
  \filldraw[shift={(35.892, 811.456)}, xscale=0.8077, yscale=1.0309, rgb color={fill=1 0.918 0.443}]
    (0, 0) rectangle (208, -112);
  \filldraw[shift={(68, 788)}, xscale=0.7222, yscale=0.85, rgb color={fill=0.588 0.176 1}]
    (0, 0) rectangle (144, -80);
  \filldraw[line cap=round, rgb color={fill=0.255 0.796 0.322}]
    (116, 760)
     -- (116, 760);
  \draw[shift={(188, 711.998)}, yscale=1.1053, gold, ipe pen ultrafat]
    (0, 0) rectangle (-136, 76);
  \draw[shift={(48, 800.002)}, yscale=1.0952, gold, ipe pen ultrafat]
    (0, 0) rectangle (144, -84);
  \draw[gold, ipe pen ultrafat]
    (44, 804) rectangle (196, 704);
  \filldraw[fill=white]
    (72, 784) rectangle (168, 724);
  \filldraw[fill=white]
    (72, 784) rectangle (168, 724);
  \filldraw[fill=white]
    (72, 784) rectangle (168, 724);
  \pic[ipe mark large, rgb color={draw=0.502 0 0}]
     at (52, 824) {ipe cross};
  \pic[ipe mark large, rgb color={draw=0.502 0 0}]
     at (16, 764) {ipe cross};
  \pic[ipe mark large, rgb color={draw=0.502 0 0}]
     at (192, 824) {ipe cross};
  \pic[ipe mark large, rgb color={draw=0.502 0 0}]
     at (216, 756) {ipe cross};
  \pic[ipe mark large, rgb color={draw=0.502 0 0}]
     at (200, 700) {ipe cross};
  \pic[ipe mark large, rgb color={draw=0 0 0.502}]
     at (24, 804) {ipe disk};
  \pic[ipe mark large, rgb color={draw=0 0 0.502}]
     at (160, 824) {ipe disk};
  \pic[ipe mark large, rgb color={draw=0 0 0.502}]
     at (76, 684) {ipe disk};
  \filldraw[shift={(232, 816)}, xscale=0.7896, yscale=0.9688, rgb color={fill=0.588 0.176 1}]
    (0, 0) rectangle (224, -128);
  \filldraw[shift={(235.892, 811.456)}, xscale=0.8077, yscale=1.0309, rgb color={fill=1 0.918 0.443}]
    (0, 0) rectangle (208, -112);
  \filldraw[shift={(268, 788)}, xscale=0.7222, yscale=0.85, rgb color={fill=0.588 0.176 1}]
    (0, 0) rectangle (144, -80);
  \filldraw[line cap=round, rgb color={fill=0.255 0.796 0.322}]
    (316, 760)
     -- (316, 760);
  \draw[shift={(388, 711.998)}, yscale=1.1053, gold, ipe pen ultrafat]
    (0, 0) rectangle (-136, 76);
  \draw[shift={(248, 800.002)}, yscale=1.0952, gold, ipe pen ultrafat]
    (0, 0) rectangle (144, -84);
  \draw[gold, ipe pen ultrafat]
    (244, 804) rectangle (396, 704);
  \filldraw[fill=white]
    (272, 784) rectangle (368, 724);
  \filldraw[fill=white]
    (272, 784) rectangle (368, 724);
  \filldraw[fill=white]
    (272, 784) rectangle (368, 724);
  \pic[ipe mark large, darkred]
     at (228, 828) {ipe cross};
  \pic[ipe mark large, darkred]
     at (380, 824) {ipe cross};
  \pic[ipe mark large, darkred]
     at (252, 684) {ipe cross};
  \pic[ipe mark large, darkred]
     at (420, 764) {ipe cross};
  \pic[ipe mark large, blue]
     at (264, 824) {ipe disk};
  \pic[ipe mark large, blue]
     at (320, 820) {ipe disk};
  \pic[ipe mark large, blue]
     at (304, 684) {ipe disk};
  \pic[ipe mark large, blue]
     at (416, 724) {ipe disk};
  \node[ipe node, font=\huge]
     at (172, 672) {$W_2$};
  \node[ipe node, font=\huge]
     at (372, 672) {$W_3
$};
  \pic[ipe mark large, red]
     at (264, 776) {ipe disk};
  \pic[ipe mark large, red]
     at (348, 768) {ipe disk};
  \pic[ipe mark large, red]
     at (312, 740) {ipe disk};
  \pic[ipe mark large, red]
     at (320, 776) {ipe disk};
\end{tikzpicture}}
}
\caption{The ``cleaning'' of the ``inner part'' of $\comp(W_{2})$.
Left:
The set $A(S)$ is depicted in cross vertices, the set $X\setminus X_{\sf in}$ is depicted in blue, and the set $X_{\sf in }$ is depicted in red.
Right:
The set $A(S)$ is depicted in cross vertices, the set $X'\setminus X_{\sf in}$ is depicted in blue, and the set $\tilde{X}$ is depicted in red.
}
\label{mkrpptenhjm}
\end{figure}
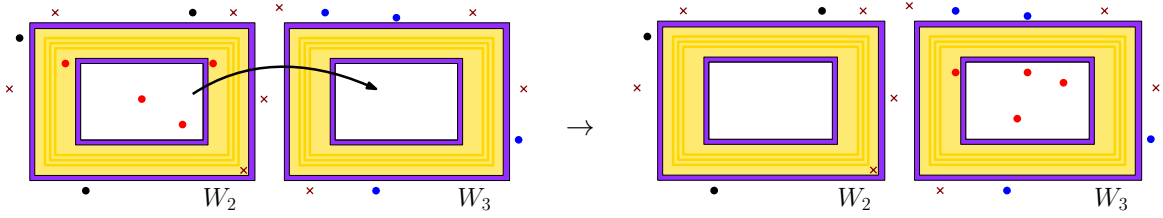
Also, we consider a $\tilde{t}\in[\rho]$ corresponding to a layer in the yellow region of the wall-annulus $A_{i}(W_{2})$ such that the annulus of the wall-annulus of $A_{i}(W_{2})$ bounded by the $(\tilde{t}-r+1)$-th and $\tilde{t}$-th layer of $W_{2}$ is not intersected by $X$.
Then, we ``replace'' the vertices of $X$ in $K_{2}^{(\tilde{t}-r+1)}\setminus P_2^{(\tilde{t}-r+1)}$, call it $X_{\sf in}$, with an ``equivalent'' vertex set $\tilde{X}$ in $K_{3}^{(\tilde{t}-r+1)}\setminus P_3^{(\tilde{t}-r+1)}$
(notice that this is achieved by arguing for $S:=\emptyset$ in the notion of $(φ,\boxtimes)$-characteristic).
This results to an $(\ell_{h}, r_{h})$-scattered set $X'$ that does not intersect $K_{2}^{(\tilde{t})}$ and  $G\boxtimes S\models \bigwedge_{x\in X'}\psi_{h}(x)$ (see \autoref{mkrpptenhjm}).

Now, we are allowed to pick a $t\in[\rho]$ corresponding to an ``orange'' layer of $A_{i}(W_{1})$ such that the annulus of the wall-annulus of $A_{i}(W_{1})$ bounded by the $(t-r+1)$-th and $t$-th layer of $W_{1}$ is not intersected by $X$.
If we set $Z$ to be the set of vertices of $X'$ in $K_{1}^{(t'-r+1)}\setminus P_1^{(t'-r+1)}$ ($P_1^{(t'-r+1)}$ is an ``extremal'' cycle of $A_{i}(W_{1})$ and therefore $X'$ does not intersect it), then
since ${{\sf sig}}_{φ,\boxtimes}(\mathfrak{K}_1, R_1, z,S_{\sf in})= {{\sf sig}}_{φ,\boxtimes}(\mathfrak{K}_2, R_2, z,\tilde{S})$, then there exists a set $\tilde{Z}$ in $K_{2}^{(t'-r+1)}\setminus P_2^{(t' -r+1)}$ that is ``equivalent'' to $Z$ (see \autoref{mkrpptdburgireenhjm}).
\begin{figure}[ht]
	\centering
{\resizebox{7cm}{!}{
\begin{tikzpicture}[ipe stylesheet]
  \filldraw[shift={(32, 816)}, xscale=0.7717, yscale=0.9688, rgb color={fill=0.588 0.176 1}]
    (0, 0) rectangle (224, -128);
  \filldraw[shift={(35.804, 811.456)}, xscale=0.7893, yscale=1.0309, rgb color={fill=1 0.918 0.443}]
    (0, 0) rectangle (208, -112);
  \filldraw[shift={(51.545, 800.201)}, xscale=0.9324, yscale=1.155, rgb color={fill=0.588 0.176 1}]
    (0, 0) rectangle (144, -80);
  \filldraw[shift={(115.41, 762.15)}, xscale=1.321, yscale=1.3588, draw=darkgreen, line cap=round, rgb color={fill=0.255 0.796 0.322}]
    (0, 0)
     -- (0, 0);
  \filldraw[shift={(56.709, 794.761)}, xscale=1.291, yscale=1.3588, rgb color={fill=0.173 0.878 0.208}]
    (0, 0) rectangle (96, -60);
  \filldraw[shift={(61.874, 789.325)}, xscale=1.0144, yscale=1.472, fill=orange]
    (0, 0) rectangle (112, -48);
  \filldraw[shift={(88.989, 767.591)}, xscale=1.211, yscale=1.6985, rgb color={fill=0.255 0.796 0.322}]
    (0, 0) rectangle (48, -16);
  \filldraw[shift={(94.273, 762.153)}, xscale=1.1889, yscale=2.0382, draw=darkgreen, fill=white]
    (0, 0) rectangle (40, -8);
  \filldraw[ipe pen ultrafat, fill=darkgray]
    (132, 820) rectangle (132, 820);
  \draw[shift={(67.182, 784)}, xscale=0.9773, rgb color={draw=0.992 0.4 0.008}, ipe pen ultrafat]
    (0, 0) rectangle (104, -60);
  \draw[shift={(71.091, 780)}, xscale=0.9773, rgb color={draw=0.992 0.4 0.008}, ipe pen ultrafat]
    (0, 0) rectangle (96, -52);
  \draw[shift={(75, 776)}, xscale=0.9773, rgb color={draw=0.992 0.4 0.008}, ipe pen ultrafat]
    (0, 0) rectangle (88, -44);
  \filldraw[shift={(232, 816)}, xscale=0.7717, yscale=0.9688, rgb color={fill=0.588 0.176 1}]
    (0, 0) rectangle (224, -128);
  \filldraw[shift={(235.804, 811.456)}, xscale=0.7893, yscale=1.0309, rgb color={fill=1 0.918 0.443}]
    (0, 0) rectangle (208, -112);
  \filldraw[shift={(251.546, 800.201)}, xscale=0.9324, yscale=1.155, rgb color={fill=0.588 0.176 1}]
    (0, 0) rectangle (144, -80);
  \filldraw[shift={(315.41, 762.15)}, xscale=1.321, yscale=1.3588, line cap=round, rgb color={fill=0.255 0.796 0.322}]
    (0, 0)
     -- (0, 0);
  \pic[ipe mark large, rgb color={draw=0.502 0 0}]
     at (252, 824) {ipe cross};
  \pic[ipe mark large, rgb color={draw=0.502 0 0}]
     at (216, 764) {ipe cross};
  \pic[ipe mark large, rgb color={draw=0.502 0 0}]
     at (388.3634, 824) {ipe cross};
  \pic[ipe mark large, rgb color={draw=0.502 0 0}]
     at (396.1816, 700) {ipe cross};
  \pic[ipe mark large, rgb color={draw=0 0 0.502}]
     at (224, 804) {ipe disk};
  \pic[ipe mark large, red]
     at (320.6938, 756.7151) {ipe disk};
  \pic[ipe mark large, red]
     at (362.9671, 729.5397) {ipe disk};
  \pic[ipe mark large, rgb color={draw=0 0 0.502}]
     at (357.0907, 824) {ipe disk};
  \pic[ipe mark large, rgb color={draw=0 0 0.502}]
     at (274.9997, 684) {ipe disk};
  \node[ipe node, font=\huge]
     at (364.909, 672) {$W_2$};
  \filldraw[shift={(256.71, 794.761)}, xscale=1.291, yscale=1.3588, rgb color={fill=0.173 0.878 0.208}]
    (0, 0) rectangle (96, -60);
  \filldraw[shift={(261.874, 789.325)}, xscale=1.0144, yscale=1.472, fill=orange]
    (0, 0) rectangle (112, -48);
  \filldraw[shift={(288.989, 767.591)}, xscale=1.211, yscale=1.6985, rgb color={fill=0.255 0.796 0.322}]
    (0, 0) rectangle (48, -16);
  \filldraw[shift={(294.273, 762.153)}, xscale=1.1889, yscale=2.0382, fill=white]
    (0, 0) rectangle (40, -8);
  \filldraw[ipe pen ultrafat, fill=darkgray]
    (332, 820) rectangle (332, 820);
  \draw[shift={(267.182, 784)}, xscale=0.9773, rgb color={draw=0.992 0.4 0.008}, ipe pen ultrafat]
    (0, 0) rectangle (104, -60);
  \draw[shift={(271.091, 780)}, xscale=0.9773, rgb color={draw=0.992 0.4 0.008}, ipe pen ultrafat]
    (0, 0) rectangle (96, -52);
  \draw[shift={(275, 776)}, xscale=0.9773, rgb color={draw=0.992 0.4 0.008}, ipe pen ultrafat]
    (0, 0) rectangle (88, -44);
  \node[ipe node, font=\huge]
     at (164, 672) {$W_1$};
  \pic[ipe mark large, darkred]
     at (24, 820) {ipe cross};
  \pic[ipe mark large, darkred]
     at (120, 828) {ipe cross};
  \pic[ipe mark large, darkred]
     at (88, 680) {ipe cross};
  \pic[ipe mark large, darkred]
     at (144, 704) {ipe cross};
  \pic[ipe mark large, darkgreen]
     at (84, 768) {ipe cross};
  \pic[ipe mark large, darkgreen]
     at (136, 756) {ipe cross};
  \pic[ipe mark large, darkgreen]
     at (84, 740) {ipe cross};
  \pic[ipe mark large, blue]
     at (48, 824) {ipe disk};
  \pic[ipe mark large, blue]
     at (148, 828) {ipe disk};
  \pic[ipe mark large, blue]
     at (192, 808) {ipe disk};
  \pic[ipe mark large, blue]
     at (20, 712) {ipe disk};
  \pic[ipe mark large, blue]
     at (176, 720) {ipe disk};
  \pic[ipe mark large, red]
     at (128, 736) {ipe disk};
  \pic[ipe mark large, red]
     at (112, 772) {ipe disk};
  \pic[ipe mark large, red]
     at (152, 768) {ipe disk};
  \draw[ipe pen ultrafat, ->]
    (152, 752)
     .. controls (224, 780) and (244, 780) .. (300, 756);
\end{tikzpicture}
}}~~~\raisebox{1cm}{$\ \rightarrow$}~~~
{\resizebox{7cm}{!}{
\begin{tikzpicture}[ipe stylesheet]
  \filldraw[shift={(32, 816)}, xscale=0.7717, yscale=0.9688, rgb color={fill=0.588 0.176 1}]
    (0, 0) rectangle (224, -128);
  \filldraw[shift={(35.804, 811.456)}, xscale=0.7893, yscale=1.0309, rgb color={fill=1 0.918 0.443}]
    (0, 0) rectangle (208, -112);
  \filldraw[shift={(51.545, 800.201)}, xscale=0.9324, yscale=1.155, rgb color={fill=0.588 0.176 1}]
    (0, 0) rectangle (144, -80);
  \filldraw[shift={(315.411, 762.15)}, xscale=1.321, yscale=1.3588, draw=darkgreen, line cap=round, rgb color={fill=0.255 0.796 0.322}]
    (0, 0)
     -- (0, 0);
  \filldraw[shift={(56.709, 794.761)}, xscale=1.291, yscale=1.3588, rgb color={fill=0.173 0.878 0.208}]
    (0, 0) rectangle (96, -60);
  \filldraw[shift={(61.874, 789.325)}, xscale=1.0144, yscale=1.472, fill=orange]
    (0, 0) rectangle (112, -48);
  \filldraw[shift={(288.989, 767.591)}, xscale=1.211, yscale=1.6985, rgb color={fill=0.255 0.796 0.322}]
    (0, 0) rectangle (48, -16);
  \filldraw[shift={(294.273, 762.153)}, xscale=1.1889, yscale=2.0382, draw=darkgreen, fill=white]
    (0, 0) rectangle (40, -8);
  \filldraw[ipe pen ultrafat, fill=darkgray]
    (132, 820) rectangle (132, 820);
  \draw[shift={(67.182, 784)}, xscale=0.9773, rgb color={draw=0.992 0.4 0.008}, ipe pen ultrafat]
    (0, 0) rectangle (104, -60);
  \draw[shift={(71.091, 780)}, xscale=0.9773, rgb color={draw=0.992 0.4 0.008}, ipe pen ultrafat]
    (0, 0) rectangle (96, -52);
  \draw[shift={(75, 776)}, xscale=0.9773, rgb color={draw=0.992 0.4 0.008}, ipe pen ultrafat]
    (0, 0) rectangle (88, -44);
  \filldraw[shift={(232, 816)}, xscale=0.7717, yscale=0.9688, rgb color={fill=0.588 0.176 1}]
    (0, 0) rectangle (224, -128);
  \filldraw[shift={(235.804, 811.456)}, xscale=0.7893, yscale=1.0309, rgb color={fill=1 0.918 0.443}]
    (0, 0) rectangle (208, -112);
  \filldraw[shift={(251.546, 800.201)}, xscale=0.9324, yscale=1.155, rgb color={fill=0.588 0.176 1}]
    (0, 0) rectangle (144, -80);
  \filldraw[shift={(315.41, 762.15)}, xscale=1.321, yscale=1.3588, line cap=round, rgb color={fill=0.255 0.796 0.322}]
    (0, 0)
     -- (0, 0);
  \pic[ipe mark large, rgb color={draw=0.502 0 0}]
     at (252, 824) {ipe cross};
  \pic[ipe mark large, rgb color={draw=0.502 0 0}]
     at (216, 764) {ipe cross};
  \pic[ipe mark large, rgb color={draw=0.502 0 0}]
     at (388.3634, 824) {ipe cross};
  \pic[ipe mark large, rgb color={draw=0.502 0 0}]
     at (396.1816, 700) {ipe cross};
  \pic[ipe mark large, rgb color={draw=0 0 0.502}]
     at (224, 804) {ipe disk};
  \pic[ipe mark large, red]
     at (320.6938, 756.7151) {ipe disk};
  \pic[ipe mark large, red]
     at (362.9671, 729.5397) {ipe disk};
  \pic[ipe mark large, rgb color={draw=0 0 0.502}]
     at (357.0907, 824) {ipe disk};
  \pic[ipe mark large, rgb color={draw=0 0 0.502}]
     at (274.9997, 684) {ipe disk};
  \node[ipe node, font=\huge]
     at (364.909, 672) {$W_2$};
  \filldraw[shift={(256.71, 794.761)}, xscale=1.291, yscale=1.3588, rgb color={fill=0.173 0.878 0.208}]
    (0, 0) rectangle (96, -60);
  \filldraw[shift={(261.874, 789.325)}, xscale=1.0144, yscale=1.472, fill=orange]
    (0, 0) rectangle (112, -48);
  \filldraw[shift={(288.989, 767.591)}, xscale=1.211, yscale=1.6985, rgb color={fill=0.255 0.796 0.322}]
    (0, 0) rectangle (48, -16);
  \filldraw[shift={(294.273, 762.153)}, xscale=1.1889, yscale=2.0382, fill=white]
    (0, 0) rectangle (40, -8);
  \filldraw[ipe pen ultrafat, fill=darkgray]
    (332, 820) rectangle (332, 820);
  \draw[shift={(267.182, 784)}, xscale=0.9773, rgb color={draw=0.992 0.4 0.008}, ipe pen ultrafat]
    (0, 0) rectangle (104, -60);
  \draw[shift={(271.091, 780)}, xscale=0.9773, rgb color={draw=0.992 0.4 0.008}, ipe pen ultrafat]
    (0, 0) rectangle (96, -52);
  \draw[shift={(275, 776)}, xscale=0.9773, rgb color={draw=0.992 0.4 0.008}, ipe pen ultrafat]
    (0, 0) rectangle (88, -44);
  \node[ipe node, font=\huge]
     at (164, 672) {$W_1$};
  \pic[ipe mark large, darkred]
     at (24, 820) {ipe cross};
  \pic[ipe mark large, darkred]
     at (120, 828) {ipe cross};
  \pic[ipe mark large, darkred]
     at (88, 680) {ipe cross};
  \pic[ipe mark large, darkred]
     at (144, 704) {ipe cross};
  \pic[ipe mark large, darkgreen]
     at (328, 764) {ipe cross};
  \pic[ipe mark large, darkgreen]
     at (312, 748) {ipe cross};
  \pic[ipe mark large, darkgreen]
     at (284, 740) {ipe cross};
  \pic[ipe mark large, blue]
     at (48, 824) {ipe disk};
  \pic[ipe mark large, blue]
     at (148, 828) {ipe disk};
  \pic[ipe mark large, blue]
     at (192, 808) {ipe disk};
  \pic[ipe mark large, blue]
     at (20, 712) {ipe disk};
  \pic[ipe mark large, blue]
     at (176, 720) {ipe disk};
  \pic[ipe mark large, red]
     at (344, 748) {ipe disk};
  \pic[ipe mark large, red]
     at (304, 760) {ipe disk};
  \pic[ipe mark large, red]
     at (308, 740) {ipe disk};
  \filldraw[shift={(88.989, 767.591)}, xscale=1.211, yscale=1.6985, rgb color={fill=0.255 0.796 0.322}]
    (0, 0) rectangle (48, -16);
  \filldraw[shift={(94.273, 762.153)}, xscale=1.1889, yscale=2.0382, fill=white]
    (0, 0) rectangle (40, -8);
\end{tikzpicture}
}}
\caption{The last part of the proof.
Left:
The set $A(S_{\sf out})$ is depicted in red cross vertices, the set $A(S_{\sf in})$ is depicted in green cross vertices, the set $Y\setminus Y_{\sf in}$ is depicted in blue, and the set $Y_{\sf in }$ is depicted in red.
Right:
The set $A(S_{\sf out})$ is depicted in red cross vertices, the set $A(\tilde{S})$ is depicted in green cross vertices, the set $X'\setminus Z$ is depicted in blue, and the set $\tilde{Z}$ is depicted in red.
}
\label{mkrpptdburgireenhjm}
\end{figure}
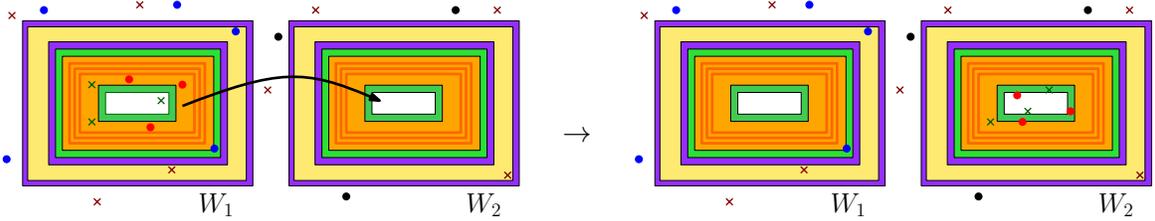
Therefore, since $\tilde{Z}$ is in the orange region of $\comp(W_{2})$ and $X'$ is ``avoiding''  $K_{2}^{(\tilde{t})}$, then we can derive that
$X'$ and $\tilde{Z}$ are ``separated'' by a green and a purple region of $A_{i}(W_{2})$.
Thus,
 $X^\star:=(X'\setminus Z) \cup \tilde{Z}$ is an $(\ell_{h}, r_{h})$-scattered set of $G\boxtimes S'$ that ``avoids'' $K_{1}^{(r)}$.
Moreover, $φ_{h}$ is satisfied given that the vertices of $X^\star$ of $G\boxtimes S'$ are assigned to the basic variables of $φ_{h}$.
The proof is concluded.

\section{Definitions and Preliminaries}\label{fenflknflkas}

We denote by $\Bbb{N}$ the set of all non-negative  integers. Given an $n\in\Bbb{N}$, we denote by $\Bbb{N}_{\geq n}$ the set containing 
all integers equal or greater than $n$.
Given  two integers $x$ and $y$, we define $[x,y]=\{x,x+1,\ldots,y-1,y\}$. Given an $n\in\Bbb{N}_{\geq 1}$, we also define $[n]=[1,n]$.
For a set $S$, we denote by $2^{S}$ the set of all subsets of $S$.

\subsection{Graphs, Walls, Wall-annuli, and Treewidth}

\paragraph{Basic concepts on Graphs.} 
All graphs in this paper are undirected,  finite, and they do not  have loops or multiple edges. Given a graph $G$, we denote by $V(G)$ and $E(G)$ the set of its vertices and edges, respectively. If $S\subseteq V(G)$, then we denote by $G\setminus S$ the graph obtained by $G$ after removing from it all vertices in $S$, together with their incident edges. 
Also, we denote by $G\setminus v$ the graph $G\setminus \{v\}$, for some $v\in V(G)$. We also denote by $G[S]$ the graph $G\setminus (V(G)\setminus S)$.
Given a graph $G$, we say that a pair $(A,B)\in 2^{V(G)}\times 2^{V(G)}$ is a {\em separation} of $G$ 
if $A\cup B=V(G)$ and there is no edge in $G$ with one endpoint in $A\setminus B$ 
and the other in $B\setminus A$.
A {\em path} ({\em cycle}) in a graph $G$ is a connected subgraph with all vertices of degree 
at most (exactly) 2.
Given a graph $G$, we define the {\em distance} $d_{G}(u,v)$ between two vertices $u,v$ of $G$, as the minimum number of edges of a path between $u$ and $v$ in $G$.
For $r\in\Bbb{N}_{\geq 1}$ and $u\in V(G)$ we define the {\em $r$-neighborhood} $N_{G}^{(\leq r)}(u)$ of $u$ in $G$ by $N_{G}^{(\leq r)}(u):=\{v\in V(G)\mid d_{G}(u,v)\leq r\}$.
We say that a set $S\subseteq V(G)$ is {\em $(\ell,r)$-scattered} if $|S|=\ell$ and for every $u,v\in V(G), u\neq v$ it holds that $d_{G}(u,v)>2r$.
An {\em annotated graph} is a pair $(G,R)$ where $G$ is a graph and $R\subseteq V(G)$.

\paragraph{Disks, annuli and partially disk-embedded graphs.}
In this paper, we consider embeddings or partial embeddings of graphs on the plane and several subsets of it.
We define a {\em closed disk} (resp. {\em open disk}) to be a subset of the plane homeomorphic to the set $\{(x,y)\in \Bbb{R}^2\mid x^2+y^2\leq 1\}$
(resp. $\{(x,y)\in \Bbb{R}^2\mid x^2+y^2< 1\}$) and
a {\em closed annulus} (resp. {\em open annulus}) to be a subset of the plane that is homeomorphic to the set $\{(x,y)\in \Bbb{R}^2 \mid 1\leq x^2+y^2\leq 2\}$ (resp. $\{(x,y)\in \Bbb{R}^2 \mid 1< x^2+y^2< 2\}$).
Given a closed disk or a closed annulus $X$, we use $\bd(X)$ to denote the boundary of $X$ (i.e., the set of points of $X$ for which every neighborhood around them contains some point not in $X$).
Notice that if $X$ is a closed disk then $\bd(X)$ is a subset of the plane homeomorphic to the set $\{(x,y)\in \Bbb{R}^2 \mid x^2+y^2 = 1\}$, while if $X$ is a closed annulus then $\bd(X)=C_{1}\cup C_{2}$ where $C_{1}, C_{2}$ are the two unique connected components of $\bd(X)$, that are two disjoint subsets of the plane, each one homeomorphic to the set $\{(x,y)\in \Bbb{R}^2 \mid x^2+y^2 = 1\}$.
We call these sets {\em boundaries} of $X$.
 Also given a closed disk (resp. closed annulus) $X$, we use $\inter(X)$ to denote the open disk $X\setminus \bd(X)$. When we embed a graph $G$ in a closed disk or in a closed annulus, we treat G as a set of points.
 This permits us to make set operations between graphs and sets of points.

%

We say that a graph $G$ is {\em partially disk-embedded in some closed disk $\Delta$}, 
if there is some subgraph $K$ of $G$ that is embedded in $\Delta$
such that $\bd(\Delta)$ is a cycle of $K$  and $(V(G)\cap \Delta,V(G)\setminus\inter(\Delta))$
is a separation of $G$. From now on, we use the term {\em partially $\Delta$-embedded graph $G$}
to denote that a graph $G$ is  partially disk-embedded in some closed disk $\Delta$. We also call the graph $K$
{\em compass}
of the partially $\Delta$-embedded graph $G$ and we always assume that we accompany
a partially $\Delta$-embedded graph $G$ together with an embedding of its compass in $\Delta$, that is the set $G\cap \Delta$.

\paragraph{Grids and walls.}
Let  $k,r\in\Bbb{N}.$ The
\emph{$(k\times r)$-grid} is the Cartesian product of two paths on $k$ and $r$ vertices respectively.
We use the term {\em $k$-grid} for the $(k\times k)$-grid.
An  \emph{elementary $r$-wall}, for some odd integer $r\geq 3,$ is the graph obtained from a
$(2 r\times r)$-grid with vertices $(x,y),$
$x\in[2r]\times[r],$ after the removal of the
``vertical'' edges $\{(x,y),(x,y+1)\}$ for odd $x+y,$ and then the removal of
all vertices of degree one. 
Notice that, as $r\geq 3$,  an elementary $r$-wall is a planar graph
that has a unique (up to topological isomorphism) embedding in the plane
such that all its finite faces are incident to exactly six edges.  
The {\em perimeter} of an elementary $r$-wall is the cycle bounding its infinite face,
while the cycles bounding its finite faces are called {\em bricks}. 
Given an elementary wall $\overline{W},$ a {\em vertical path} of $\overline{W}$ is  one whose 
vertices, in ordering of appearance, are $(i,1),(i,2),(i+1,2),(i+1,3),
(i,3),(i,4),(i+1,4),(i+1,5),
(i,5),\ldots,(i,r-2),(i,r-1),(i+1,r-1),(i+1,r)$, for some $i\in \{1,3,\ldots,2r-1\}$.
Also  an {\em horizontal path} of $\overline{W}$
is the one whose 
vertices, in ordering of appearance, are $(1,j),(2,j),\ldots,(2r,j)$, for some  $j\in[2,r-1]$, 
or  $(1,1),(2,1),\ldots,(2r-1,1)$ or  $(2,r),(2,r),\ldots,(2r,r)$.

%

An {\em $r$-wall} is any graph $W$ obtained from an elementary $r$-wall $\overline{W}$
by subdividing edges (see \autoref{asfdsfdsfasdfsdfdsf}).
We call the vertices that where added after the subdivision operations {\em subdivision vertices}, while we call the rest of the vertices (i.e., those of $\overline{W}$) {\em branch vertices}.
The {\em perimeter} of $W$, denoted by $\perim(W)$, is the cycle of $W$ whose non-subdivision vertices are the vertices of the perimeter of $\overline{W}$. Also, a vertical (resp. horizontal) path of $W$ is a subdivided vertical (resp. horizontal) path of $\overline{W}$.

A graph $W$ is a {\em wall} if it is an $r$-wall for some odd integer $r\geq 3$ and we refer to $r$ as the {\em height} of $W$.
Given a graph $G$, a {\em wall of $G$} is a subgraph of $G$ that is a wall.
We insist that, for every $r$-wall, the number $r$ is always odd.

Let $W$ be a wall of a graph $G$ and $K'$ be the connected component of $G\setminus \perim(W)$ that contains $W\setminus \perim(W)$.
The {\em compass} of $W$, denoted by $\comp(W)$, is the graph $G[V(K')\cup V(\perim(W))]$. Observe that $W$ is a subgraph of $\comp(W)$ and $\comp(W)$ is connected.

The {\em layers} of an $r$-wall $W$  are recursively defined as follows.
The first layer of $W$ is its perimeter. For $i=2,\ldots,(r-1)/2,$ the $i$-th layer of $W$ is the $(i-1)$-th layer of the subwall $W'$ obtained from $W$ after removing from $W$ its perimeter and all occurring vertices of degree one. Notice that each $(2r+1)$-wall has $r$ layers (see \autoref{asfdsfdsfasdfsdfdsf}). The {\em central  vertices} of $W$, denoted by $\cen(W)$, are the two branch vertices of $W$ that do not belong to any of its layers and that are connected by a path of $W$ that does not intersect any layer.

\paragraph{Treewidth.}A \emph{tree decomposition} of a graph~$G$
is a pair~$(T,\chi)$ where $T$ is a tree and $\chi: V(T)\to 2^{V(G)}$
such that
\begin{enumerate}
	\item $\bigcup_{t \in V(T)} \chi(t) = V(G)$,
	\item for every edge~$e$ of~$G$ there is a $t\in V(T)$ such that 
	$\chi(t)$
	contains both endpoints of~$e$, and
	\item for every~$v \in V(G)$, the subgraph of~${T}$
	induced by $\{t \in V(T)\mid {v \in \chi(t)}\}$ is connected.
\end{enumerate}
The \emph{width} of $(T,\chi)$ is defined as
$\w(T,\chi):= 
\max\big\{\left|\chi(t)\right|-1\ \bigmid t\in V(T)\big\}.$
The \emph{treewidth of $G$} is defined as
\[
\tw(G):= 
\min\big\{\w(T,\chi)\ \bigmid (T,\chi) \text{ is a tree decomposition of }G\big\}.
\]

The following result from \cite{GolovachKMT17thep} intuitively states that given an odd $q\in \Bbb{N}_{\geq 3}$ and a graph $G$  of ``big-enough'' treewidth, we can find a $q$-wall of $G$ whose compass has ``small enough'' treewidth.

\begin{proposition}[\cite{GolovachKMT17thep}]
	\label{something_good}
	There exists a constant $\newcon{sdafsdfsd}$  and an algorithm with the following specifications:

	\smallskip\noindent {\bf Find\_Wall}$(G,q)$\\
	\noindent{\sl Input:} a planar graph $G$ and an odd $q\in \Bbb{N}_{\geq 3}$.
	
	\noindent{\sl Output:} 	
	\begin{enumerate}
		\item A  $q$-wall $W$ of $G$ whose  compass  has treewidth at most $\conref{sdafsdfsd}\cdot q$ or  
		\item a tree decomposition of $G$ of width at most $\conref{sdafsdfsd}\cdot q$.
	\end{enumerate}
	Moreover, this algorithm runs in ${\cal O}_{q}(n)$ steps.
\end{proposition}

\subsection{Definitions and preliminary results on logic}

\paragraph{First-order and monadic second-order logic on graphs.} In this paper we deal with logic formulas on graphs. In particular we deal with formulas of first-order logic (FOL) and monadic second-order logic (MSOL).
The syntax of FOL-formulas includes the logical connectives $\vee, \wedge, \neg$, a set of variables for vertices, the quantifiers $\forall, \exists$ that are applied to these variables, the predicate $u\sim v$, where $u$ and $v$ are vertex variables and whose interpretation is that $u$ and $v$ are adjacent, and the equality of variables representing vertices.
A MSOL-formula, in addition to the variables for vertices of FOL-formulas, may also contain variables for subsets of vertices or subsets of edges.
The syntax of MSOL-formulas is obtained by enhancing the syntax of FOL-formulas
so to further allow quantification on subsets of vertices or subsets of edges 
and introducing the predicates $v\in S$  (resp. $e\in F$) whose 
interpretation is that the vertex  $v$ belongs in the vertex set $S$ (resp. the edge $e$ belongs in the edge set $F$).

An FOL-formula $φ$ is in {\em prenex normal form} if it is written as  $φ=Q_{1}x_{1}\ldots Q_{n}x_{n}\psi$ such that for every $i\in[n]$, $Q_{i}\in\{\forall, \exists\}$ and $\psi$ is a quantifier-free formula such that $x_{1},\ldots,x_{n}$ appear as variables in $\psi$.
Then $Q_{1}x_{1}\ldots Q_{n}x_{n}$ is referred as the {\em prefix} of $φ$.
For the rest of the paper, when we mention the term ``FOL-formula'', we mean an FOL-formula on graphs that is in prenex normal form.
Given an FOL-formula $φ$,  we say that a variable $x$ is a {\em free variable} in $φ$ if it does not occur in the prefix of  $φ$.
We write $φ(x_{1},\ldots, x_{r})$ to denote that $φ$ is a formula with free variables $x_{1},\ldots, x_{r}$.
We call a formula without free variables a {\em sentence}.
For a sentence $φ$ and a graph $G$, we write $G\models φ$ to denote that $φ$ evaluates to {\em true} on $G$.
Also, for a sentence $φ$ we denote its length by $|φ|$.
\medskip

We now prove that the property whether a given (planar) graph remains planar after making adjacent some given pairs of vertices can be expressed by an MSOL-formula.

\begin{lemma}\label{kdsafksdfkalsdfkla}
	Let $\boxtimes={\sf ea}$, $G$ be a graph, and $S\subseteq \boxtimes\langle G,V(G)\rangle$ where $S=\{\{v_{1},u_{1}\}, \ldots, \{v_{r},u_{r}\}\}$.
	Then there exists an MSOL-formula $φ_{\cal P,S}$ that is evaluated on structures of type $(G,x_{1},y_{1},\ldots,x_{r},y_{r})$ such that \[G\boxtimes S\text{ is a planar graph}\iff (G,v_{1},u_{1},\ldots,v_{r},u_{r}) \models φ_{{\cal P},S}.\]
\end{lemma}

\begin{proof}
	Notice that there exists an MSOL-formula $φ_{{\cal P}}$ on graphs such that $G$ is planar if and only if $G\models φ_{{\cal P}}$ (this holds since planarity is characterized by a finite set of forbidden topological minors, see also~\cite[Corollary 1.15]{CourcelleE12grap}).

	Now, modify the formula $φ_{{\cal P}}$ in order to transform it to a formula evaluated on structures of type $(G,v_{1},u_{1},\ldots,v_{r},u_{r})$. We define a new predicate $x\sim'y$, where $x,y$ are vertex variables such that
	\[x\sim'y := (x\sim y)\vee\bigvee_{i\in[r]}\big((x=v_{i} \wedge y=u_{i})\vee(x=u_{i} \wedge y=v_{i})\big)
	\]
	and replace in $φ_{{\cal P}}$ every occurrence of the predicate $x\sim y$ with $x\sim'y$.
	In other words, given two vertices $u,v$ of $G$ and two variables $x,y$ in $φ_{{\cal P},S}$,  where the variables $x,y$ are interpreted as the vertices $u,v$,  the predicate $x\sim'y$ is true if and only if $u,v$ are adjacent or $\{u,v\}\in S$.
	This implies that $G\boxtimes S\text{ is a planar graph}\iff (G,v_{1},u_{1},\ldots,v_{r},u_{r})\models φ_{{\cal P},S}.$
\end{proof}

\section{Proof of \autoref{sdmgflsadmgfrlka}}\label{dnksldngklfdnglkad}

In the proof \autoref{sdmgflsadmgfrlka}, the most intriguing part after finding a ``big-enough'' wall $W$ in $G$ such that $G\cap \comp(W)$ is a ``flat'' part of $G$,  is to prove that every inclusion-minimal planarizer of $G$ ``avoids'' the compass of $W$. In order to prove the latter, we define some notions regarding graphs that are ``partially embedded'' in an annulus and prove that we can ``glue'' together two such planar graphs on a way that the resulting graph is planar. This is materialized in \autoref{dnsfkanklda}
that we state and prove before we proceed to the proof of \autoref{sdmgflsadmgfrlka}.

\paragraph{Central subwalls and wall-annuli.}
Let $W$ be an $r$-wall of $G$, for some odd integer $r\geq 3$, and $L_{1},\ldots, L_{(r-1)/2}$ be the layers of $W$.
Let $q$ be an odd integer in $[3,r]$. We define the {\em central $q$-subwall} of $W$, which we denote by $W^{(q)}$, to be the graph obtained from $W$ after removing from $W$ its first $(r-q)/2$ layers and all occurring vertices of degree one (see \autoref{malkegmklrgnklg} for an example).

\begin{figure}[ht]
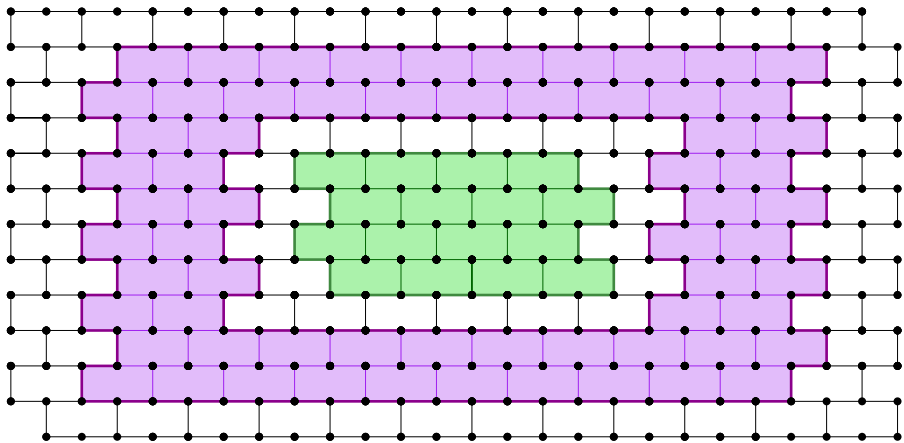

\centering
\resizebox{12cm}{!}{

}
\caption{A $13$-wall $W$, the central $5$-subwall $W^{(5)}$ of $W$ (depicted in green), and the $(5,3)$-wall-annulus $\mathcal{A}_{5}^{(3)}(W)$ of $W$ (depicted in green).}
\label{malkegmklrgnklg}
\end{figure}
Let $r\in \Bbb{N}_{\geq 7}$ be an odd integer, $p\in[3,(r-1)/2]$ and $\ell\in [3,p]$.
We define the {\em $(p,\ell)$-wall-annulus} of $W$, denoted by ${\cal A}_{p}^{(\ell)}(W)$, to be the graph obtained from  $W^{(2p+1)}$ after removing the vertices of $W^{(2(p-\ell)+1)}$ and all occurring vertices of degree one (see \autoref{malkegmklrgnklg} for an example).
Observe that, for every $i\in[p-\ell+1,p]$, ${\cal A}_{p}^{(\ell)}(W)$ contains the $i$-th layer of $W$ as a subgraph.
A {\em brick} of the $(p,\ell)$-wall-annulus ${\cal A}_{p}^{(\ell)}(W)$ of $W$ is a subgraph of ${\cal A}_{p}^{(\ell)}(W)$ that is also a brick of $W$.
A {\em 3-wall-annulus} of $W$ is a $(p,3)$-wall-annulus of $W$ for some $p\in \left[3,(r-1)/2\right]$.
Notice that every $(p,\ell)$-wall-annulus contains two ``boundary'' cycles that we call its {\em extremal cycles}.
Since $\ell\geq 3$, then ${\cal A}_{p}^{(\ell)}(W)$ is a subdivision of a 3-connected graph and therefore has a unique embedding in the plane.
Thus, given the embedding of  ${\cal A}_{p}^{(\ell)}(W)$ in the plane, we define the {\em annulus} of ${\cal A}_{p}^{(\ell)}(W)$, denoted by $\ann({\cal A}_{p}^{(\ell)}(W))$, to be the closed annulus in the plane bounded by the two extremal cycles of ${\cal A}_{p}^{(\ell)}(W)$.
%
%
%

%
%
	

\paragraph{Oriented annuli.}
An {\em oriented closed annulus} is a triple $\Bbb{A}=(A, C_{\sf in}, C_{\sf out})$ where $A$ is a closed annulus and $C_{\sf in}, C_{\sf out}$ are its boundaries, such that the connected component of $\Bbb{R}^{2}\setminus C_{\sf in}$ that does not intersect $A$, which we call the {\em inner compass} of $\Bbb{A}$ and we denote by $\comp_{\sf in}(\Bbb{A})$, is an open disk.
Also, we define the {\em outer compass} of $\Bbb{A}$ as the connected component of $\Bbb{R}^{2}\setminus C_{\sf out}$ that intersects $C_{\sf in}$ and denote it by $\comp_{\sf out}(\Bbb{A})$.
Given an oriented annulus $\Bbb{A}=(A, C_{\sf in}, C_{\sf out})$  we define ${\bf rev}(\Bbb{A}) =(A, C_{\sf out}, C_{\sf in} )$.

\paragraph{Annulus-boundaried graphs.}
An {\em annulus-boundaried graph} is a quadruple $(G,K,Y,\Bbb{A})$ (see \autoref{dnfsalknlksadnfklasdnflk}), where
\begin{itemize}
	\item $G$ is a graph,
	\item $K$ is a connected subgraph of $G$,
	\item $Y$ is a $3$-wall-annulus that is a subgraph of $K$,
	\item $\Bbb{A}=(A, C_{\sf in}, C_{\sf out})$ is an oriented closed annulus,
	\item $Y$ is embedded in $A$ such that $C_{\sf in}$ and $C_{\sf out}$ are the two extremal cycles of $Y$, and $G\cap A=K$.
\end{itemize}
We call the cycle of $Y$ that is identical to $C_{\sf in}$ (resp. $C_{\sf out}$) the {\em inner} (resp. {\em outer}) cycle of $(G,K,Y,\Bbb{A})$.

\paragraph{Wall-components of annulus-boundaried graphs}
Let $(G,K,Y,\Bbb{A})$ be an annulus-boundaried graph.  We now define the notion of a {\em wall-component} of $(G,K,Y,\Bbb{A})$.
We define two types of wall-components: edges of the form $e=uv\in E(G)\setminus E(Y)$ such that $u,v\in V(Y)$ and subgraphs of $K$ that are maximal connected components of $K\setminus V(Y)$.
A wall-component $Q$ is {\em attached} to a vertex $v\in V(Y)$ if it has a vertex adjacent to $v$, or (if $Q$ is an edge) one of its endpoints is $v$.
We say that a wall-component $Q$ of $(G,K,Y,\Bbb{A})$ is a {\em brick-component} if there exists a brick $B$ of $Y$ such that $Q$ is attached only to vertices in $V(B)$.
Given a subgraph $H$ of $Y$, let ${\sf att}(H)$ denote the subgraph of $G$ induced by the vertices of $H$ and the vertices of the wall-components which are only attached to $H$.

\begin{figure}[ht]
\centering
\includegraphics[width=9cm]{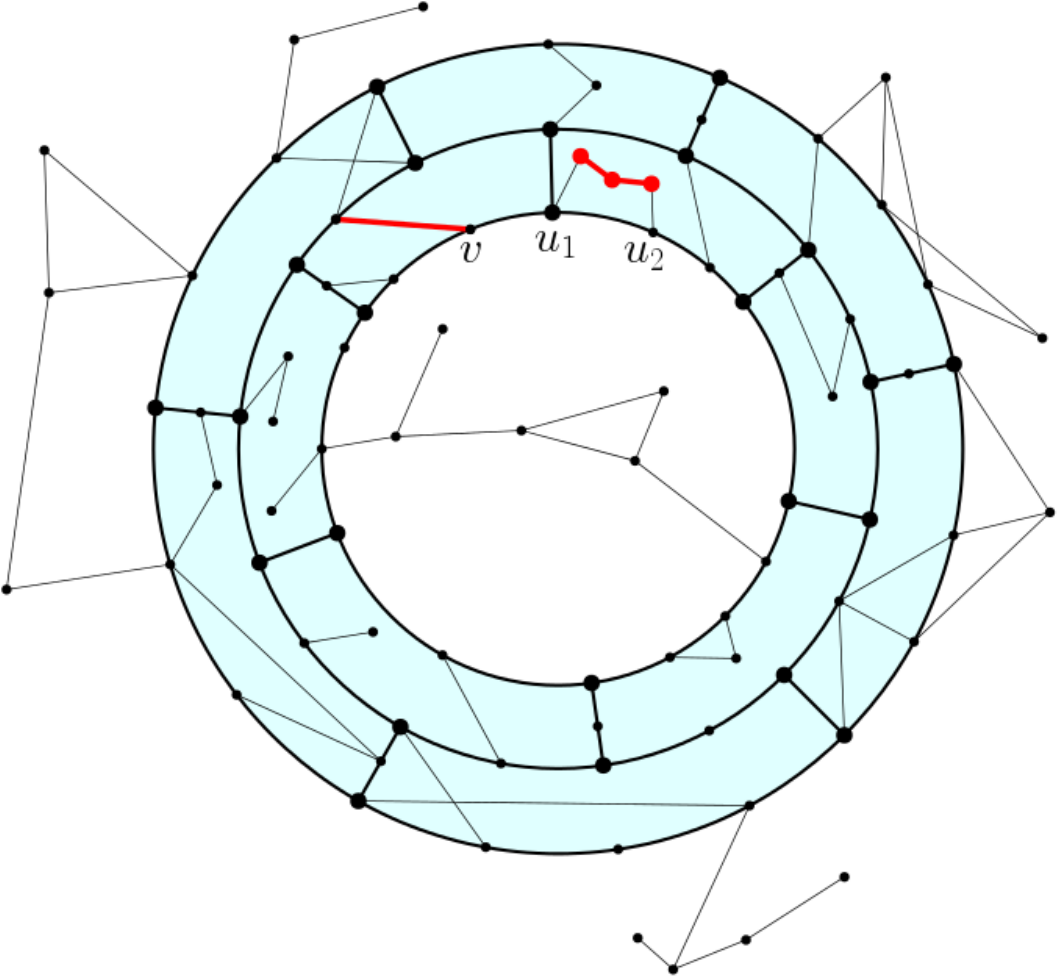}
\caption{An example of an annulus-boundaried graph $(G,K,Y,\Bbb{A})$. The annulus $A$ is depicted in blue. Two wall-components of $(G,K,Y,\Bbb{A})$  are depicted in red: An edge attached to $v\in V(G)$ and a subgraph of $K$ attached to $u_{1},u_{2}\in V(G)$.}
\label{dnfsalknlksadnfklasdnflk}
\end{figure}

The fact that $Y$ is a subdivision of a 3-connected graph and all embedding of the latter are equivalent implies the following result:

\begin{observation}\label{obsasfaffafgre}
	Let $(G,K,Y,\Bbb{A})$ be an annulus-boundaried graph and let $H$ be a subgraph of $Y$. If ${\sf att}(H)$ is planar, then every wall-component of $(G,K,Y,\Bbb{A})$ that is a subgraph of ${\sf att}(H)$ is either attached only to vertices of the inner/outer cycle of $(G,K,Y,\Bbb{A})$ or is a brick-component.
\end{observation}

%


\paragraph{Annulus-embedded separators.}
Let $G$ be a graph. Let also $(K,Y,\Bbb{A})$ be a triple where $K$ is a graph, $Y$ is a subgraph of $K$ and $\Bbb{A}$ is an oriented closed annulus.
We say that $(K,Y, \Bbb{A})$ is an {\em annulus-embedded separator} of $G$
if there are two subgraphs $G_{\sf in}$ and $G_{\sf out}$ of $G$ such that 
$V(G_{\sf in})\cup V(G_{\sf out})= V(G)$, $V(G_{\sf in}) \cap V(G_{\sf out}) = V(K)$, $(V(G_{\sf in}),V(G_{\sf out}))$ is a separation of $G$, and
both $(G_{\sf in}, K, Y, \Bbb{A})$ and $(G_{\sf out}, K, Y, {\bf rev}(\Bbb{A}))$ are annulus-boundaried graphs.
We call $G_{\sf in}$ (resp. $G_{\sf out}$) the {\em inner} (resp. {\em outer}) {\em component} of $(K,Y,\Bbb{A})$  in $G$.

We now prove the following result:

\begin{lemma}
	\label{dnsfkanklda}
	Let $G$ be a graph and let $(K,Y,\Bbb{A})$ be an annulus-embedded separator of $G$.
	Let also $G_{\sf in}$ and $G_{\sf out}$ be the inner and outer component of $(K,Y,\Bbb{A})$  in $G$, respectively.
	Then $G$ is a planar graph if and only if $G_{\sf in}$ and $G_{\sf out}$ are planar graphs.
\end{lemma}

\begin{proof}
	Observe that if $G$ is a planar graph then, trivially, $G_{\sf in}$ and $G_{\sf out}$ are planar graphs.
	We now prove that if $G_{\sf in}$ and $G_{\sf out}$ are planar graphs, then $G$ is also planar.
	
Suppose that $G_{\sf in}$ and $G_{\sf out}$ are planar graphs and also keep in mind that, since $G_{\sf in}$ and $G_{\sf out}$ are the inner and outer component of $(K,Y,\Bbb{A})$  in $G$, both $(G_{\sf in}, K, Y, \Bbb{A})$ and $(G_{\sf out}, K, Y, {\bf rev}(\Bbb{A}))$ are annulus-boundaried graphs.
Also, let $R_{\sf in}$ (resp. $R_{\sf out}$) be the subgraph of $G$ induced the union of the vertex sets of all bricks of $Y$ that intersect the inner (resp. outer) cycle of $(G,K,Y,\Bbb{A})$.

We begin by fixing a planar embedding $\theta$ of $G_{\sf out}$. 
Keep in mind that since $Y$ is a subdivision of a 3-connected planar graph, then all its plane embeddings are equivalent.
Observe that $\theta(R_{\sf out})$ is a region that divides the plane in two other regions (one finite and one infinite).
Assume that the graph $G_{\sf out}\setminus K$ is embedded in the infinite region.
	
Let $Q_{\sf in} := {\sf att}(R_{\sf in})$ and let $U$ denote the vertices of $Q_{\sf in}$ that are adjacent to some vertex of $G_{\sf out}\setminus V(Q_{\sf in})$. For more intuition, notice that $U$ is a subset of $V(R_{\sf in})\cap V(R_{\sf out})$.
To prove the latter, suppose towards a contradiction that there is  a vertex $v\in U$ that is not in $V(R_{\sf in}) \cap V(R_{\sf out})$.
Observe that $v$ is a vertex of a wall-component $H$ of $(G_{\sf out}, K, Y, {\bf rev}(\Bbb{A}))$ that is also a subgraph of $Q_{\sf in}$.
Since $v\in U$, there exists a vertex $u$ of $G_{\sf out}\setminus V(Q_{\sf in})$ such that $v$ and $u$ are adjacent.
Notice that by the definition of wall-component, it follows that $u\in V(Y)$.
But then $H$ is attached to $u$ and since $u\notin V(Q_{\sf in})$, we arrive to a contradiction to the definition of $Q_{\sf in}$ and \autoref{obsasfaffafgre}.
Observe that the restriction of $\theta$ to $G_{1}:=G_{\sf out}\setminus V(Q_{\sf in}\setminus U)$ has a face whose boundary contains $U$.

	Now let $φ$ be a planar embedding of $G_{\sf in}$ and let us restrict $φ$ to $G_{2}:=G\setminus V(G_{1}\setminus U)$. Observe that $Q_{\sf in}\subseteq V(G_{2})$.
	Note that $U$ contains only vertices that are adjacent to some vertex in $R_{\sf out}$ or are adjacent to brick-components belonging to a brick of $R_{\sf out}$. But $φ$ embeds $R_{\sf out}$ and its brick-components also, and therefore the restriction of $φ$ to $G_{2}$ results in a face whose boundary contains $U$.
	
	Now observe that by combining $\theta$ and $φ$ in such a way that we embed $G_{1}$ according to $\theta$ and  $G_{2}$  according to $φ$ and then ``match'' them by identifying $\theta(u)$ and $φ(u)$ for all $u\in U$, we get a planar embedding of $G$.
\end{proof}

%

Before we proceed with the proof of \autoref{sdmgflsadmgfrlka}, we need some more definitions.

\paragraph{Graph contractions.}
Let $G$ and $H$ be graphs and let $\rho : V(G)\rightarrow V(H)$ be a surjective mapping such that:
\begin{enumerate}
	\item for every vertex $v\in V(H)$, its codomain $\rho^{-1}(v)$ induces a connected graph $G[\rho^{-1}(v)]$,
	\item for every edge $\{u,v\}\in E(H)$, the graph $G[\rho^{-1}(u)\cup \rho^{-1}(v)]$ is connected, and
	\item for every edge $\{u,v\}\in E(G)$, either $\rho(u)=\rho(v)$ or $\{\rho(u), \rho(v)\}\in E(H)$.
\end{enumerate}
We say that {\em $H$ is a contraction of $G$ (via $\rho$)} and for a vertex $v\in V(H)$ we call the codomain $\rho^{-1}(v)$ the {\em model of $v$} in $G$.

\paragraph{Central grids.}
Let  $k,r\in\Bbb{N}_{\geq 2}.$
We define the {\em perimeter} of a $(k\times r)$-grid to be the unique cycle of the grid of length at least three that that does not contain vertices of degree four.
Let $r\in \Bbb{N}_{\geq 2}$ and $H$ be an $r$-grid.
Given an $i\in\lceil \frac{r}{2}\rceil,$ we define the {\em $i$-th layer} of $H$ recursively as follows.
The first layer of $H$ is its perimeter, while, if $i\geq 2,$ the $i$-th layer of $H$ is
the $(i-1)$-th layer of the grid created if we remove from $H$ its perimeter.
Given two  odd integers $q,r\in\Bbb{N}_{\geq 3}$ such that $q\leq r$ and an $r$-grid $H,$
we define the {\em central $q$-grid} of $H$ to be the graph obtained from $H$
if we remove from $H$ its $\frac{r-q}{2}$ first layers.

\paragraph{Triangulated grids.}
We now define the {\em triangulated $k$-grid} $\Gamma_{k}$.
Consider a plane embedding of the $k$-grid such that all external vertices are on the boundary of the infinite face.
We triangulate the internal faces of the $k$-grid (the faces that are incident to exactly four edges) such that all internal vertices have degree $4$ in the obtained graph and all non-corner external vertices have degree $4$.
Finally, one corner of degree $2$ is joined by edges with all the extremal vertices and we call this vertex {\em loaded} (see example in \autoref{ndslkfvnadsklf}). We refer to the initial $k$-grid as the {\em underlying grid of $\Gamma_{k}$}.
\begin{figure}[ht]
\centering
	\includegraphics[width=6cm]{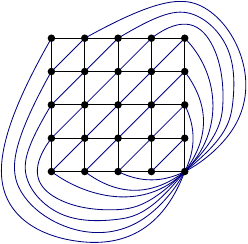}
\caption{The graph $\Gamma_{5}$.}
\label{ndslkfvnadsklf}
\end{figure}

Before we proceed to the proof of \autoref{sdmgflsadmgfrlka}, we need two results that will be useful.

\begin{proposition}[~\cite{FominGT11cont}]\label{bidimensionaltri}
	Let $G$ be a connected planar graph and $k$ be a positive integer. There is a constant $\newcon{dsjkfnajk}$ such that if $\tw(G)>\conref{dsjkfnajk}\cdot k$, then $G$ contains $\Gamma_{k}$ as a contraction.
\end{proposition}

\begin{proposition}[~\cite{DemaineFHT04bidi}]
	\label{localtw}
	Let $H$ be the $m$-grid and a subset $U$ of vertices in the central $(m-2\ell)$-grid $\hat{H}$ of $H$, where $|U|=s$ and $\ell=\lfloor \sqrt[4]{s}\rfloor$. Then $H$ contains the $\ell$-grid $R$ as a minor such that the model of each vertex of $R$ intersects $U$.
\end{proposition}

In the following proof, we use $\conref{sdafsdfsd},\conref{dsjkfnajk}$ to denote the constants in \autoref{something_good} and \autoref{bidimensionaltri}, respectively.

\begin{proof}[Proof of \autoref{sdmgflsadmgfrlka}]
	
We set $m = 3 \cdot (2k+1)$,
\begin{align*}
r & :=  2\cdot (2m+q)+1,&
z &:=\conref{sdafsdfsd}\cdot r+2,&
\funref{jsdfjbnvjfdak}(k,q) & :=z-2,\\
\ell & :=4\lceil\sqrt{k+1}\rceil -1,&
{b} & :=2\ell + \sqrt{\ell^{4}\cdot k } \cdot {z},~\mbox{and}&
\funref{lksgjreklgjrnjighpotrr}(k,q) & :=\max\{\conref{dsjkfnajk}\cdot {b} + k, \conref{sdafsdfsd}\cdot q\}.
\end{align*}
We begin with the case where $\boxtimes={\sf ea}$.
Observe that if $G$ is not planar, then $(G,k)$ is a \no-instance of \mnb{\mnb{\sc \sc  G$\boxtimes$MP$φ$}}.
If $G$ is planar and if it is the case that $\tw(G)>\funref{lksgjreklgjrnjighpotrr}(k,q)\geq  \conref{sdafsdfsd}\cdot q$, we call  the algorithm {\bf Find\_Wall}$(G,q)$ of \autoref{something_good} and  we get a $q$-wall $W$ of $G$ whose  compass  has treewidth at most $\conref{sdafsdfsd}\cdot q$.
Since $\conref{sdafsdfsd}\cdot q<\funref{jsdfjbnvjfdak}(k,q)$, the claimed bound on the treewidth of $\comp(W)$ follows.
We also set $\Delta:=\perim(W)\cup J$, where $J$ is the connected component of $\Bbb{R}^{2}\setminus \perim(W)$ that contains $W\setminus V(\perim(W))$.
Observe that $\Delta$ is a closed disk and therefore $G$ is partially $\Delta$-embedded, where $G\cap \Delta=\comp(W)$ and $\bd(\Delta)=\perim(W)$.

Therefore, in the rest of the proof we consider the case where $\boxtimes\in\{{\sf vd}, {\sf ed}, {\sf ec}\}$.
We consider an embedding $\theta$ of $G\setminus S$ in the plane.
Suppose that $\tw(G)>\funref{lksgjreklgjrnjighpotrr}(k,q)$.
Let $G'$ be a connected component of $G\setminus S$ (if $G\setminus S$ is connected, $G' := G\setminus S$) such that $\tw(G')= \tw(G\setminus S)$.
Therefore, we have that $\tw(G')>\funref{lksgjreklgjrnjighpotrr}(k,q)-k\geq \conref{dsjkfnajk}\cdot {b}$.
Then, by \autoref{bidimensionaltri}, $G'$ contains $\Gamma_{{b}}$ as a contraction.
Let $H$ be the underlying grid of $\Gamma_{{b}}$ and $\hat{H}$ be the central $({b}-2\ell)$-grid of $H$.
 
For every vertex $u\in S$, let \[N_{u}:=\{v\in V(\hat{H}) \mid ~\mbox{the model of}~v~\mbox{intersects}~N_{G'}(u)\}.\]
Let $N:=\bigcup_{u\in S}N_{u}$.
	We consider the following cases, concerning the size of $N$:\medskip
	
	\noindent{\bf Case 1:} $|N|\geq\ell^{4}\cdot k$.\smallskip
	
	In this case, there exists a vertex $u\in S$ such that $|N_{u}|\geq \ell^{4}$.
	Let $U$ be a subset of $N_{u}$ such that $|U|=\ell^{4}$.
	Then, by \autoref{localtw}, $H$ contains the $\ell$-grid as a minor and every vertex of the latter is adjacent to $u$.
	This, together with the fact that  $\ell=4\lceil\sqrt{k+1}\rceil -1$, implies that $G$ contains a $(K_{5},k+1)$-star\footnote{Given an $r\geq 1$ we define the graph $(K_{5},r)$-star as the graph obtained by taking $r$ copies of $K_{4}$ (that is the complete graph on $4$ vertices) and a vertex $v$ and making $v$ adjacent to all vertices of the $r$ copies of $K_{4}$. We call $v$ the {\em central vertex of the $(K_{5},r)$-star}.}
 as a minor with $u$ as its central vertex. 
Observe that if $\boxtimes={\sf ed},{\sf ec}$, the latter implies that $(G,k)$ is a \no-instance (since we can not eliminate all $k+1$ copies of $K_{5}$ from $G$ by deleting/contracting $k$ edges),
while if $\boxtimes={\sf vd}$, for every ${\sf vd}$-planarizer $S'$ of $G$ of size at most $k$ it holds that $u\in S'$
 (intuitively, $u$ is an ``obligatory'' vertex for every ${\sf vd}$-planarizer of $G$ of size at most $k$).
Also, observe that $S\setminus \{u\}$ is a ${\sf vd}$-planarizer of $G\setminus u$ of size at most $k-1$ and notice that $(G,k)$ and $(G\setminus u, k-1)$ are equivalent instances of \mnb{\mnb{\sc \sc  G$\boxtimes$MP$φ$}}.
The above consitute the first possible output of the algorithm {\bf Find\_Area}$(k,q,G,S)$ of \autoref{sdmgflsadmgfrlka} and this concludes Case 1.\medskip

\noindent{\bf Case 2:} $|N|< \ell^{4}\cdot k$.\smallskip
	
In this case, we first argue that the following holds:\medskip

\noindent{\em Claim 1:} There exists a wall $\tilde{W}$ of $G'$ of height $2m+q$ such that $N_{G'}(S)\cap V(\comp(\tilde{W}))=\emptyset$.\smallskip

\noindent{\em Proof of Claim 1:}
Since $|N|< \ell^{4}\cdot k$, $\hat{H}$ is a $(b-2\ell)$-grid, and $b - 2\ell = \sqrt{\ell^{4}\cdot k}\cdot z$,
there exists a $z$-grid $H'$ that is a subgraph of $\hat{H}$ such that $N\cap V(H')=\emptyset$.
	
Let $w$ denote some corner of $H'$.
Consider a surjective mapping $\rho: V(\Gamma_{{b}}) \rightarrow V(H')$ that maps every vertex in $V(H')$ to itself and every vertex in $V(\Gamma_{{b}})\setminus V(H')$ to $w$. This results to a graph $R$ that is a contraction of $G'$ (via $\rho$).
Notice that $R\cong \Gamma_{ {z}}$, where the model of its loaded vertex $w$ contains $N_{G'}(S)$, and $\tw(R)\geq  {z}$.
	
Consider now the set
$V_{w}:=\left\{v\in V(G')\mid v~\mbox{is in the model of}~ w\right\}$, and observe that $G'[V_{w}]$ is a connected graph.
Since $\tw(R)\geq  {z} $, then $\tw(G'\setminus V_{w})\geq {z}-1>\conref{sdafsdfsd}\cdot r $.
By applying the algorithm {\bf Find\_Wall}$(G,q)$ of \autoref{something_good} for $G'\setminus V_{w}$ and $r$, we get a $r$-wall $W'$ of $G'\setminus V_{w}$ whose  compass  has treewidth at most $\conref{sdafsdfsd} \cdot r=\funref{jsdfjbnvjfdak}(k,q)$.
Notice that, since $G'[V_{w}]$ is connected and $G'$ is planar, then $N_{G'}(S)$ (being a subset of $V_w$) is entirely contained in a unique face of $W'$ (recall that since we fixed an embedding $\theta$ of $G'$, we can treat the vertices of $G'$ as points on the plane).
Therefore, since $W'$ has height $r = 2\cdot (2m +q)+1$, there exists a subwall $\tilde{W}$ of $W'$ of height $2m + q$ that is a wall of $G'$ and $N_{G'}(S)\cap V(\comp(\tilde{W}))=\emptyset$.
Claim 1 follows.\hfill$\diamond$\medskip
	
	By Claim 1, there exists a wall $\tilde{W}$ of $G'$ of height $2m + q$ such that  $N_{G'}(S)\cap V(\comp(\tilde{W}))=\emptyset$.
	Therefore, by restricting the embedding $\theta$ of $G\setminus S$ in $\comp(\tilde{W})$, we get that $\comp(\tilde{W})$ is a planar graph.
	Let $W$ be the central $q$-subwall of $\tilde{W}$.
	We now argue that the following holds:\medskip
	
	\noindent{\em Claim 2:} The set $V(\comp(W))$ is $\boxtimes$-planarization irrelevant.\medskip
	
	\noindent{\em Proof of Claim 2:}
	Suppose, towards a contradiction, that there is a set $Z\subseteq\boxtimes\langle G,V(G)\rangle$ such that $Z$ is an inclusion-minimal $\boxtimes$-planarizer and $A(Z)\cap V(\comp(W))\neq \emptyset$.
	
Since $\tilde{W}$ is a wall of height $2m+ q$, it has at least $m$ layers.
For every $i\in[m]$, let $C_{i}$ be the $i$-th layer of $\tilde{W}$.
For every $i\in[m-2]$, let $A_{i}$ be the finite region of $\Bbb{R}^{2}$ bounded by $φ(C_{i})$ and $φ(C_{i+2})$ 	(the wall $\tilde{W}$ is a subdivision of a 3-connected graph and therefore all its embeddings in the plane are equivalent) and let $\Bbb{A}_{i}:=(A_{i}, C_{i+2}, C_{i})$.
	
	Since $|A(Z)|\leq 2k$ and $m = 3\cdot (2k+1)$, then there exists an $i_{Z}\in [m-2]$ 
	and a subgraph $Y$ of $\tilde{W}$ such that $A(Z)\cap A_{i_{Z}}=\emptyset$ and $Y$ is a $3$-wall-annulus whose extremal cycles are $C_{i_{Z}}$, $C_{i_{Z}+2}$.
	For simplicity, we denote $A:=A_{i_{Z}}$ and $\Bbb{A}:=\Bbb{A}_{i_{Z}}$.

Let $K$  be the maximal connected subgraph of $G$ such that $G\cap A=K$.
We denote by $G_{\sf in}$ the graph
$G[(V(G)\cap \comp_{\sf out}(\Bbb{A}))\cup V(C_{i_{Z}+2})]$ and
with $G_{\sf out}$ the graph $G\setminus (V(G) \cap \comp_{\sf in}(\Bbb{A}))$ 
and consider the annulus-boundaried graphs $(G_{\sf in}, K, Y, \Bbb{A})$ and 
$(G_{\sf out}, K, Y, {\bf rev}(\Bbb{A}))$.
Notice that $(K, Y, \Bbb{A})$ is an annulus-embedded separator of $G$.
Also, since $S$ is a ${\sf vd}$-planarizer of $G$ and
$N_{G'}(S)\cap V(\comp(\tilde{W}))=\emptyset$, then $G_{\sf in}$
is planar (since $G_{\sf in}$ is a subgraph of $G'$ and $G'$ is planar).

Notice that since $A(Z)\cap A=\emptyset$, $\comp(\tilde{W})$ is planar and 
$Y$ is a $3$-wall-annulus of $G'$ whose extremal cycles are the boundaries 
of $A$, then
there is no $x\in Z$ that affects vertices of $G$ in both connected components 
of $\Bbb{R}^{2}\setminus A$.
In other words, $Z$ is partitioned in two sets $Z_{\sf in}$ and $Z_{\sf out}$, 
where $A(Z_{\sf in})$ is in $\comp_{\sf in}(\Bbb{A})$ and  $A(Z_{\sf out})$ is in 
$\Bbb{R}^{2}\setminus \comp_{\sf out}(\Bbb{A})$.
Now, observe that since $\comp(W)$ is a graph embedded in a subset of
$\comp_{\sf in}(\Bbb{A})$ and $A(Z)\cap V(\comp(W))\neq \emptyset$, then 
$Z_{\sf in}\neq \emptyset$.
Thus $Z_{\sf out}$ is a proper subset of $Z$.
Also, the fact that $Z$ is a
$\boxtimes$-planarizer of $G$, implies that $Z_{\sf out}$
is a $\boxtimes$-planarizer of $G_{\sf out}$.
Hence, $G_{\sf out}\boxtimes Z_{\sf out}$ is planar.
Moreover, $(K,Y, \Bbb{A})$ is an annulus-embedded separator of $G\boxtimes Z_{\sf out}$.
	
	Therefore, since $(K,Y, \Bbb{A})$ is an annulus-embedded separator of $G\boxtimes Z_{\sf out}$ and $G_{\sf in}$ and $G_{\sf out}\boxtimes Z_{\sf out}$ are planar graphs, by \autoref{dnsfkanklda} we have that $G\boxtimes Z_{\sf out}$ is a planar graph, a contradiction to the minimality of $Z$. Claim 2 follows.\hfill$\diamond$\medskip
	
Following Claim 2, $W$ is a $q$-wall of $G$ 
whose compass has treewidth at most $\funref{jsdfjbnvjfdak}(k,q)$ and  $V(\comp(W))$ is $\boxtimes$-planarization irrelevant.
Keep in mind that $\comp(W)$ is a planar graph, since it is a subgraph of the planar graph $\comp(\tilde{W})$.
Now, let $J$ be the connected component of $\Bbb{R}^{2}\setminus \perim(W)$ that contains $W\setminus\perim(W)$. Observe that $\Delta:=\perim(W)\cup J$ is a closed disk and therefore $G$ is partially $\Delta$-embedded, where $G\cap \Delta=\comp(W)$.
Therefore, the algorithm {\bf Find\_Area}$(k,q,G,S)$ of \autoref{sdmgflsadmgfrlka} returns $W$ and $\Delta$ and this completes the proof of the lemma.
\end{proof}

\section{Proof of \autoref{wedndlsakfnlkdsafnkalen}}\label{dksdlgdsgjkloew}

In this section we present the proof of \autoref{wedndlsakfnlkdsafnkalen}, that is the main technical result of this paper.
In~\autoref{subsec6_1}, we define the notion of characteristic of the panelled compass of a wall,
that encodes 
all possible ways that a $\boxtimes$-planarizer $S$ of $G$ affects $\comp(W)$ along with the different ways a vertex assignment to the basic variables of the Gaifman formula $φ$ in $\comp(W)$ can certify $G\boxtimes S\modelsφ$.
In~\autoref{subsec6_2} we describe the algorithm {\tt Find\_Vertex} of~\autoref{wedndlsakfnlkdsafnkalen} and in~\autoref{subsec6_3} we prove its correctness.
Also, throughout this section, we use $\funref{jsdfjbnvjfdak}$ to denote the function in \autoref{sdmgflsadmgfrlka}, bounding the treewidth of the compass of the wall that the claimed algorithm outputs.

\subsection{Characteristic of the panelled compass of a wall}\label{subsec6_1}

\paragraph{Panelled compass of a wall.}
Let $\rho\in\mathbb{N}_{\geq 1}$,
let  $G$ be a partially $\Delta$-embedded graph, let $W$ be a  $(2\rho+1)$-wall of $G$ such that  $\comp(W)\subseteq \Delta$.
We set $K=\comp(W)$ and, for every $t\in [\rho]$, we set $K^{(t)}= \comp(W^{(2t+1)})$ and $P^{(t)} = V(\perim(W^{(2t+1)}))$.
Let ${\bf K} = (V(K^{(1)}),\ldots, V(K^{(\rho)}))$. We call the tuple $\mathfrak{K}_W = (K,{\bf K})$ the {\em panelled compass} of the wall $W$ in $G$.

\paragraph{Characteristics.}
Let $φ$ be a Gaifman sentence.
By definition, $φ$ is a Boolean combination of sentences $φ_{1}, \ldots, φ_{m}$ such that, for every $h\in[m]$, 		\[φ_{h}=\exists x_{1}\ldots\exists x_{\ell_{h}}\big( \bigwedge_{1\leq i<j\leq \ell_{h}} d(x_{i}, x_{j})> 2r_{h}\wedge \bigwedge_{i\in [\ell_{h}]}\psi_{h}(x_{i})\big),\]
where $\ell_{h},r_{h}\geq 1$ and $\psi_{h}(x)$ is $r_{h}$-local.
We consider the sentence $\tilde{φ}$ and recall that it is the same Boolean combination 
of sentences $\tilde{φ}_{1}, \ldots, \tilde{φ}_{m}$ such that, for every $h\in[m]$, 	
\begin{eqnarray*}
\tilde{φ}_{h}=\exists x_{1}\ldots\exists x_{\ell_{h}}\big(\bigwedge_{i\in [\ell_{h}]}x_{i}\in R\wedge \bigwedge_{1\leq i<j\leq \ell_{h}} d(x_{i}, x_{j})> 2r_{h}\wedge \bigwedge_{i\in [\ell_{h}]}\psi_{h}(x_{i})\big),\label{gaidsfd}
\end{eqnarray*}
and the formulas $\tilde{φ}$ and $\tilde{φ}_1, \ldots, \tilde{φ}_m$ are evaluated on annotated graphs of the form $(G,R)$.

We set $r:=\max_{h\in[m]}\{r_{h}\}$ and $\ell:=\sum_{h\in[m]}\ell_{h}$ and
\begin{align*}	
	d & :=2\left(r + (\ell+1)r+r\right),\\
	\rho & := (2k+1)\cdot d.
\end{align*}
Let
$${\sf SIG}= 2^{[\ell_1]}\times\cdots\times 2^{[\ell_{m}]}\times[\rho].$$

Let $\boxtimes\in{\sf OP}$.
Let $G$ be a partially $\Delta$-embedded graph, let $W$ be a  $(2\rho+1)$-wall of $G$ such that  $\comp(W)\subseteq \Delta$.
Given the panelled compass $\mathfrak{K}_W$ of $W$ in $G$, a set $R\subseteq V(\comp(W))$, an integer $z\in [d,\rho]$, and a set $S\subseteq \boxtimes\langle K,R\rangle$ such that $A(S)\subseteq V(K^{(z-d+1)})\cap R$, we define

\begin{eqnarray*}
{{\sf sig}}_{φ,\boxtimes}(\mathfrak{K}_W, R, z,S) & =&  \{(Y_{1},\ldots,Y_{m},t)\in {\sf SIG}  \mid  t\leq z \mbox{~and~} \exists \ (\tilde{X}_{1},\ldots,\tilde{X}_m)\mbox{~such that~} \forall h\in[m]\  \\
& &~~~~~~~~~~~~~~~~~~~~~~~~~~~~~~~~~~~~~~~~ \tilde{X}_{h}=\{x_{i}^{h}\mid i\in Y_{h}\},\\
& &~~~~~~~~~~~~~~~~~~~~~~~~~~~~~~~~~~~~~~~~\tilde{X}_{h}\subseteq V((K^{(t-r+1)}\boxtimes S)\setminus P^{(t-r+1)})\cap R,\\
 & &~~~~~~~~~~~~~~~~~~~~~~~~~~~~~~~~~~~~~~~~\tilde{X}_h \text{~is~} (|Y_{h}|, r_{h})\text{-scattered in }K^{(t)}\boxtimes S,\text{ and }\\
& &~~~~~~~~~~~~~~~~~~~~~~~~~~~~~~~~~~~~~~~~K^{(t)}\boxtimes S\models \bigwedge_{x\in \tilde{X}_h}\psi_{h}(x)\}.\\
\end{eqnarray*}

Notice that $(Y_{1},\ldots,Y_{m},t)\in {{\sf sig}}_{φ,\boxtimes}(\mathfrak{K}_W, R, z,S)$ only if for every $h\in [m]$, $\tilde{X}_{h}\subseteq V(K^{(z-r+1)}\boxtimes S)$ (since, otherwise, $K^{(z)}\boxtimes S$ can not be a model of $\bigwedge_{x\in \tilde{X}_h}\psi_{h}(x)$).
Recall that $\rho =(2k+1)\cdot d$.
We also define the {\em $(φ,\boxtimes)$-characteristic} of $(\mathfrak{K}_W,R)$
as follows
\begin{eqnarray*}
\text{\sf  $(φ,\boxtimes)$-char}(\mathfrak{K}_W,R) & =& \{(z,{\sf sig},s)\in  [d,\rho]\times 2^{{\sf SIG}}\times [0,k]\mid \exists S\subseteq\boxtimes\langle K,R\rangle\mbox{~such that},\\
& &\hspace{7cm} A(S)\subseteq V(K^{(z-d+1)})\cap R,\\ 
& &\hspace{7cm}  |S|=s, K\boxtimes S \text{~is planar, and}   \\
& &\hspace{7cm}  {{\sf sig}}_{φ,\boxtimes}(\mathfrak{K},R,z,S)={\sf sig}\}.
\end{eqnarray*}

Notice that all queries in the definition of 
$(φ,\boxtimes)\text{\sf -char}(\mathfrak{K}_W, R)$ can be expressed in MSOL.
Indeed, this is easy to see when 
$\boxtimes\in \{{\sf vd}, {\sf ed}, {\sf ec}\}$, as in this case
the query ``$K\boxtimes S$ is planar'' is trivially true, since 
$V(\comp(\tilde{W}))$ is $\boxtimes$-planarization irrelevant.
In the case where $\boxtimes={\sf ea}$, MSOL expressibility follows from \autoref{kdsafksdfkalsdfkla}.

\subsection{An algorithm for finding irrelevant vertices}\label{subsec6_2}
In this subsection, we present the algorithm {\tt Find\_Vertex} of~\autoref{wedndlsakfnlkdsafnkalen}.
Throughout the rest of this section we assume that we are given a Gaifman sentence $φ$ and a $\boxtimes\in{\sf OP}$.

\paragraph{The algorithm {\bf Find\_Vertex}.}
The algorithm {\bf Find\_Vertex} receives as an input a $k\in \Bbb{N}$, a partially $\Delta$-embedded graph $G$, a set of (annotated) vertices $R\subseteq V(G)$, and a $q$-wall $\tilde{W}$ of $G$ such that 
\begin{itemize}
	\item $q = \funref{sngklargnklrangl}(k,|φ|)$,
	\item the compass of $\tilde{W}$ has treewidth at most $\funref{jsdfjbnvjfdak}(k,q)$ (where $\funref{jsdfjbnvjfdak}$ is the function of~\autoref{sdmgflsadmgfrlka}),
	\item $G\cap \Delta=\comp(\tilde{W})$, $\bd(\Delta)=\perim(\tilde{W})$,
	\item $V(\comp(\tilde{W}))$ is $\boxtimes$-planarization irrelevant, and 
\end{itemize}
The algorithm has four steps.
First, recall that any given Gaifman sentence $φ$ is a Boolean combination of sentences $φ_{1}, \ldots, φ_{m}$ such that, for every $h\in[m]$,
\[φ_{h}=\exists x_{1}\ldots\exists x_{\ell_{h}}\big( \bigwedge_{1\leq i<j\leq \ell_{h}} d(x_{i}, x_{j})> 2r_{h}\wedge \bigwedge_{i\in [\ell_{h}]}\psi_{h}(x_{i})\big),\]
where $\ell_{h},r_{h}\geq 1$ and $\psi_{h}(x)$ is $r_{h}$-local.
We consider the sentence $\tilde{φ}$ and recall that it is the same Boolean combination 
of sentences $\tilde{φ}_{1}, \ldots, \tilde{φ}_{m}$ such that, for every $h\in[m]$, 	
\begin{eqnarray*}
\tilde{φ}_{h}=\exists x_{1}\ldots\exists x_{\ell_{h}}\big(\bigwedge_{i\in [\ell_{h}]}x_{i}\in R\wedge \bigwedge_{1\leq i<j\leq \ell_{h}} d(x_{i}, x_{j})> 2r_{h}\wedge \bigwedge_{i\in [\ell_{h}]}\psi_{h}(x_{i})\big),\label{gaidsfd}
\end{eqnarray*}
and the formulas $\tilde{φ}$ and $\tilde{φ}_1, \ldots, \tilde{φ}_m$ are evaluated on annotated graphs of the form $(G,R)$.

We set $r:=\max_{h\in[m]}\{r_{h}\}$, $\ell:=\sum_{h\in[m]}\ell_{h}$,
\begin{align*}
	d & :=2\left(r + (\ell+1)r+r\right),\\
	\rho & := (2k+1)\cdot d,\\
	w & :=2^{\rho\cdot (k+1) \cdot 2^{2^\ell \cdot \rho}}\cdot (2k+1)(\ell +3),\text{ and}\\
	\funref{sngklargnklrangl}(k,|φ|) & := \lceil (2\rho+1)\cdot \sqrt{w}\rceil.
\end{align*}

\paragraph{Step 1.}
We first find a collection ${\cal W}$  of $w$-many $(2\rho+1)$-subwalls of $\tilde{W}$ whose compasses are pairwise disjoint. This collection exists because $\tilde{W}$ is a $q$-wall, where $q=  \funref{sngklargnklrangl}(k,|φ|) =\lceil (2\rho+1)\cdot \sqrt{w}\rceil$. Observe that ${\cal W}$ can be computed in linear time.

\paragraph{Step 2.}
We check whether there is a wall  $W\in {\cal W}$ such that $ V(\comp(W))\cap R=\emptyset$.
If there is such a wall $W$, we set $X:=V(\comp(W^{(\rho-1)}))$ and $v$ to be a vertex in $\cen(W)$ and our algorithm returns the vertex set $X$ and the vertex $v$.
If $V(\comp(W))\cap R\neq\emptyset$
for every $W\in {\cal W}$, we continue to Step 3.

At this point, we wish to argue about the correctness of  {\bf Step 2}.
First, note that for every $u\notin V(\comp(W^{(\rho -1)}))$ we have that $d(u,v)\geq \rho-1$.
This holds since
 $\comp(W^{(\rho -1)})$ is a planar graph
 and there exist at least $\rho-1$ layers of $W$ separating a vertex $u\notin V(\comp(W^{(\rho -1)}))$ and $v$.
Thus, given that for every $u\notin V(\comp(W^{(\rho -1)}))$ it holds that $d(u,v)\geq \rho-1>r$ and
for every $h\in[m]$, the formula $\psi_{h}(x)$ is $r_{h}$-local,
we derive that $(G,R,k)$ is a $(φ,\boxtimes)$-triple if and only if $(G\setminus v,R\setminus X, k)$ is a $(φ,\boxtimes)$-triple.
Therefore, our algorithm can safely return the vertex set $X$ and the vertex $v$.

\paragraph{Step 3.}
For every $i\in [w]$, we set $R_i= R\cap V(\comp(W_i))$ and $(\mathfrak{K}_{i}, R_i)$ be the panelled compass of $W_i$ in $G$, where $\mathfrak{K}_i: = \mathfrak{K}_{W_i}$, $K_i : = \comp(W_i)$, and  for every $j\in [\rho]$, $K_i^{(j)}: = \comp(W_i^{(2j+1)})$.
Also, for every $j\in [\rho]$, $K_i^{(j)}: = \comp(W_i^{(2j+1)})$ , we set $P_i^{(j)}: = V(\perim(W_i^{(2j+1)}))$.
Then, for every $i\in [w]$, we compute $(φ,\boxtimes)\text{\sf -char}(\mathfrak{K}_i, R_i)$. 
As all queries in the definition of $(φ,\boxtimes)\text{\sf -char}(\mathfrak{K}_W, R)$ can be expressed in MSOL and, by the hypothesis of the lemma, the compass of each $W\in {\cal W}$ has treewidth at most $\funref{jsdfjbnvjfdak}(k,q)$, it follows by the 
theorem of Courcelle that $(φ,\boxtimes)\text{\sf -char}(\mathfrak{K}_i, R_i), i\in [w]$ can be computed 
in ${\cal O}_{k,|φ|}(n)$ time.
We say that two walls $W_{i},W_{j}\in{\cal W}$ are {\em $(φ,\boxtimes)$-equivalent}  if $(φ,\boxtimes)\text{\sf -char}(\mathfrak{K}_{i}, R_i)= (φ,\boxtimes)\text{\sf -char}(\mathfrak{K}_{j}, R_j)$,  and we denote this 
by $W_{i}\sim_{φ,\boxtimes}W_{j}$.

\paragraph{Step 4.}
We find a collection ${\cal W}'\subseteq {\cal W}$ of  $(2k+1)(\ell+3)$ walls that are pairwise $(φ,\boxtimes)$-equivalent.
This can be done since  $w=2^{\rho\cdot (k+1) \cdot 2^{2^\ell \cdot \rho}}\cdot (2k+1)(\ell +3)$ and for every $i\in[w]$, $(φ,\boxtimes)\text{\sf -char}(\mathfrak{K}_{i}, R_i)\subseteq [d+1, \rho]\times 2^{\sf SIG}\times [0,k]$.
Observe that ${\cal W}'$ can be computed in time ${\cal O}_{k,|φ|}(n)$.
We fix a wall $W_{1}\in {\cal W'}$, and set $X:=V(K_{1}^{(r)})$.
Our algorithm returns $X$ and a vertex $v\in\cen(W_{1}^{(r)})$.

\subsection{Proof of correctness of the algorithm}\label{subsec6_3}
To complete the proof of \autoref{wedndlsakfnlkdsafnkalen}, we have to prove that $(G,R,k)$ is a $(φ,\boxtimes)$-triple if and only if $(G\setminus v,R\setminus X, k)$ is a $(φ,\boxtimes)$-triple.

Let $R^{\prime}:=R\setminus X$.
We now prove that the following holds:

\medskip
	
\noindent{\em Claim:} 
If  $S$ is a subset of $\boxtimes\langle G,R\rangle$, where $|S|=k$ and $G\boxtimes S$ is a planar graph, then there exists a set $S'\subseteq \boxtimes\langle G,R'\rangle$ such that 
\begin{itemize}
\item $|S|=|S'|$,
\item $G\boxtimes S'$ is a planar graph, and
\item  $(G\boxtimes S,R)\models \tilde{φ}$ if and only if $(G\boxtimes S^{\prime},R')\models \tilde{φ}$.
\end{itemize}

\medskip

\noindent{\em Proof of Claim:}
Let $S$ be a subset of $\boxtimes\langle G,R\rangle$, where $|S|= k$ and $G\boxtimes S$ is a planar graph.

\paragraph{Finding an equivalent panelled compass that is disjoint from $S$.}
Since the collection ${\cal W}'$ of walls that are $(φ,\boxtimes)$-equivalent with $W_{1}$ has size $(2k+1)(\ell +3)$ and $|A(S)|\leq 2k$, there exists a collection ${\cal W}''\subseteq {\cal W}'\setminus\{W_{1}\}$ of size $(\ell + 2)$, such that for every $\hat{W}\in {\cal W}''$, it holds that $\hat{W}\sim_{φ, \boxtimes} W_{1}$ and $V(\comp(\hat{W}))\cap A(S)=\emptyset$.
Let $W_{2}\in {\cal W}''$.

\paragraph{Every solution $S$ leaves an intact buffer in $W_1$.}
Since $W_1$ has height $2\rho+1$, where $\rho=(2k+1)\cdot d$,
observe that there is a collection of $2k+1$ closed annuli $\{\ann({\cal A}_{i\cdot d}^{(d)}(W_{1}))\mid i\in [2k+1]\}$ that are pairwise disjoint and keep in mind that  each $\ann({\cal A}_{i\cdot d}^{(d)}(W_{1}))$ is a closed annulus that is a subset of $\Delta$ and, intuitively, ``crops'' an area of $d$ consecutive layers of $W_1$.
Therefore, the fact that $|A(S)|\leq 2k$ implies that there exists an $i\in [2k+1]$ such that $A(S)$ does not intersect $\ann({\cal A}_{i\cdot d}^{(d)}(W_{1}))$.
Notice that, since $G\boxtimes S$ is planar and $d\geq 3$, $S$ is partitioned into the sets $S_{\sf in}$ and $S_{\sf out}$, where $A(S_{\sf in})\subseteq V(K_{1}^{(i\cdot d-d+1)})\cap R$ and $A(S_{\sf out})\cap V(K_{1}^{(i\cdot d)})=\emptyset$. We set $z := i\cdot d$.

\paragraph{Finding a substitute for $S_{\sf in}$ in the compass of $W_2$.}
Since $S_{\sf in}\subseteq \boxtimes\langle K_1,R_1\rangle$, $A(S_{\sf in})\subseteq V(K_{1}^{(z-d+1)})\cap R = V(K_{1}^{(z-d+1)})\cap R_1$, and $K_{1}\boxtimes S$ is planar, the fact that
$W_{1}\sim_{φ, \boxtimes} W_{2}$
 implies that there exists a set $\tilde{S}\subseteq \boxtimes\langle K_2,R_2\rangle$,
such that $|\tilde{S}|=|S_{\sf in}|$, $A(\tilde{S})\subseteq V(K_{2}^{(z-d+1)})\cap R_2$,  
$K_{2}\boxtimes \tilde{S}$ is planar, and
${{\sf sig}}_{φ,\boxtimes}(\mathfrak{K}_2, R_2, z,\tilde{S})={{\sf sig}}_{φ,\boxtimes}(\mathfrak{K}_1, R_1, z,S_{\sf in})$.
We set
\[S^{\prime}:=\tilde{S}\cup S_{\sf out}.\]

\paragraph{Planarity is preserved by replacing $S$ with $S'$.}
Notice that $S^{\prime}\subseteq \boxtimes\langle G,R'\rangle$, $|S^{\prime}|=k$, and $G\boxtimes S^{\prime}$ is planar. As a proof of the latter, in the case where $\boxtimes={\sf vd}, {\sf ed}$, or  ${\sf ec}$, since $V(\comp(\tilde{W}))$ is $\boxtimes$-planarization irrelevant, every inclusion-minimal $\boxtimes$-planarizer of $G$ is a subset of $S_{\sf out}$. Also, in the case where $\boxtimes={\sf ea}$, $G\boxtimes S^{\prime}$ is planar since $(G\setminus V(K_{2}^{(z-d+1)}))\boxtimes S_{\sf out}$ and  $K_{2}^{(z)}\boxtimes \tilde{S}$ are planar and $d\geq 3$ (due to  \autoref{dnsfkanklda} presented in \autoref{dnksldngklfdnglkad}).
Therefore, our goal now is to prove that $(G\boxtimes S, R)\models \tilde{φ}$ if and only if $(G\boxtimes S^{\prime},R')\models \tilde{φ}$.

\paragraph{Satisfiability of $\tilde{φ}$ is preserved by replacing $S$ with $S'$.}
Since $(G\boxtimes S,R)\models  \tilde{φ}$ and $ \tilde{φ}$ is a Boolean combination of the formulas $ \tilde{φ}_{1}, \ldots,  \tilde{φ}_{m}$, there is a set $J\subseteq [m]$ such that for every $j\in J$ it holds $(G\boxtimes S,R)\models  \tilde{φ}_{j}$ and for every $j\notin J$ it holds that $(G\boxtimes S,R)\models \neg  \tilde{φ}_{j}$.
In order to show that $(G\boxtimes S^{\prime},R')\models \tilde{φ}$, we show that  for every $j\in J$ it holds $(G\boxtimes S',R')\models  \tilde{φ}_{j}$ and for every $j\notin J$ it holds that $(G\boxtimes S',R')\models \neg \tilde{φ}_{j}$. Therefore, we distinguish two cases.
\bigskip

\noindent{\bf Case 1:} $j\in J$.
We aim to prove that $(G\boxtimes S,R)\models  \tilde{φ}_{j} \iff (G\boxtimes S',R')\models \tilde{φ}_{j}$.
In other words, we will prove that there exists an $(\ell_{j}, r_{j})$-scattered set $X_{j}\subseteq R$ in $G\boxtimes S$ such that $G\boxtimes S\models \bigwedge_{x\in X_{j}}\psi_{j}(x)$ if and only if there is an $(\ell_{j}, r_{j})$-scattered set $X\subseteq R^{\prime}$ in $G\boxtimes S^{\prime}$ such that $G\boxtimes S^{\prime}\models \bigwedge_{x\in X}\psi_{j}(x)$.
Let $X_{j}\subseteq R$ be an $(\ell_{j}, r_{j})$-scattered set in $G\boxtimes S$ such that $G\boxtimes S\models \bigwedge_{x\in X_{j}}\psi_{j}(x)$.
Recall that $S^{\prime}:=\tilde{S}\cup S_{\sf out}$, where $A(\tilde{S})\subseteq V(K_{2}^{(z-d+1)})\cap R_2$ and $A(S_{\sf out})\cap V(K_{2})=\emptyset$.
We prove the following, which intuitively states that, given the set $X_{j}$, we can find an other set $X_j '$ that ``behaves'' in the same way as $X_j$ but also ``avoids'' some inner part of $K_2$.
\medskip

\noindent{{\em Subclaim:}} There exists a $t\in [z-\frac{d}{2}+2r+1, z - r]$ and an  $(\ell_{j}, r_{j})$-scattered set ${X}^{\prime}_{j}\subseteq R$ in $G\boxtimes S$ such that $G\boxtimes S\models \bigwedge_{x\in X_{j}}\psi_{j}(x) \iff G\boxtimes S\models \bigwedge_{x\in {X}^{\prime}_{j}}\psi_{j}(x)$ and ${X}^{\prime}_{j}\cap V(K_{2}^{(t)}) = \emptyset$.

\medskip

\noindent{\em Proof of Subclaim:}	
Recall that there is a collection ${\cal W}''$ of size $(\ell +2)$ of walls $(φ, \boxtimes)$-equivalent to $W_{1}$ whose compasses are disjoint from $A(S)$.
Therefore, since $X_j$ has size at most $\ell$, there exists a wall $W_{3}\in {\cal W}''\setminus\{ W_{2}\}$ such that $V(K_{3})\cap (A(S)\cup  X_{j})=\emptyset$.

We now focus on the closed annulus $\ann({\cal A}_{z}^{(d)}(W_2))$, which, since $A(S)\cap V(K_2) = \emptyset$, does not intersect $A(S)$.
We have that
$d=2(r+(\ell+1)r+r)$ and $|X_{j}|\leq  \ell$ and therefore there exists a $t \in [z-\frac{d}{2}+2r+1, z- r]$ (see \autoref{fsjnaklfnsaklfn})
such that $X_{j}$ does not intersect $\ann({\cal A}_t^{(r)}(W_{2}))$.
Intuitively, we separate the $d$ layers of $W_2$ that are in $\ann({\cal A}_{z}^{(d)}(W_2))$ into two parts, the first $d/2$ layers and the second $d/2$ layers, and then we find some layer among the ``central'' $(\ell +1)r$ layers of the second part ($t$ corresponds to a layer in the yellow area of \autoref{fsjnaklfnsaklfn}). This layer (corresponding to $t$) together with its preceding $r-1$ layers define an annulus of size $r$, $\ann({\cal A}_t^{(r)}(W_{2}))$, which  $X_{j}$ ``avoids''.
Since $\ann({\cal A}_t^{(r)}(W_{2}))\subseteq \ann({\cal A}_{z}^{(d)}(W_2))$ and $A(S)\cap\ \ann({\cal A}_{z}^{(d)}(W_2)) = \emptyset$, it also holds that $A(S)\cap \ann({\cal A}_t^{(r)}(W_{2}))=\emptyset$).

\begin{figure}[ht]
	\centering
\resizebox{0.7\textwidth}{!}{
\begin{tikzpicture}[ipe stylesheet]
  \fill[orange, ipe opacity 30]
    (40, 600) rectangle (500, 544);
  \fill[gold, ipe opacity 30]
    (40, 724) rectangle (500, 668);
  \fill[purple, ipe opacity 30]
    (40, 756) rectangle (500, 728);
  \draw
    (40, 568)
     -- (500, 568);
  \draw
    (500, 564)
     -- (40, 564);
  \draw
    (40, 560)
     -- (500, 560);
  \draw
    (500, 576)
     -- (39.999, 576.137);
  \draw
    (40.002, 580.137)
     -- (500, 580);
  \draw
    (500, 584)
     -- (40.005, 584.137);
  \draw
    (40, 600)
     -- (500, 600);
  \draw
    (40, 596)
     -- (500, 596);
  \draw
    (500, 592)
     -- (40, 592);
  \draw
    (40, 588)
     -- (500, 588);
  \draw
    (500, 608)
     -- (39.999, 608.137);
  \draw
    (40.002, 612.137)
     -- (500, 612);
  \draw
    (500, 616)
     -- (40.005, 616.137);
  \draw[ipe pen fat]
    (40, 632)
     -- (500, 632);
  \draw
    (40, 628)
     -- (500, 628);
  \draw
    (500, 624)
     -- (40, 624);
  \draw
    (40, 620)
     -- (500, 620);
  \draw
    (500, 640)
     -- (39.999, 640.137);
  \draw
    (40.002, 644.137)
     -- (500, 644);
  \draw
    (500, 648)
     -- (40.005, 648.137);
  \draw
    (40, 664)
     -- (500, 664);
  \draw
    (40, 660)
     -- (500, 660);
  \draw
    (500, 656)
     -- (40, 656);
  \draw
    (40, 652)
     -- (500, 652);
  \draw
    (500, 672)
     -- (39.999, 672.137);
  \draw
    (40.002, 676.137)
     -- (500, 676);
  \draw
    (500, 680)
     -- (40.005, 680.137);
  \draw[ipe pen fat]
    (40, 696)
     -- (500, 696);
  \draw
    (40, 692)
     -- (500, 692);
  \draw
    (500, 688)
     -- (40, 688);
  \draw
    (40, 684)
     -- (500, 684);
  \draw
    (500, 700)
     -- (39.999, 700.137);
  \draw
    (40.002, 704.137)
     -- (500, 704);
  \draw
    (500, 708)
     -- (40.005, 708.137);
  \draw
    (40, 728)
     -- (500, 728);
  \draw
    (40, 720)
     -- (500, 720);
  \draw
    (500, 716)
     -- (40, 716);
  \draw
    (40, 712)
     -- (500, 712);
  \draw
    (500, 732)
     -- (39.999, 732.137);
  \draw
    (40.002, 736.137)
     -- (500, 736);
  \draw
    (500, 740)
     -- (40.005, 740.137);
  \draw[ipe pen fat]
    (40, 756)
     -- (500, 756);
  \draw
    (40, 752)
     -- (500, 752);
  \draw
    (500, 748)
     -- (40, 748);
  \draw
    (40, 744)
     -- (500, 744);
  \draw[ipe pen fat]
    (40, 544)
     -- (500, 544);
  \draw
    (40, 536)
     -- (500, 536);
  \draw
    (500, 532)
     -- (40, 532);
  \draw
    (40, 528)
     -- (500, 528);
  \draw
    (500, 548)
     -- (39.999, 548.137);
  \draw
    (40.002, 552.137)
     -- (500, 552);
  \draw
    (500, 556)
     -- (40.005, 556.137);
  \draw[ipe pen fat]
    (40, 512)
     -- (500, 512);
  \draw
    (500, 516)
     -- (39.999, 516.137);
  \draw
    (40.002, 520.137)
     -- (500, 520);
  \draw
    (500, 524)
     -- (40.005, 524.137);
  \draw[ipe pen fat]
    (40, 600)
     -- (500, 600);
  \draw
    (40, 596)
     -- (500, 596);
  \draw
    (500, 592)
     -- (40, 592);
  \draw
    (40, 588)
     -- (500, 588);
  \draw
    (500, 608)
     -- (39.999, 608.137);
  \draw
    (40.002, 612.137)
     -- (500, 612);
  \draw
    (500, 616)
     -- (40.005, 616.137);
  \node[ipe node, font=\LARGE]
     at (512, 752) {$P_2^{(z)}$};
  \node[ipe node, font=\LARGE]
     at (512, 720) {$P_2^{(z-r)}$};
  \node[ipe node, font=\LARGE]
     at (512, 692) {$P_2^{(t)}$};
  \node[ipe node, font=\LARGE]
     at (512, 664) {$P_2^{(z-\frac{d}{2}+r+1)}$};
  \node[ipe node, font=\LARGE]
     at (512, 628) {$P_2^{(z-\frac{d}{2})}$};
  \node[ipe node, font=\LARGE]
     at (512, 600) {$P_2^{(z-\frac{d}{2}-r)}$};
  \node[ipe node, font=\LARGE]
     at (512, 568) {$P_2^{(t')}$};
  \node[ipe node, font=\LARGE]
     at (512, 540) {$P_2^{(z-d+r+1)}$};
  \node[ipe node, font=\LARGE]
     at (512, 508) {$P_2^{(z-d+1)}$};
  \fill[purple, ipe opacity 30]
    (40, 664) rectangle (500, 636);
  \fill[green, ipe opacity 30]
    (40, 632) rectangle (500, 604);
  \fill[green, ipe opacity 30]
    (40, 540) rectangle (500, 512);
  \draw[ipe pen fat]
    (40, 724)
     -- (500, 724);
  \draw[ipe pen fat]
    (40, 668)
     -- (500, 668);
  \draw
    (40, 636)
     -- (500, 636);
  \draw
    (40, 604)
     -- (500, 604);
  \draw[ipe pen fat]
    (40, 572)
     -- (500, 572);
  \draw
    (40, 540)
     -- (500, 540);
\end{tikzpicture}
}
\caption{Visualization of the layers of $W_{2}$ that are subsets of $ \ann(\mathcal{A}_{z}^{(d)}(W_2))$.
For every $h\in[\rho]$, $P_{h}:=\perim(W^{(2h+1)})$.
The color-shadowed areas follow the colors in \autoref{ndajkgndskngdasklgnasd}.}
\label{fsjnaklfnsaklfn}
\end{figure}
We set $X_{j}^\star:=X_{j}\cap V(K_{2}^{(t -r+1)})$ and $Y_{j}\subseteq [\ell_j]$  to be the set of indices of the vertices in $X_{j}^{\star}$.
Notice that $X_{j}^\star\subseteq R\cap V(K_{2}^{(t -r+1)})\subseteq R_2$ and that, since $X_{j}$ does not intersect $\ann({\cal A}_t^{(r)}(W_{2}))$, also $X_{j}^\star$ does not intersect $P_2^{(t-r+1)}$ (that is an extremal cycle of $\ann({\cal A}_t^{(r)}(W_{2}))$).
Also, observe that, since
$A(S)\cap V(K_2)=\emptyset$, we have that $V(K_2)\subseteq V(G\boxtimes S)$ and $G\boxtimes S [V(K_2)]  = K_2$.
Therefore, since
$X_{j}^{\star}\subseteq V(K_{2}^{(t -r+1)}\setminus P_2^{(t-r+1)})$,
$\psi_j (x)$ is an $r_j$-local formula, 
and $r\geq r_j$, we have that $G\boxtimes S\models \bigwedge_{x\in X_{j}^\star}\psi_{j}(x)
\iff K_2^{(t)}\models \bigwedge_{x\in X_{j}^{\star}}\psi_{j}(x)$.
To sum up, we have that the set $X_{j}^{\star}$ is a subset of $V(K_{2}^{(t -r+1)}\setminus P_2^{(t-r+1)})\cap R_2$ that is $(|Y_j|,r_j)$-scattered in $K_2^{(t)}$ (being a subset of $X_j$)  and $K_2^{(t)}\models \bigwedge_{x\in X_{j}^{\star}}\psi_{j}(x)$.

Notice that, since $A(S)\cap V(K_2) = \emptyset$ and $W_{2}\sim_{φ,\boxtimes} W_{3}$, we have that ${{\sf  sig}}_{φ,\boxtimes}(\mathfrak{K}_2, R_2, t',\emptyset)={{\sf  sig}}_{φ,\boxtimes}(\mathfrak{K}_3, R_3, t',\emptyset)$, for every $t'\in [\rho]$. 
Therefore, we have that
${{\sf  sig}}_{φ,\boxtimes}(\mathfrak{K}_2, R_2, t,\emptyset)={{\sf  sig}}_{φ,\boxtimes}(\mathfrak{K}_3, R_3, t,\emptyset)$ and this implies that there is a set $\tilde{X_{j}}\subseteq V(K_{3}^{(t -r+1)}\setminus P_3^{(t-r+1)})\cap R_3$ such that
$\tilde{X_{j}}$ is $(|Y_{j}|, r_{j})$-scattered in $K_{3}^{(t)}$ and $K_2^{(t)}\models \bigwedge_{x\in X_{j}^{\star}}\psi_{j}(x)\iff K_3^{(t)}\models \bigwedge_{x\in \tilde{X}_{j}}\psi_{j}(x)$.
Observe that since $A(S)\cap V(K_{3})=\emptyset$ and $\tilde{X_{j}}\subseteq V(K_{3}^{(t-r+1)}\setminus P_3^{(t-r+1)})\subseteq V(K_3^{(\rho -r)})$, for every $x\in\tilde{X}_{j}$ it holds that $N_{G\boxtimes S}^{(\leq r)}(x)\cap A(S)=\emptyset$.
Thus, since every $\psi_{h}(x), h\in[m]$ is $r_{h}$-local, it follows that $K_3^{(t)}\models \bigwedge_{x\in \tilde{X}_{j}}\psi_{j}(x)\iff G\boxtimes S\models \bigwedge_{x\in \tilde{X}_{j}}\psi_{j}(x)$.

We now consider the set
\[{X}^{\prime}_{j}:=\left( X_{j}\setminus X_{j}^{\star}\right)\cup \tilde{X}_{j}.\]
Since $V(K_{3})\cap (X_{j} \cup A(S))=\emptyset$ and $r\geq r_j$, for every $x\in X_{j}$, and thus for every $x\in X_{j}\setminus X_{j}^{\star}$, it holds that $N_{G\boxtimes S}^{(\leq r_{j})}(x)\cap V(K_{3}^{(\rho-r+1)})=\emptyset$.
Also, since $t \leq \rho -r$ and $\tilde{X_{j}}\subseteq V(K_{3}^{(t -r+1)}\setminus P_3^{(t-r+1)})$, for every $x\in\tilde{X}_{j}$ it holds that $N_{G\boxtimes S}^{(\leq r_{j})}(x) \subseteq V(K_{3}^{(\rho-r+1)})$.
Thus, for every $x\in X_{j}\setminus X_{j}^{\star}$ and $x'\in \tilde{X}_{j}$ we have that
$N_{G\boxtimes S}^{(\leq r_{j})}(x)\cap N_{G\boxtimes S}^{(\leq r_{j})}(x')=\emptyset$.
The latter, together with the fact that the set $X_{j}\setminus X_{j}^{\star}$ is $(\ell_{j}-|Y_{j}|, r_{j})$-scattered in $G\boxtimes S$, $\tilde{X_{j}}$ is $(|Y_{j}|, r_{j})$-scattered in $K_{3}^{(t)}$, and $K_{3}^{(t)} = G\boxtimes S[V(K_{3}^{(t)})]$, implies that ${X}^{\prime}_{j}$ is an $(\ell_{j}, r_{j})$-scattered set in $G\boxtimes S$.
Moreover, by definition, we have that ${X}^{\prime}_{j}\subseteq R$ and ${X}^{\prime}_{j}$ does not intersect $V(K_{2}^{(t)})$, while we already argued why $G\boxtimes S\models \bigwedge_{x\in {X}_{j}}\psi_{j}(x) \iff G\boxtimes S\models \bigwedge_{x\in {X}^{\prime}_{j}}\psi_{j}(x)$. Subclaim follows.
\bigskip

Following the above subclaim, let a $t\in [z-\frac{d}{2}+2r+1, z - r]$ and an  $(\ell_{j}, r_{j})$-scattered set ${X}^{\prime}_{j}\subseteq R$ in $G\boxtimes S$ 
such that $G\boxtimes S\models \bigwedge_{x\in X_{j}}\psi_{j}(x) \iff G\boxtimes S\models \bigwedge_{x\in {X}^{\prime}_{j}}\psi_{j}(x)$ and ${X}^{\prime}_{j}\cap V(K_{2}^{(t)}) = \emptyset$.

Since $d=2(r+(\ell +1)r +r)$ and $|{X}^{\prime}_{j}|\leq \ell$, there exists a $t'\in [z -d+2r+1, z- \frac{d}{2}-r]$ such that  ${X}^{\prime}_{j}$  does not intersect $\ann({\cal A}_{t'}^{(r)}(W_{1}))$ ($t'$ corresponds to a layer in the orange area in \autoref{fsjnaklfnsaklfn}).

Now, consider the set $Z:={X}^{\prime}_{j}\cap V(K_{1}^{(t'-r+1)}\boxtimes S_{\sf in})$.
Observe that $Z\subseteq R_1$ and therefore $Z\subseteq V(K_{1}^{(t'-r+1)}\boxtimes S_{\sf in})\cap R_1$.
Also, notice that, since $A(S_{\sf in})\subseteq V(K_1^{z-d+1})$ and $t'\geq z -d+2r+1$, $P_1^{(t'-r+1)}\subseteq V(K_{1}^{(t'-r+1)}\boxtimes S_{\sf in})$.
Thus, $Z\subseteq V((K_{1}^{(t'-r+1)}\boxtimes S_{\sf in})\setminus P_1^{(t'-r+1)})\cap R_1$.
Recall that $R'=R\setminus V(K_{1}^{(r)})$ and observe that, since $({X}^{\prime}_{j}\setminus Z)\cap V(K_{1}^{(t')}) = \emptyset$ and $t'>r$, it holds that ${X}^{\prime}_{j}\setminus Z\subseteq R^{\prime}$.	
Let $Y_{j}^{\prime}\subseteq [\ell_j]$ be the set of the indices of the vertices of ${X}^{\prime}_{j}$ in $Z$.
Also, notice that since $Z\subseteq V((K_{1}^{(t'-r+1)}\boxtimes S_{\sf in})\setminus P_1^{(t'-r+1)})\cap R_1$, $Z$ is a subset of ${X}^{\prime}_{j}$, and ${X}^{\prime}_{j}$ is
$(\ell_{j}, r_{j})$-scattered in  $G\boxtimes S$, it holds that $Z$ is $(|Y_{j}^{\prime}|,r_j)$-scattered in $K_1^{(t')}\boxtimes S_{\sf in}$ and $K_1^{(t')}\boxtimes S_{\sf in}\models\bigwedge_{x\in {Z}}\psi_{j}(x)$.
As we mentioned before, ${{\sf  sig}}_{φ,\boxtimes}(\mathfrak{K}_2, R_2, z,\tilde{S})={{\sf  sig}}_{φ,\boxtimes}(\mathfrak{K}_1, R_1, z,S_{\sf in})$.
This implies the existence of a set $\tilde{Z}\subseteq V((K_{2}^{(t'-r+1)}\boxtimes \tilde{S})\setminus P_2^{(t'-r+1)})\cap R_2\subseteq R^{\prime}$
such that $\tilde{Z}$ is $(|Y_{j}^{\prime}|,r_{j})$-scattered in $K_{2}^{(t')}\boxtimes \tilde{S}$ and $K_1^{(t')}\boxtimes S_{\sf in}\models\bigwedge_{x\in {Z}}\psi_{j}(x)\iff K_2^{(t')}\boxtimes \tilde{S}\models\bigwedge_{x\in \tilde{Z}}\psi_{j}(x)$.
At this point, observe that, since the formula $\psi_{j}(x)$ is $r_{j}$-local and $Z\subseteq V((K_{1}^{(t'-r+1)}\boxtimes S_{\sf in})\setminus P_1^{(t'-r+1)})$, $N_{G\boxtimes S}^{(\leq r_j)} (x)\subseteq V(K_1^{(t')}\boxtimes S_{\sf in})$, for every $x\in Z$.
Also, $A(S_{\sf out})\cap V(K_1^{(z)}) = \emptyset$, which implies that $K_1^{(t')}\boxtimes S_{\sf in}\models\bigwedge_{x\in {Z}}\psi_{j}(x)\iff G\boxtimes S\models\bigwedge_{x\in {Z}}\psi_{j}(x)$.
Thus, $K_2^{(t')}\boxtimes \tilde{S}\models\bigwedge_{x\in \tilde{Z}}\psi_{j}(x)\iff G\boxtimes S\models\bigwedge_{x\in {Z}}\psi_{j}(x)$.

Also, since $S^{\prime}=\tilde{S}\cup S_{\sf out}$, where $A(S_{\sf out})\cap V(K_{2})=\emptyset$, and $\tilde{Z}$ is $(|Y_{j}^{\prime}|,r_{j})$-scattered in
$K_{2}^{(t')}\boxtimes \tilde{S}$, where $\tilde{Z}\subseteq V((K_{2}^{(t'-r+1)}\boxtimes \tilde{S})\setminus P_2^{(t'-r+1)})$ and $t'\leq \rho - r$, we notice that $\tilde{Z}$ is also $(|Y_{j}^{\prime}|,r_{j})$-scattered in $G\boxtimes S^{\prime}$.
Moreover, the formula $\psi_{j}(x)$ is $r_{j}$-local, so $K_2^{(t')}\boxtimes \tilde{S}\models\bigwedge_{x\in \tilde{Z}}\psi_{j}(x)\iff G\boxtimes S^{\prime}\models\bigwedge_{x\in \tilde{Z}}\psi_{j}(x)$.
Therefore, we have $G\boxtimes S\models\bigwedge_{x\in {Z}}\psi_{j}(x)\iff G\boxtimes S^{\prime}\models\bigwedge_{x\in \tilde{Z}}\psi_{j}(x)$.

Consider the set $$X:=({X}^{\prime}_{j}\setminus Z)\cup \tilde{Z}.$$

Notice that since ${X}^{\prime}_{j}\setminus Z$ is an $(\ell_{j}-|Y_{j}^{\prime}|, r_{j})$-scattered set in $G\boxtimes S$
and it does not intersect neither
$V(K_{2}^{(z-d+1)})$ (where $A(\tilde{S})$ lies),
nor $V(K_{1}^{(t'-r+1)})$ (where $A(S_{\sf in})$ lies),
it is also an $(\ell_{j}-|Y_{j}^{\prime}|, r_{j})$-scattered set in $G\boxtimes S'$.
Since $\tilde{Z}\subseteq V((K_{2}^{(t'-r+1)}\boxtimes \tilde{S})\setminus P_2^{(t' -r+1)})$,
${X}^{\prime}_{j}\cap  V(K_{2}^{(t)})=\emptyset$,
and $t'< t-2r$,
for every $x\in {X}^{\prime}_{j}\setminus Z$ and $x'\in \tilde{Z}$ it holds that
$N_{G\boxtimes S^{\prime}}^{(\leq r_{j})}(x)\cap N_{G\boxtimes S^{\prime}}^{(\leq r_{j})}(x')=\emptyset$.
The latter, together with the fact that ${X}^{\prime}_{j}\setminus Z$ is an $(\ell_{j}-|Y_{j}^{\prime}|,r_{j})$-scattered set in $G\boxtimes S'$ and $\tilde{Z}$ is $(|Y_{j}^{\prime}|,r_{j})$-scattered in $G\boxtimes S^{\prime}$,
implies that  $X\subseteq R^{\prime}$ is an $(\ell_{h}, r_{h})$-scattered set in $G\boxtimes S^{\prime}$.
Furthermore, since the formula $\psi_{j}(x)$ is $r_{j}$-local, we obtain $G\boxtimes S^{\prime}\models\bigwedge_{x\in X_{j}}\psi_{j}(x)\iff G\boxtimes S^{\prime}\models\bigwedge_{x\in X}\psi_{j}(x)$.

Thus, assuming that there is an $(\ell_{j}, r_{j})$-scattered set $X_{j}\subseteq R$ in $G\boxtimes S$ such that $G\boxtimes S\models \bigwedge_{x\in X_{j}}\psi_{j}(x)$, we proved that there is an $(\ell_{j}, r_{j})$-scattered set $X\subseteq R^{\prime}$ in $G\boxtimes S^{\prime}$ such that $G\boxtimes S^{\prime}\models \bigwedge_{x\in X}\psi_{j}(x)$.
To conclude Case 1, notice that we can prove the inverse implication analogously.
That is, by assuming the existence of an $(\ell_{j}, r_{j})$-scattered set $X_{j}\subseteq R'$ in $G\boxtimes S'$ such that $G\boxtimes S'\models \bigwedge_{x\in X_{j}}\psi_{j}(x)$ and using the same arguments as above (replacing $W_{1}$ with $W_{2}$, $S$ with $S'$ and $R$ with $R'$), we can prove the existence of an $(\ell_{j}, r_{j})$-scattered set $X\subseteq R$ in $G\boxtimes S$ such that $G\boxtimes S\models \bigwedge_{x\in X}\psi_{j}(x)$.

\bigskip

\noindent{\bf Case 2:} $j\notin J$.
We aim to prove that $(G\boxtimes S,R)\models\neg\tilde{φ}_{j} \iff (G\boxtimes S',R')\models\neg \tilde{φ}_{j}$.

In other words, we need to prove that  for every $(\ell_{j}, r_{j})$-scattered set $X_{j}\subseteq R$ in $G\boxtimes S$, $G\boxtimes S\models \neg \psi_{j}(x)$, for some $x\in X_{j}$ if and only if for every $(\ell_{j}, r_{j})$-scattered set $X_{j}\subseteq R'$ in $G\boxtimes S'$, $G\boxtimes S'\models \neg \psi_{j}(x)$, for some $x\in X_{j}$.
In Case 1 we argued that there is an $(\ell_{j}, r_{j})$-scattered set $X_{j}\subseteq R$ in $G\boxtimes S$ such that $G\boxtimes S\models \bigwedge_{x\in X_{j}}\psi_{j}(x)$
if and only if there is an $(\ell_{j}, r_{j})$-scattered set $X_{j}\subseteq R'$ in $G\boxtimes S'$ such that $G\boxtimes S'\models \bigwedge_{x\in X_{j}}\psi_{j}(x)$.
This directly implies that $(G\boxtimes S,R)\models\neg\tilde{φ}_{j} \iff (G\boxtimes S',R')\models\neg \tilde{φ}_{j}$.
This concludes Case 2 and completes the proof of our claim.
\bigskip

We conclude the proof of the lemma by proving that $(G,R,k)$ is a $(φ,\boxtimes)$-triple if and only if $(G\setminus v,R', k)$ is a $(φ,\boxtimes)$-triple. As a proof of the latter, notice that by the above claim, we get that $(G,R,k)$ is a $(φ,\boxtimes)$-triple if and only if $(G,R', k)$ is a $(φ,\boxtimes)$-triple.
By the definition of the $(φ, \boxtimes)$-triple, $(G,R', k)$ is a $(φ,\boxtimes)$-triple if and only if there exists an $S\subseteq \boxtimes\langle G,R'\rangle$ such that $|S|=k$, $G\boxtimes S$ is a planar graph, and $(G\boxtimes S,R')\models \tilde{φ}$.
Since for every $h\in[m]$ the FOL-formula $\psi_{h}(x)$ is $r_{h}$-local, then the validity of $\psi_{h}(x)$ does not depend on the central vertex $v$ of $W_{1}^{(r)}$.
Therefore, $(G,R', k)$ is a $(φ,\boxtimes)$-triple if and only if $(G\setminus v,R', k)$ is a $(φ,\boxtimes)$-triple.

\end{document}